\documentclass[a4paper,onecolumn,11pt,accepted=2025-01-22]{quantumarticle}
\pdfoutput=1
\usepackage[utf8]{inputenc}
\usepackage[english]{babel}
\usepackage[T1]{fontenc}
\usepackage{amsmath}
\usepackage{hyperref}

\usepackage{tikz}
\usepackage{lipsum}

\usepackage{bm}% bold math
\usepackage{color, soul}
\usepackage{appendix}
\usepackage{blkarray}
\usepackage{multirow}
\usepackage{float}
\usepackage{mathtools}
\usepackage{textcomp}
\usepackage{subcaption}
\usepackage{amsthm}
\usepackage{dsfont}
\theoremstyle{definition}
\newtheorem{theorem}{Theorem}
\newtheorem{lemma}{Lemma}
\newcounter{propertycounter}[theorem]
\newcounter{definitioncounter}[theorem]

\newcounter{informalthmcounter}[theorem]
\newcounter{remarkcounter}[theorem]
\newcounter{sublemmacounter}[theorem]
\newcounter{extralemmacounter}[theorem]
\newtheorem{property}[propertycounter]{Property}
\newtheorem{definition}[definitioncounter]{Definition}

\newtheorem{informalthm}[informalthmcounter]{Theorem}
\newtheorem{remark}[remarkcounter]{Remark}
\newtheorem{sublemma}[sublemmacounter]{Sublemma}
\newtheorem{extralemma}[extralemmacounter]{Lemma}

\usepackage{xcolor}

\renewcommand{\thesection}{\arabic{section}}
\renewcommand{\thesubsection}{\thesection.\arabic{subsection}}
\renewcommand{\thesubsubsection}{\thesubsection.\arabic{subsubsection}}

\makeatletter
\renewcommand{\p@subsection}{}
\renewcommand{\p@subsubsection}{}
\makeatother

\usepackage[numbers, sort&compress]{natbib}
\usepackage{amssymb}
\usepackage[export]{adjustbox}
\usepackage{physics}

\begin{document}

\title{Fundamental charges for dual-unitary circuits}

\author{Tom Holden-Dye}
\email[]{thomas.holden-dye.21@ucl.ac.uk }
\affiliation{Department of Physics and Astronomy, University College London, United Kingdom}
\author{Lluís Masanes}
\affiliation{Department of Computer Science, University College London, United Kingdom}
\affiliation{London Centre for Nanotechnology, University College London, United Kingdom}
\author{Arijeet Pal}
\affiliation{Department of Physics and Astronomy, University College London, United Kingdom}

\maketitle

\begin{abstract}
    Dual-unitary quantum circuits have recently attracted attention as an analytically tractable model of many-body quantum dynamics. Consisting of a 1+1D lattice of 2-qudit gates arranged in a `brickwork’ pattern, these models are defined by the constraint that each gate must remain unitary under swapping the roles of space and time. This dual-unitarity restricts the dynamics of local operators in these circuits: the support of any such operator must grow at the effective speed of light of the system, along one or both of the edges of a causal light cone set by the geometry of the circuit. Using this property, it is shown here that for 1+1D dual-unitary circuits the set of width-$w$ conserved densities (constructed from operators supported over $w$ consecutive sites) is in one-to-one correspondence with the set of width-$w$ solitons - operators which, up to a multiplicative phase, are simply spatially translated at the effective speed of light by the dual-unitary dynamics. A number of ways to construct these many-body solitons (explicitly in the case where the local Hilbert space dimension $d=2$) are then demonstrated: firstly, via a simple construction involving products of smaller, constituent solitons; and secondly, via a construction which cannot be understood as simply in terms of products of smaller solitons, but which does have a neat interpretation in terms of products of fermions under a Jordan-Wigner transformation. This provides partial progress towards a characterisation of the microscopic structure of complex many-body solitons (in dual-unitary circuits on qubits), whilst also establishing a link between fermionic models and dual-unitary circuits, advancing our understanding of what kinds of physics can be explored in this framework.
\end{abstract}

\tableofcontents

\section{Introduction}

Over the last decade, a number of models have been identified in the framework of local 1+1D brickwork quantum circuits that appear to be both quantum chaotic and partially solvable, heralding a rare opportunity to obtain sharp analytical results in chaotic many-body quantum systems. Perhaps the prototypical example of such a model is that of random unitary circuits \cite{nahum2017entanglementgrowthRUCs, nahum2018opspreadingRUCs, vonkeyserlingk2018ophydrodynamicsRUCs, bertinipiroli2020scrambling, chan2018periodicrucs}, for which it was shown that out-of-time-order correlators (OTOCs) and the spectral form factor - amongst a number of other interesting dynamical quantities - are not only exactly calculable but also quantum chaotic, with a ballistic growth of generic local operators revealed by the OTOC \cite{nahum2018opspreadingRUCs,vonkeyserlingk2018ophydrodynamicsRUCs} and random-matrix level statistics revealed by the spectral form factor \cite{chan2018periodicrucs}.
\par
These calculations, however, require averaging over many realisations of the random circuit\footnote{This isn't really a problem for the spectral form factor - which, by definition, requires averaging over an ensemble of matrices anyway - but isn't desirable for the calculation of other quantities, such as OTOCs.}, and have exactness only in the limit of a large local Hilbert space dimension. These conditions can be circumvented in the model of \textit{dual-unitary circuits} (DUCs), where the gates forming the circuit are unitary in both the time direction and the spatial direction of the circuit \cite{bertini2019exactcorrfuncs, bertini2021randomsffDUCs, claeys2020mvqcs, claeys2021ergodic, aravinda2021bernoulli, rather2020ensembles, rather2022construction, borsi2022constructionofduqcs, prosen2021many, bertini2019entanglementspreading, gopalakrishnan2019finitedepthinfinitewidth,  piroli2020exactdynamics,reid2021entanglementbarriersinducs, zhou2022maximal, foligno2023growth, bertini2020opent1, bertini2020opent2, gombor2022superintegrableCA, fritzsch2021eth, ippoliti2022dynamical, ho2022exactemergent, Claeys2022emergentquantum,rampp2023haydenpreskill,zhou2020entangmemchaoticmbs, ippoliti2022fractalduMIPT, ippoliti2022postselectionfree, lu2021stdualityMIPT, claeys2022exactMIPT, kos2023circuitsofspacetime, masanes2023discrete, suzuki2022computational, logaric2023scarsinDUCs, stephen2022universal} (see Section \ref{sec:DUsSubsec}). This property can then be exploited to calculate the spectral form factor \cite{bertini2021randomsffDUCs} and OTOCs \cite{claeys2020mvqcs}, without requiring any averaging (for OTOCs) and without requiring a large local Hilbert space dimension. It is now also generally accepted that generic DUCs exhibit `maximal' chaos, with a spectral form factor that matches the prediction of random matrix theory shown in Ref.~\cite{bertini2021randomsffDUCs}, and a maximal butterfly velocity and exponential decay of the OTOC inwards from the edge of the strict causal light cone (enforced by the geometry of the circuit) shown in Ref.~\cite{claeys2020mvqcs}. They have also been shown to be hard\footnote{The problem of calculating expectation values of local observables at late times under the dynamics of 1+1D DUCs is $\mathsf{BQP}$-complete \cite{suzuki2022computational}.} to classically simulate, in general \cite{suzuki2022computational}.
\par
Remarkably, the analytical tractability of these circuits extends well beyond calculations of OTOCs and the spectral form factor; for instance, correlation functions of local observables \cite{bertini2019exactcorrfuncs} and state  \cite{bertini2019entanglementspreading, gopalakrishnan2019finitedepthinfinitewidth, piroli2020exactdynamics, reid2021entanglementbarriersinducs,zhou2022maximal,foligno2023growth,bertinipiroli2020scrambling} and operator \cite{bertini2020opent1, bertini2020opent2} entanglement entropies are all analytically accessible in DUCs, as well as there existing a class of `solvable' matrix product states for which the time evolution of local observables under any dual-unitary dynamics is exactly calculable \cite{piroli2020exactdynamics}. Again, calculation of these quantities for generic dual-unitary circuits often reveals ergodic behaviour (such as the exponential decay of local correlation functions \cite{bertini2019exactcorrfuncs}), and accordingly DUCs have been used to prove many aspects of thermalisation from first principles \cite{fritzsch2021eth,ippoliti2022dynamical, ho2022exactemergent, Claeys2022emergentquantum}. They have been used to obtain exact results on the use of the Hayden-Preskill protocol to recover quantum information (with it being analytically demonstrable that they dynamically realise perfect decoding) \cite{rampp2023haydenpreskill}, and were also the first non-random Floquet spin model shown to be treatable by entanglement membrane theory \cite{zhou2020entangmemchaoticmbs}. Much of this tractability also survives under the addition of measurements and non-unitarity (allowing the study of measurement-induced phase transitions in the context of these models) \cite{ippoliti2022fractalduMIPT, ippoliti2022postselectionfree, lu2021stdualityMIPT, claeys2022exactMIPT, kos2023circuitsofspacetime}, under weak unitary perturbations away from dual-unitarity \cite{kos2021ducpathint, rampp2023perturbedDUCsopspreading}, and in a recently introduced hierarchical generalisation of dual-unitarity \cite{yu2023hierarchical,foligno2023spreadinggendu,liu2023entangdynamgendu,rampp2023entanglementmem}.
\par
Despite primarily attracting attention as a solvable model of ergodic dynamics, however, DUCs are also capable - when finely tuned - of realising non-ergodic dynamics. In particular, they are capable of hosting \textit{solitons}: local operators which, up to a multiplicative phase, are simply spatially translated at the effective speed of light (i.e.~along an edge of the causal light cone) by the dual-unitary dynamics \cite{bertini2020opent2,gombor2022superintegrableCA}. Dynamical quantities in DUCs with solitons display hallmarks of non-ergodicity - the level statistics may become Poissonian (see numerical results in the Supplementary Material of Ref.~\cite{claeys2021ergodic}), and correlation functions pertaining to the solitons will not decay \cite{bertini2019exactcorrfuncs,claeys2021ergodic}. Moreover, an exponential number of conserved quantities can be constructed from each soliton\footnote{In a finite chain, a condition on the multiplicative phase needs to be satisified in order for this to be strictly true - see Section \ref{sec:results}, Eqns.~\ref{eqn:Q_s} and \ref{eqn:Q_t}.} \cite{bertini2020opent2,gombor2022superintegrableCA}. Typically, the presence of an extensive number of (local) conserved charges is associated with the notion of \textit{integrability} \cite{faddeev1995algebraicaspects,faddeev1987hamiltonian,faddeev1996algebraic,caux2011remarksonquantumintegrability}. Indeed, there exists a family of dual-unitary circuits on qubits that corresponds to a particular point of the trotterised XXZ model, which is solvable by Bethe Ansatz techniques \cite{vanicat2018integrabletrotterization,ljubotina2019ballisticspin,claeys2022corrsinintegrablecircs}. Connections between integrable dual-unitary circuits and cellular automata generated from Yang-Baxter maps have also been explored in Ref.~\cite{gombor2022superintegrableCA}.
\par
Notably, however, the dynamics of solitons in dual-unitary circuit models is rather atypical. Quantum integrable systems are thought to have weakly-interacting quasi-particles (that generate the constants of motion), and in a many-body quantum system undergoing some `generic' unitary dynamics, localised densities of conserved charge will typically diffuse and become `mixed in' with nonconserved operators (such as in random unitary circuits \cite{khemani2018conslawsRUCs,rakovszky2018diffusive}). In contrast, the solitons to be studied herein do not diffuse (the number of sites over which they are supported remains strictly constant) nor mix with other operators (or even other solitons), in that they do not produce any non-conserved operator density under the dynamics of the circuit. They can be thought of as \textit{completely non-interacting} quasi-particles (they do not scatter at all); a particularly strong form of integrability.
\par
It has previously been shown for 1+1D brickwork dual-unitary circuits that the existence of a conserved quantity formed from a sum over $n$ of $2n$-site spatial translations of some finite range operator also implies the existence of a soliton with unit phase \cite{bertini2020opent1}. In this work, we extend on this result by showing a similar one-to-one correspondence between a more general class of conserved quantities (we allow them to be inhomogeneous, by allowing for variation in the finite-range operator over the course of the sum) and solitons with complex phases. Specifically, in Section \ref{sec:results} we prove a more formal version of the following theorem.
\begin{informalthm}
    \textit{Fundamental charges for DUCs} (\textit{informal}). Take a $1$D chain of local quantum systems. If a quantity formed from a linear combination of finite-range operators (i.e.~each with non-trivial support over some finite number - say $w$ - of neighbouring sites) is conserved by a $1+1$D brickwork dual unitary circuit acting on the chain, then it can be broken down into a linear combination of conserved quantities generated by the solitons (which are finite sized, also with support over at most $w$ neighbouring sites) of the circuit. In the thermodynamic limit ($L \rightarrow \infty$), a result in the converse direction holds - a conserved quantity can be constructed from each soliton - giving a one-to-one correspondence between solitons and conserved quantities (see Remark \ref{remark:onetoonecorrespondence}).
\end{informalthm}
Our belief is that this adds to the growing understanding of operator dynamics and non-ergodicity in dual-unitary many-body quantum systems, and also strengthens the sense in which solitons can be seen as fundamental charges for these systems. In essence, given a conserved quantity in a dual-unitary circuit, one only needs to know that it is formed from a linear combination of finite-range operators in order to know that it can be decomposed in terms of solitons.
\par 
The paper is structured as follows. In Section \ref{sec:DUCs}, we recap the model of DUCs, with a focus on solitons and conserved densities. In Section \ref{sec:results}, our main results are presented, including a proof of the aforementioned correspondence between solitons and conserved densities (Theorem \ref{thm:fundamentalcharges}). In Section \ref{sec:constructingmbcqs}, the structure of many-body solitons in DUCs is considered (which are encompassed by Theorem \ref{thm:fundamentalcharges}), and we show that these can be constructed from smaller fermionic and solitonic constituent quantities. In Section \ref{sec:conc}, the results are recapped and directions for future work are discussed.

\section{Dual-unitary circuits}\label{sec:DUCs}
\subsection{Local 1+1D quantum circuits}\label{sec:localqcircs}
The model which will be studied here is that of a brickwork 1+1D local unitary quantum circuit. The spatial part of this model consists of a $1$D geometry of $2L$ sites arranged on an integer lattice with periodic boundary conditions, $x \in \mathbb{Z}_{2L}$. We will often want to deal with sets of sites of the same parity, and so to facilitate this we introduce a notation
\begin{equation}
    \mathbb{E}_x = \{ e \in \mathbb{N} \; | \; e < x, \; e = 0 \mod 2\},
\end{equation}
and
\begin{equation}
    \mathbb{O}_x = \{ o \in \mathbb{N} \; | \; o < x, \; o = 1 \mod 2\},
\end{equation}
to denote the sets of even and odd positive integers in the interval $\left[0,x-1\right]$. To each site we associate a local Hilbert space $\mathcal{H}_1 = \mathbb{C}^d$, such that the Hilbert space of the whole system is $\mathcal{H}_{2L} = \mathcal{H}_1^{\otimes 2L}$; in other words, we have a chain of $2L$ $d$-dimensional qudits.
\par
For the temporal part of the model we define a Floquet operator, $\mathbb{U} \in \textrm{End}(\mathcal{H}_{2L})$, which acts on the chain of qudits as
\begin{equation}\label{eqn:floquetoperator}
    \mathbb{U} = \Pi_{2L}V^{\otimes L}\Pi_{2L}^{-1}U^{\otimes L} = \;\;\;\includegraphics[width=0.4\linewidth, valign=c]{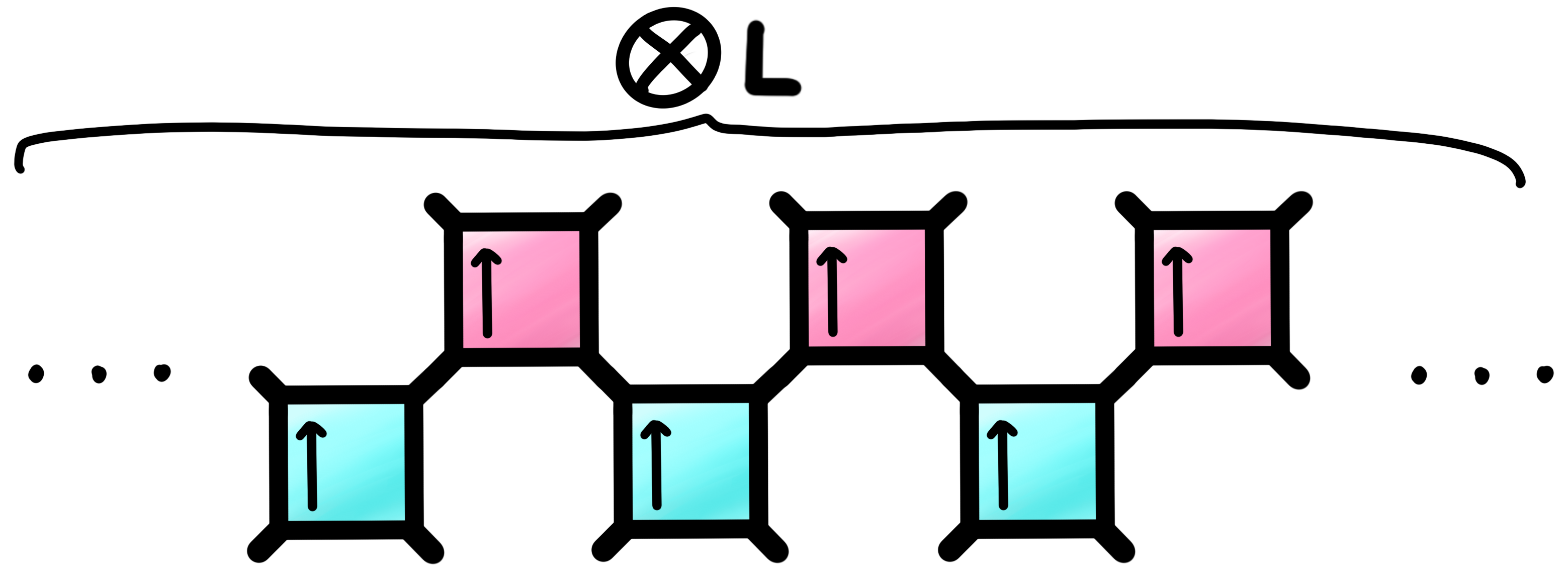}\;\;,
\end{equation}
where 
\begin{equation}
    U = \;\;\;\includegraphics[width=0.07\linewidth, valign=c]{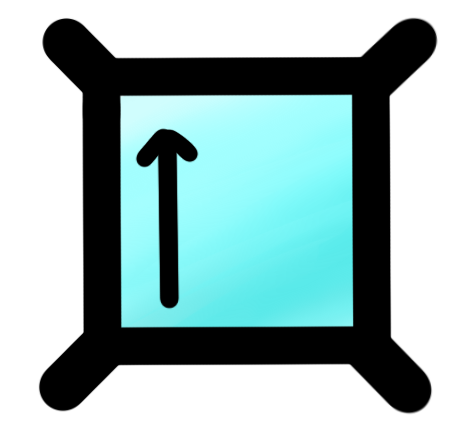}\;\;, \quad V = \;\;\;\includegraphics[width=0.07\linewidth, valign=c]{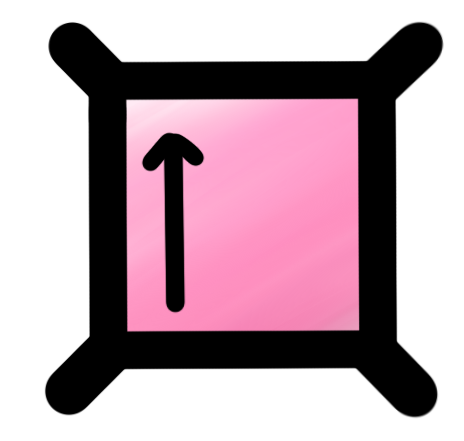}\;\; \quad \in \textrm{End}(\mathcal{H}_1 \otimes \mathcal{H}_1),
\end{equation}
are some 2-qudit unitary gates, and where we have introduced an $l$-site periodic one-site translation operator,
\begin{equation}
    \Pi_{l}\ket{i_1}\otimes...\otimes \ket{i_l} = \ket{i_2}\otimes...\otimes\ket{i_l}\otimes\ket{i_1},
\end{equation}
with the states $\{\ket{i}; i=1,...,d\}$ forming a complete orthonormal basis for $\mathcal{H}_1$. Note that we are considering a (spatially) homogeneous circuit, so the Floquet operator is invariant under spatial translations by two sites,
\begin{equation}
    \Pi_{2L}^2\mathbb{U}\Pi_{2L}^{-2} = \mathbb{U},
\end{equation}
which manifests in parity effects that motivate the $\mathbb{O}$ and $\mathbb{E}$ notation introduced above.
\par
A circuit enacting some unitary evolution for some time $t$ can then be constructed by taking $t$ sequential applications of the Floquet operator (so, save for the variation between the even and odd layers ($U\neq V$ in general), the circuit is also temporally homogeneous), i.e.~
\begin{equation}
    \mathcal{U}(t) = \mathbb{U}^t.
\end{equation}
This gives us a brickwork 1+1D local quantum circuit, as shown in Fig.~\ref{brickworkcircuit}. The structure of the Floquet operator enforces periodic boundary conditions on our chain; we have a $1$D ring topology.
\begin{figure}
      \includegraphics[width=0.9\linewidth]{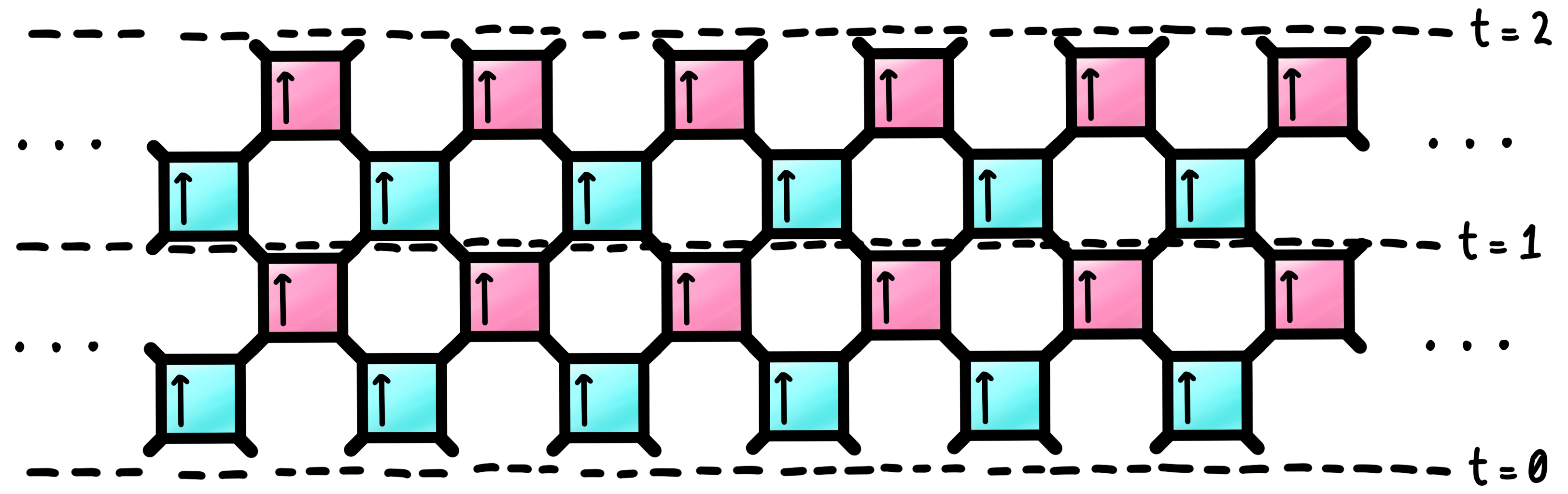}
      \caption{A brickwork circuit shown for two time steps, $t=2$, consisting of two Floquet periods, i.e.~two sequential applications of the Floquet operator $\mathbb{U}$.}
    \label{brickworkcircuit}
\end{figure}
\par
We will be primarily interested in the dynamics of `local' (we use the term local quite loosely here; we will be considering operators which act non-trivially on a connected region of consecutive sites of length $\leq L$, i.e.~half the length of the chain, $2L$) operators under the Floquet operator $\mathbb{U}$. To differentiate these operators from those generating the dynamics ($U$, $V$, $\mathbb{U}$) we will use bold font, representing some local operator on $l$ sites by $\bm{q} \in \mathcal{A}_l$, where $\mathcal{A}_l$ is the algebra of all matrices acting on $\mathcal{H}_1^{\otimes l}$.  We will use $\bm{q}_x \in \mathcal{A}_{2L}$ to denote an operator which acts as identity on all sites except the site $x$, where it acts as the single-qudit operator $\bm{q} \in \mathcal{A}_1$. We can then generalise this notation to represent more extensive operators: $\bm{q}_{x,y} \in \mathcal{A}_{2L}$ denotes the $2$-qudit operator $\bm{q} \in \mathcal{A}_2$ acting on the sites $x$ and $y$; $\bm{q}_{x,...,x+w-1} \in \mathcal{A}_{2L}$ denotes the width-$w$ operator $\bm{q} \in \mathcal{A}_w$ acting on the connected region of the chain $\mathcal{R} = \left[x,x+w-1\right] \subseteq \mathbb{Z}_{2L}, \; |\mathcal{R}| = w$.
Throughout this work we will almost always want to assume, when talking about some operator $\bm{q}_{x,...,y}$ acting over a region of the chain $[x,y]$, that $\tr_x(\bm{q}_{x,...,y}) = \tr_y(\bm{q}_{x,...,y}) = 0$ (i.e.~the operator is strictly traceless at the boundary sites of its region of non-trivial support). Accordingly, we introduce the notation
\begin{equation}
    \bar{\mathcal{A}}_w = \left\{\bm{a} \in \mathcal{A}_w | \tr_0(\bm{a}_{0,...,w-1}) = \tr_{w-1}(\bm{a}_{0,...,w-1})= 0 \right\},
\end{equation}
to explicitly pick out such operators.

\subsection{Dual-unitaries}\label{sec:DUsSubsec}
In order for the circuit $\mathcal{U}(t)$ to be unitary, the local gates $U$ and $V$ must also be unitary: $UU^{\dagger} = U^{\dagger}U = \mathds{1}$ and $VV^{\dagger} = V^{\dagger}V = \mathds{1}$. This can be represented diagrammatically as
\begin{equation}\label{eqn:unitarity}
    UU^{\dagger} = \;\;\; \includegraphics[width=0.07\linewidth, valign=c]{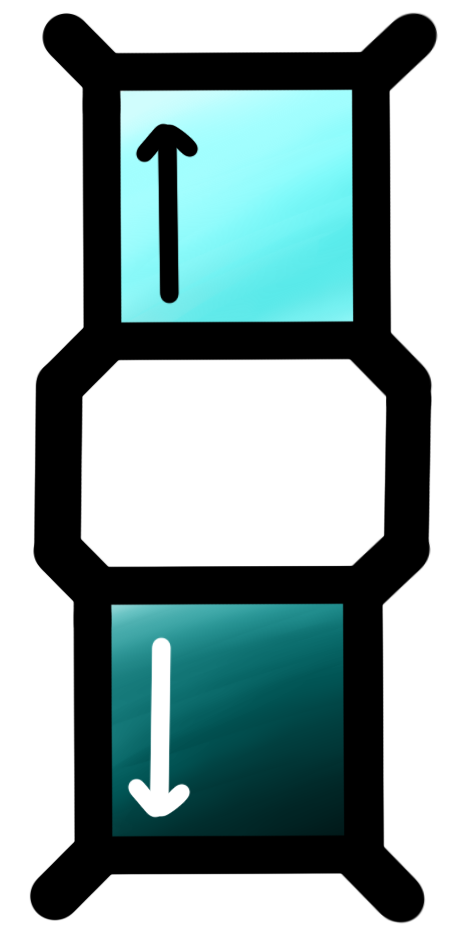} \;\;\; = U^{\dagger}U = \;\;\; \includegraphics[width=0.07\linewidth, valign=c]{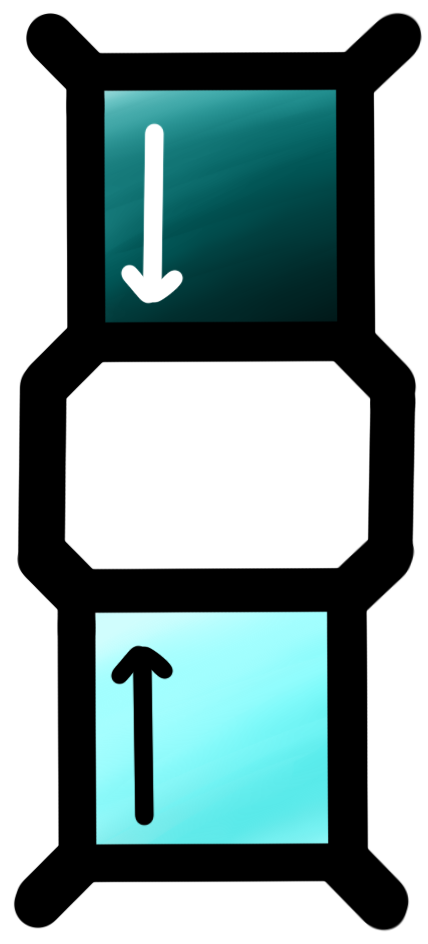} \;\;\; = \;\;\; \includegraphics[width=0.055\linewidth, valign=c]{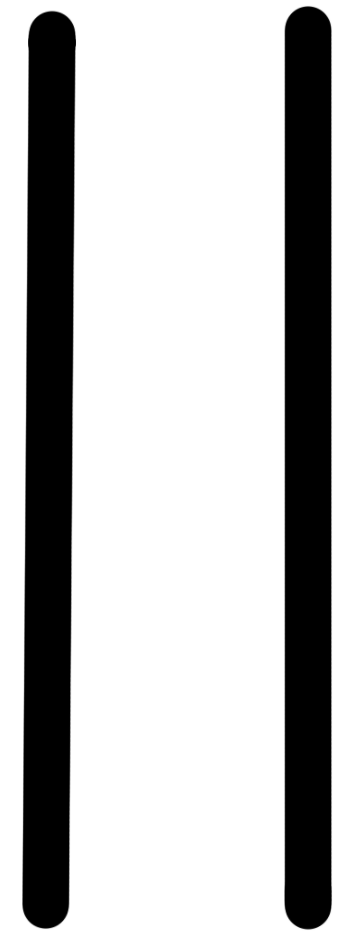} \;\;\;  = \mathds{1},
\end{equation}
where we are using the convention from \cite{masanes2023discrete} that a darker shade represents complex conjugation (i.e.~$\includegraphics[width=0.03\linewidth, valign=c]{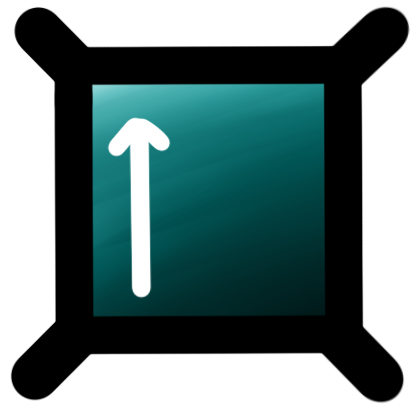} = U^*$), and flipping the arrow (which is included purely for the following illustrative, bookkeeping purposes) represents transposition (i.e~$\includegraphics[width=0.03\linewidth, valign=c]{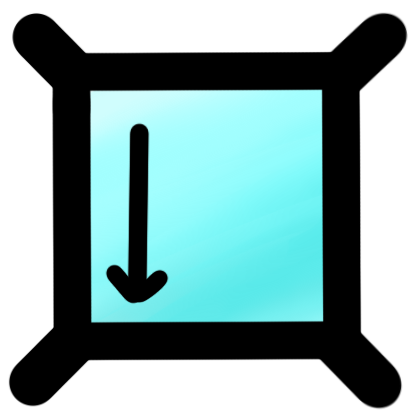} = U^T$), such that $\includegraphics[width=0.03\linewidth, valign=c]{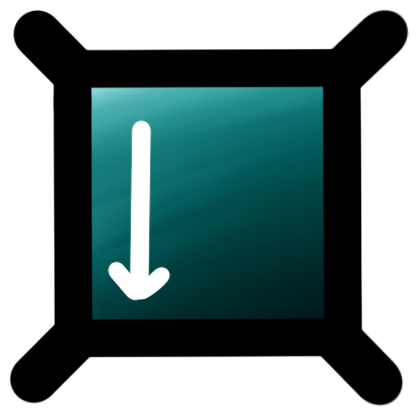} = U^{\dagger}$. We now enforce a further constraint on the circuit, in that it must be \textit{dual-unitary}; the dual-gates, defined via a reshuffling of indices like so
\begin{equation}
    \tilde{U}^{ij}_{kl} = \bra{i} \otimes \bra{j}\Tilde{U}\ket{k}\otimes\ket{l} = \bra{l} \otimes \bra{j}U\ket{k}\otimes\ket{i},
\end{equation}
must also be unitary, $\Tilde{U}\Tilde{U}^{\dagger}= \Tilde{U}^{\dagger}\Tilde{U}= \mathds{1}$ and $\Tilde{V}\Tilde{V}^{\dagger}= \Tilde{V}^{\dagger}\Tilde{V}= \mathds{1}$. 
Diagrammatically, this can be represented as
\begin{equation}
    \Tilde{U}\Tilde{U}^{\dagger} = \;\;\; \includegraphics[width=0.07\linewidth, valign=c]{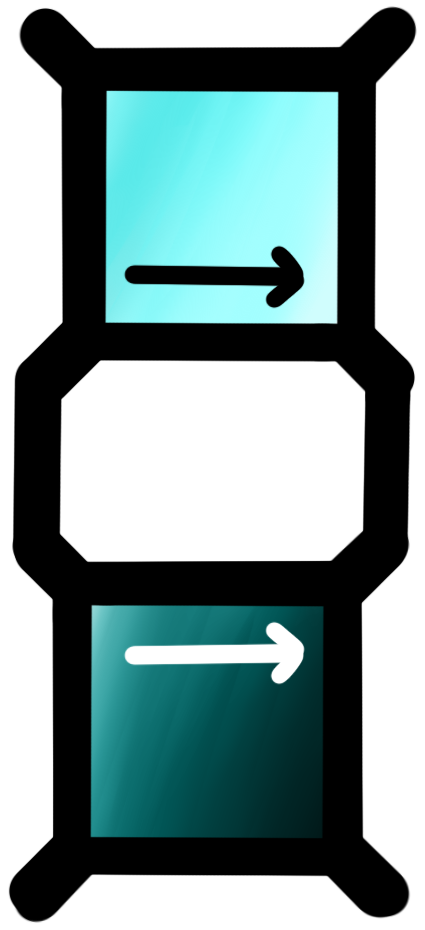} \;\;\; = \tilde{U}^{\dagger}\tilde{U} = \;\;\; \includegraphics[width=0.08\linewidth, valign=c]{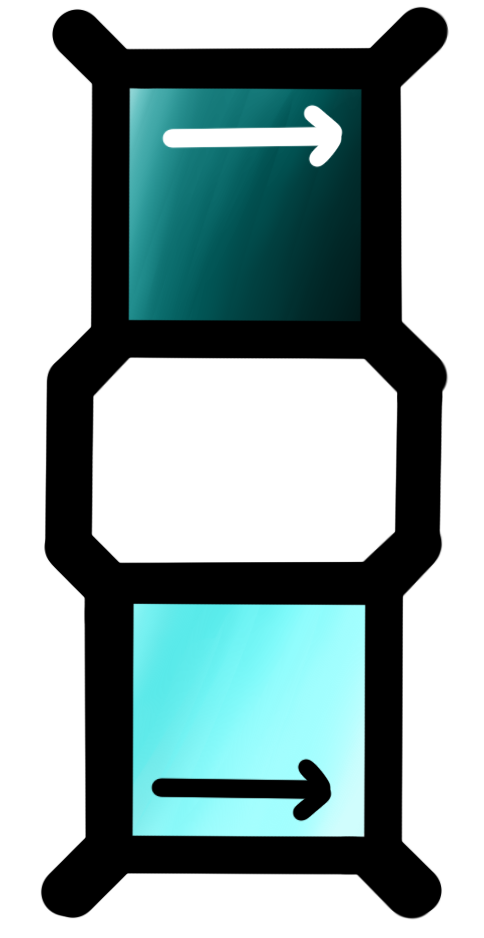} \;\;\; = \;\;\; \includegraphics[width=0.055\linewidth, valign=c]{identity.png} \;\;\; = \mathds{1},
\end{equation}
or alternatively as
\begin{equation}\label{eqn:dual_unitarity}
    \includegraphics[width=0.07\linewidth, valign=c]{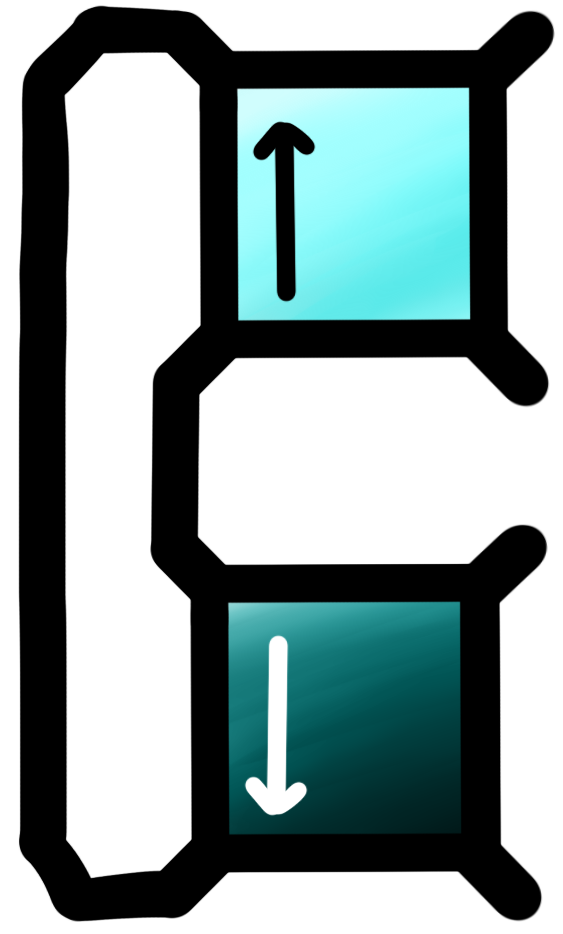} \;\;\; = \;\;\; \includegraphics[width=0.07\linewidth, valign=c]{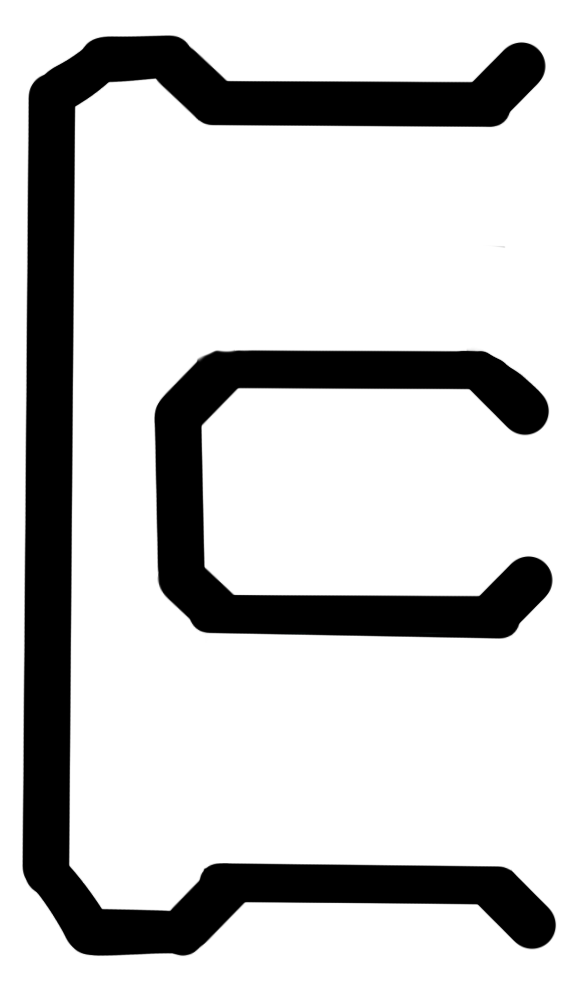} \;\;\;, \;\;\; \includegraphics[width=0.07\linewidth, valign=c]{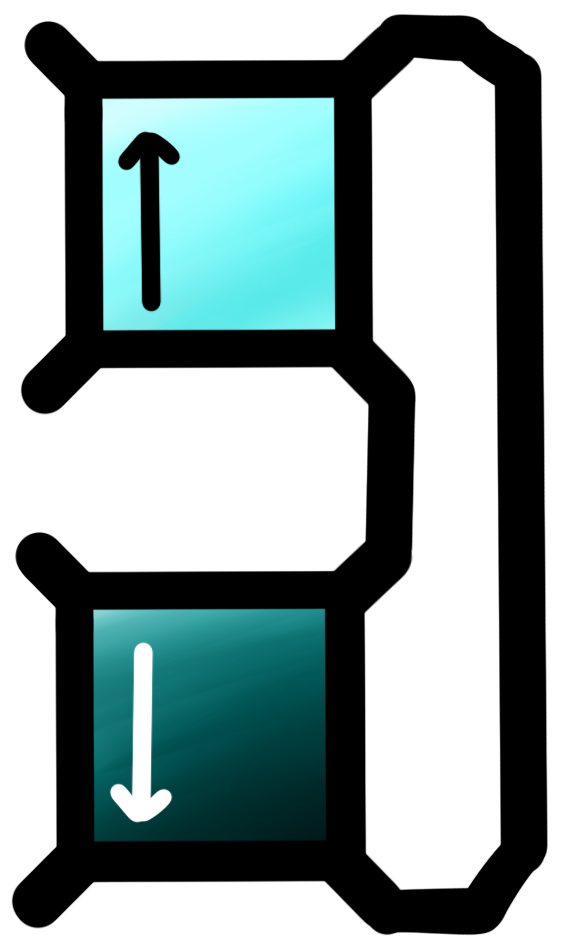} \;\;\; = \;\;\; \includegraphics[width=0.07\linewidth, valign=c]{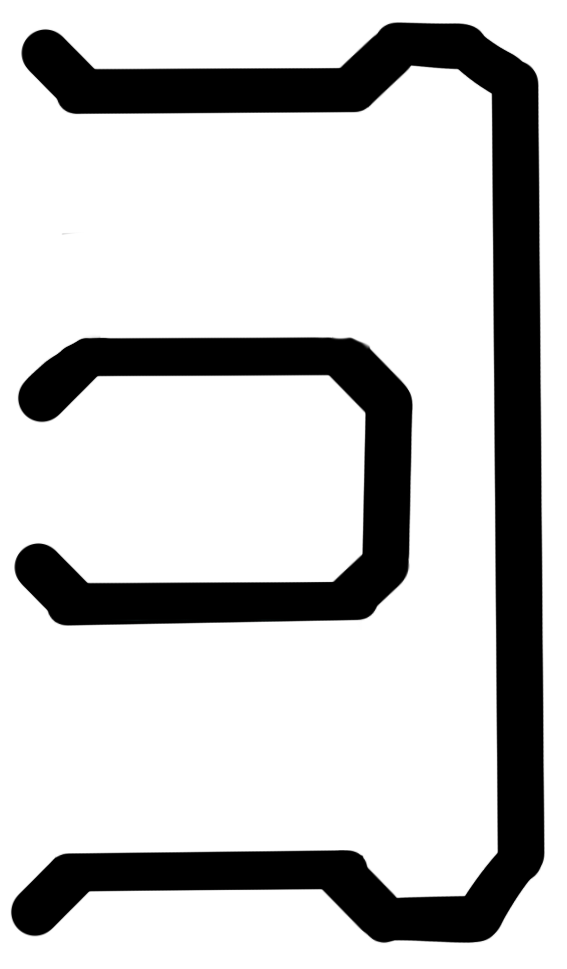} \;\;\;.
\end{equation}
This is equivalent to enforcing unitarity in the spatial direction of our 1+1D circuit.
\par
In \cite{masanes2023discrete}, it was noted that dual-unitaries have the following property:
\par
\begin{property}\label{property1}
Dual unitaries map any single-site traceless operator into a linear combination of operators that are necessarily traceless on the neighbouring site. That is, for a quantity $\bm{a} \in \mathcal{A}_1,\; \tr(\bm{a}) = 0$, initialised on a site $x$ and then acted upon by a dual-unitary gate $U$ coupling sites $x$ and $y$, the trace of $U\bm{a}_xU^{\dagger}$ is vanishing on site $y$,
\begin{equation}\label{eqn:du_property}
    \tr_y(U\bm{a}_xU^{\dagger}) = \;\;\; \includegraphics[width=0.09\linewidth, valign=c]{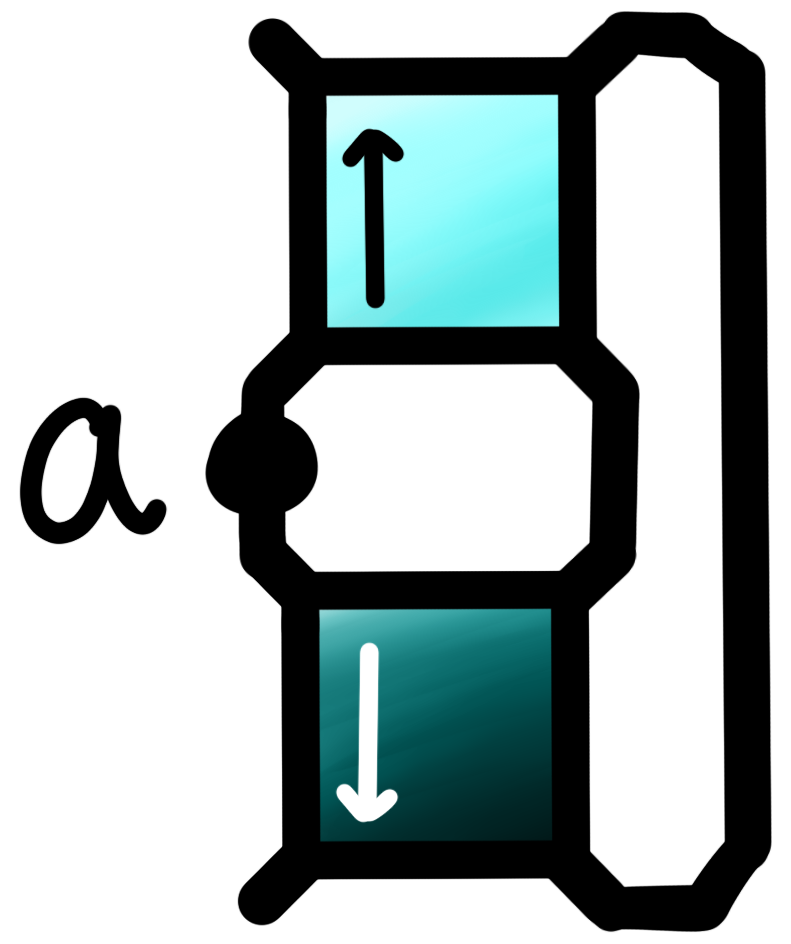} \;\;\; = \;\;\; \includegraphics[width=0.09\linewidth, valign=c]{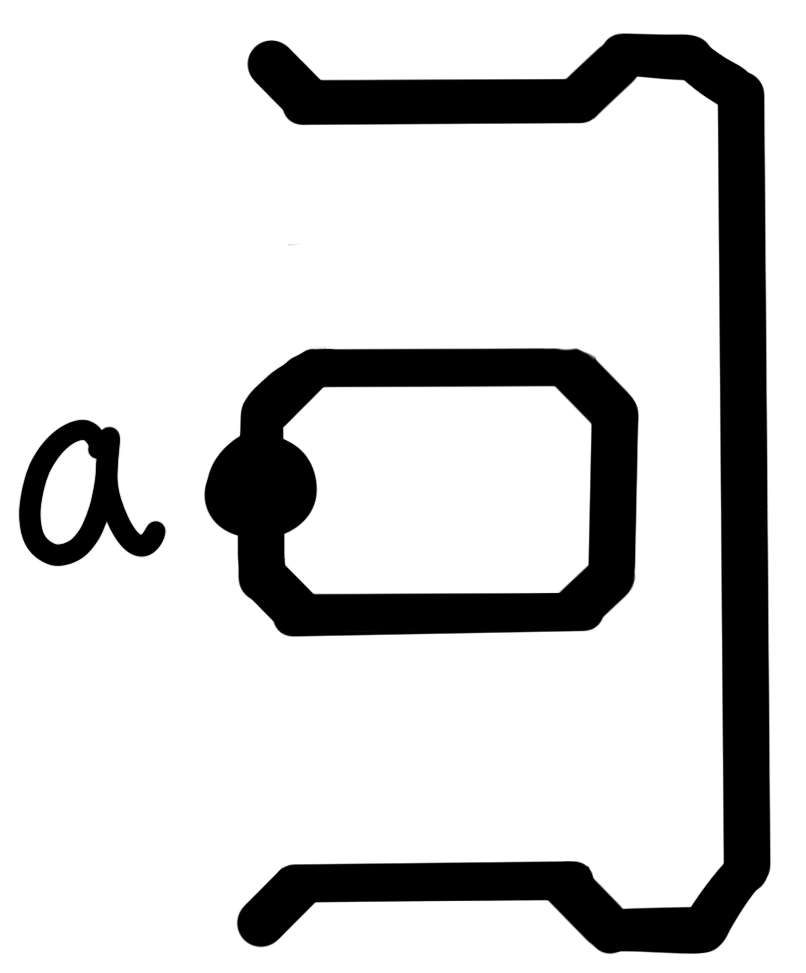} \;\;\; = \tr(\bm{a})\mathds{1}_x = 0,
\end{equation}
where we have used the dual-unitarity of $U$ as given by Eqn.~\ref{eqn:dual_unitarity} to simplify the diagram.
\end{property}
\par
The main result of \cite{bertini2019exactcorrfuncs} states that all correlation functions of local observables in 1+1D dual-unitary circuits are non-zero only on the edge of a causal light-cone set by unitarity and the geometry of the circuit.
Property \ref{property1} is essentially a `base case' of this.
We will use this property heavily in our proof of Theorem \ref{thm:fundamentalcharges}, the main result of this paper.

\subsection{Solitons and local conserved quantities} \label{sec:solitonsandlocalconservedQs}
In \cite{bertini2019exactcorrfuncs}, the authors also showed that correlations on the light cone edge can be calculated by determining the spectrum of the single-qudit map
\begin{equation}
    M_+(\;\cdot\;) = \frac{1}{d}\tr_0\left[ U( \;\cdot\; \otimes \mathds{1})U^{\dagger}\right] = \frac{1}{d} \quad \includegraphics[width=0.09\linewidth, valign=c]{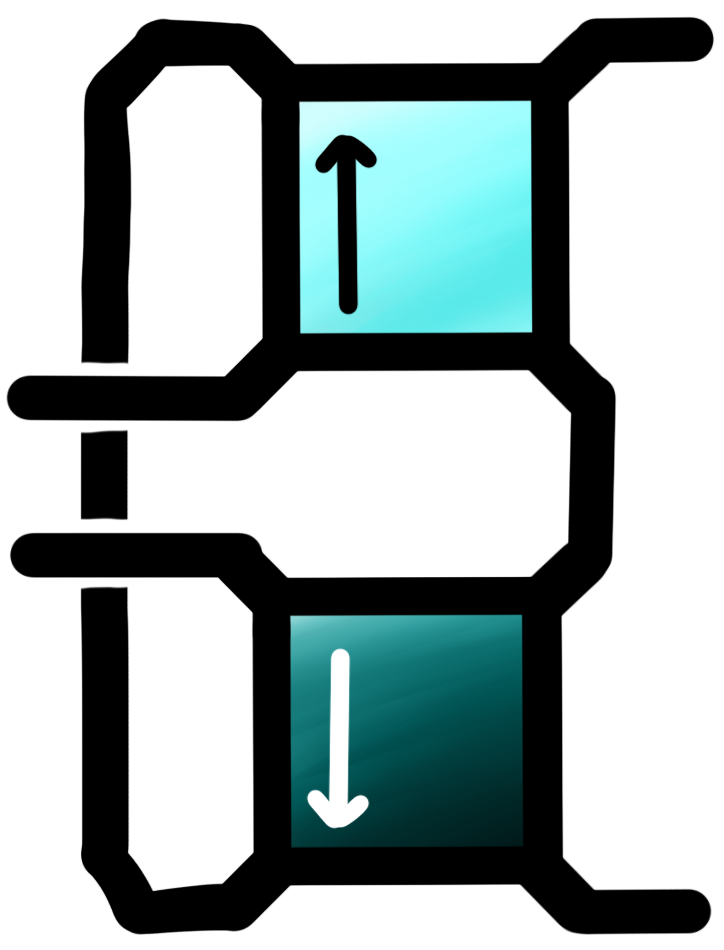} \quad,
\end{equation}
if the local observable $a$ is initialised on an even site, $x \in \mathbb{E}_{2L}$, or the single-qudit map
\begin{equation}
    M_-(\;\cdot\;) = \frac{1}{d}\tr_1\left[ U(  \mathds{1} \otimes \;\cdot\;)U^{\dagger}\right] = \frac{1}{d} \quad \includegraphics[width=0.09\linewidth, valign=c]{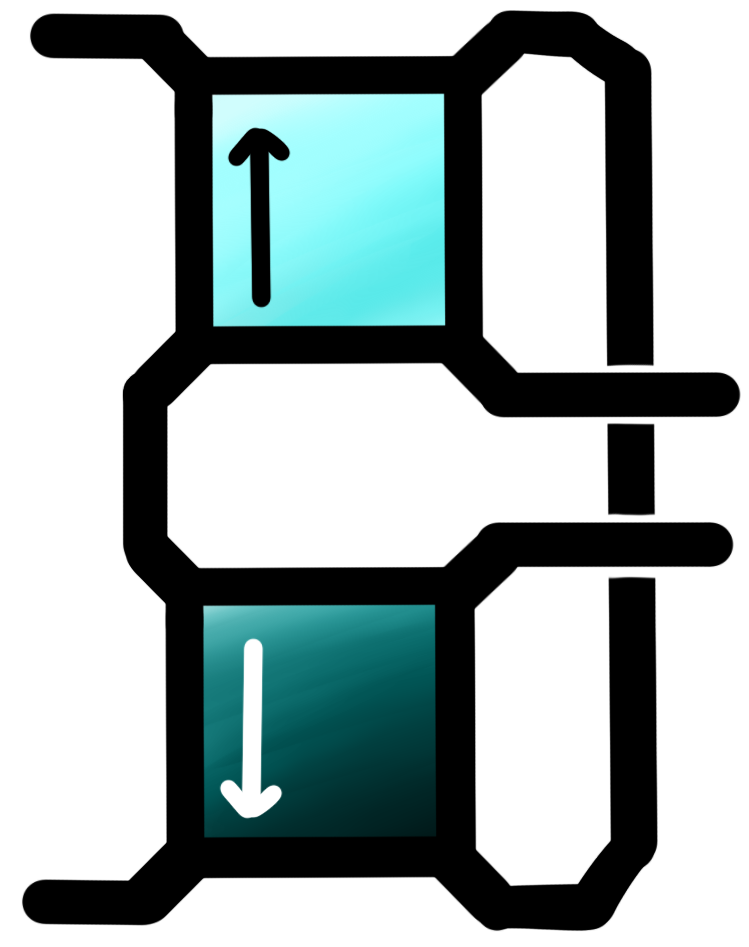} \quad,
\end{equation}
if $a$ is initialised on an odd site, $y \in \mathbb{O}_{2L}$. These maps are unital, CPTP maps, and consequently all their eigenvalues lie within the unit disc, $|\lambda| \leq 1$.
\par
Eigenvectors of the $M_+$ and $M_-$ maps with unimodular eigenvalue, $|\lambda| = 1$, correspond to single-site operators $\bm{a} \in \mathcal{A}_1$ which are simply shifted one site along the chain by each layer of the circuit, acquiring a phase $\lambda$ with each layer \cite{bertini2019exactcorrfuncs}. We note some important properties of these eigenvectors: not only are they eigenvectors of $M_{\pm}$, but they are also strictly localised, particle-like solutions for the dynamics generated by the Floquet operator, $\mathbb{U}$ (they are eigenvectors of the linear map corresponding to unitary conjugation by $\mathbb{U}^L$, specifically); they move at a constant velocity (the effective speed of light, $v_{LC} = 2$) without dissipation or distortion (as per Property \ref{property1} and Ref.~\cite{bertini2019exactcorrfuncs}, their support does not grow with time, nor do their correlations on the lightcone decay); and, if moving in opposite directions, they pass straight through each other without interacting (this can be seen by noting that $U(\bm{a}\otimes\mathds{1})U^{\dagger} = \mathds{1}\otimes \bm{a}$ and $U(\mathds{1}\otimes \bm{b})U^{\dagger} = \bm{b} \otimes \mathds{1}$ (for some $\bm{a},\;\bm{b}\in \mathcal{A}_1$) implies $U(\bm{a}\otimes \bm{b})U^{\dagger} = \bm{b} \otimes \bm{a}$). These properties have led to them being refered to as \textit{solitons} - borrowed from the study of integrability in many-body quantum systems \cite{faddeev1987hamiltonian,faddeev1996algebraic} - or \textit{gliders} - borrowed from the study of quantum cellular automata \cite{gutschow2010time,prosen2023two} - in the dual-unitaries literature \cite{bertini2020opent2,borsi2022constructionofduqcs}. In this work, we will choose to refer to them as solitons. We make explicit a definition of a width-$w$ soliton as follows.
\begin{definition}\label{def:wbodysolitons}
    \textit{Width-$w$ solitons}. An operator $\bm{a} \in \bar{\mathcal{A}}_{w}$ is a right-moving width-$w$ soliton of $\mathbb{U}$ if
    \begin{equation}
        \mathbb{U}\bm{a}_{x,..x+w-1}\mathbb{U}^{\dagger} = \lambda \bm{a}_{x+2,..x+w+1}, \; |\lambda|=1, \quad \forall \; x \in \mathbb{E}_{2L},
    \end{equation}
    or a left-moving width-$w$ soliton of $\mathbb{U}$ if
    \begin{equation}
        \mathbb{U}\bm{a}_{x,..x+w-1}\mathbb{U}^{\dagger} = \lambda \bm{a}_{x-2,..x+w-3}, \; |\lambda|=1, \quad \forall \; x \in \mathbb{O}_{2L},
    \end{equation}
    where $\mathbb{U}$ is a brickwork Floquet operator as defined in Eqn.~\ref{eqn:floquetoperator}.
\end{definition}
Recall that our definition of $\bar{\mathcal{A}}_w$ incorporates an assumption of non-trivial support at the boundary sites $x$ and $x+w-1$; this assumption ensures that it is meaningful to talk about the soliton having non-trivial support over a spatial region of the chain of size $w$.
\par
We note that our definition above does \textit{not} include operators which are (up to a multiplicative phase) spatially translated by just one site ($v = \pm 1$), or which remain still ($v = 0$), under conjugation by $\mathbb{U}$. Such operators have also been referred to as solitons in the literature (see Ref.~\cite{bertini2020opent2}). In Appendix \ref{sec:v1solitons}, we highlight that such operators cannot be many-body solitons in dual-unitary circuits (their stability under $\mathbb{U}$ is at best short-lived), extending the results of Ref.~\cite{bertini2020opent2} to the many-body case, and (along with the results of Appendix \ref{appendix:fate}) show that they cannot contribute to conserved quantities in dual-unitary circuits. In this work we seek to establish a correspondence between conserved quantities and solitons in DUCs; to avoid, for brevity, the need to repeatedly differentiate between the maximal velocity ($v=\pm 2$) solitons defined above and the $v = 0, \pm1$ solitons considered elsewhere when stating our results, we have chosen to incorporate an assumption of maximal velocity in our definition of solitons. To be clear, whenever we talk about solitons throughout the rest of this work we strictly mean the $v = \pm 2$ solitons that fit with Definition \ref{def:wbodysolitons}.
\par
For each soliton, we can construct a quantity $Q$ which is conserved under the dual-unitary dynamics of the Floquet operator, $\mathbb{U}$. For instance, for a 1-body soliton $\bm{a}$ with unit phase ($\lambda=1)$ the quantity 
\begin{equation}\label{eqn:unitphasesolitonconservedQs}
    Q = \sum_{x\in \mathbb{E}_{2L}} \bm{a}_x = \mathbb{U}Q\mathbb{U}^{\dagger},
\end{equation}
will be conserved under the dual-unitary dynamics. This can be seen by noting that $Q$ is invariant under 2-site translations,
\begin{equation}
    \mathbb{U}Q\mathbb{U}^{\dagger} = \sum_{x\in\mathbb{E}_{2L}} \mathbb{U}\bm{a}_x\mathbb{U}^{\dagger} = \sum_{x\in\mathbb{E}_{2L}} \bm{a}_{x+2} = \sum_{x\in\mathbb{E}_{2L}} \bm{a}_x = Q .
\end{equation}
We call $Q$ a \textit{width-1 conserved density}, according to the following definition of \textit{width-$w$ conserved densities}.
\begin{definition}\label{def:wbodydensities}
    \textit{Width-$w$ conserved densities}. Given a Floquet operator $\mathbb{U}$, an operator $Q \in \mathcal{A}_{2L}$ is a width-$w$ conserved density if $\mathbb{U}Q\mathbb{U}^{\dagger} = Q$, and it can be written as
    \begin{equation}
        Q = \sum_x \bm{q}(x)_{x,...,x+w-1},
    \end{equation}
    where $\bm{q}(x) \in \bar{\mathcal{A}}_w$ is a width-$w$ operator that can depend on the site $x$.
\end{definition}
As noted in Ref.~\cite{borsi2022constructionofduqcs}, the solitons form an algebra that is closed under addition and multiplication; products of solitons on sites on the same parity are also solitons, leading to an exponential number of conserved quantities \cite{bertini2020opent2,gombor2022superintegrableCA} (we recap this in further detail in Sec.~\ref{sec:boundsolitons}). If $\bm{a}$ is a soliton with phase $\lambda$, then $\bm{a}^n$ will also be a soliton with phase $\lambda^n$. If $\lambda$ is then a root of unity, there will be some finite $n$ for which a unit phase soliton can be generated.\footnote{There are a few caveats to note here: the product $\bm{a}^n$ may no longer be traceless, though could always be made traceless by subtracting multiples of the identity; it may transpire that $\bm{a}^n = \bm{\mathds{1}}$, which is still a soliton, but trivially so (the product would become zero upon making it traceless); some solitons may not generate a cyclic group under multiplication, with $\bm{a}^n=0$ instead for some finite $n$ (the fermionic conserved quantities considered later herein (Sec.~\ref{sec:fermionic}) would be one such example).}
\par
For these unit phase solitons, a strong correspondence with local conserved densities has already been established. In \cite{bertini2020opent1}, the authors proved that for all conserved densities of the form
\begin{equation}
    Q^{+} = \sum_{x\in \mathbb{E}_{2L}} \bm{q^{+}}_{x,...,x+w-1} = \mathbb{U}Q^+\mathbb{U}^{\dagger},
\end{equation}
where $\bm{q^{+}}_{x,...,x+w-1}$ is some local operator supported on a region of $w$ consecutive sites from the site $x$ onwards (i.e it acts non-trivially on the interval $[x, x+w-1]$\footnote{In \cite{bertini2020opent1} the authors use a slightly different convention for labelling the sites and defining the Floquet operator, so the statement of the results is slightly different.}), the dynamics of some local charge density, $\bm{q^+}_{x,...,x+w-1}$, will always be solitonic if $\mathbb{U}$ is dual-unitary, in that
\begin{equation}
    \mathbb{U}\bm{q^+}_{x,...,x+w-1}\mathbb{U}^{\dagger} = \bm{q^+}_{x+2,...,x+w+1}, \; \forall \; x\in\mathbb{E}_{2L}.
\end{equation}
Note that the authors only set $\tr_x(\bm{q^+}_{x,...,x+w-1})=0$, so that $\bm{q^+}_{x,...,x+w-1}$ only need strictly act non-trivially on the left-most site, $x$, and so in general $\bm{q^+}_{x,...,x+w-1} \notin \bar{\mathcal{A}}_w$. Consequently, $\bm{q^+}$ could be formed from a sum of solitons, each individually of width less than or equal to $w$, and all with $\lambda=1$.
For the above quantity, $Q^{+}$, we have summed over the even sites, but the results hold analogously if we instead consider the odd sites; if some operator $\bm{q^-}_{x-w+1,...,x}$, supported non-trivially on $w$ sites over the interval $[x-w+1, x]$, where now the right-most site is set to be strictly traceless, $\tr_x(\bm{q^-}_{x-w+1,...,x})=0$, forms a quantity
\begin{equation}
    Q^- = \sum_{x\in\mathbb{O}_{2L}} \bm{q^-}_{x-w+1,...,x},
\end{equation}
which is conserved under the dynamics of the circuit, $\mathbb{U}Q^-\mathbb{U}^{\dagger} = Q^-$, then it must be the case that
\begin{equation}
    \mathbb{U}\bm{q^-}_{x-w+1,...,x}\mathbb{U}^{\dagger} = \bm{q^-}_{x-w-1,...,x-2}, \; \forall \; x\in\mathbb{O}_{2L}.
\end{equation}
The above results of \cite{bertini2020opent1} establish a correspondence between some conservation densities (those with spatial homogeneity) and unit phase solitons in DUCs. It was clear that for each width-$w$ soliton with $\lambda=1$, we could construct a width-$w$ conserved density associated to this soliton; but, as is shown in Ref.~\cite{bertini2020opent1}, we now also know that for any width-$w$ conserved density in a DUC where $\bm{q}(x) = \bm{q}$ (i.e.~such that $Q$ is homogeneous - the local operator generating it does not depend on $x$\footnote{Technically, as shown here, the conserved quantities considered in Ref.~\cite{bertini2020opent1} can incorporate parity effects (and hence some dependence on $x$) - $Q$ can be formed from two independent conserved quantities $Q^+$ (even $x$) and $Q^-$ (odd $x$). These two quantities are themselves spatially homogeneous and independent of each other. The spatial inhomogeneity incorporated in Theorem \ref{thm:fundamentalcharges} cannot be captured by these parity effects alone.}), the local operator $\bm{q}$ evolves solitonically (with $\lambda = 1$) under the dynamics of the circuit. Concretely, we have
\begin{equation}
    Q = \sum_{x\in \mathbb{E}_{2L}} \bm{a}_{x,...,x+w-1} = \mathbb{U}Q\mathbb{U}^{\dagger} \iff \mathbb{U}\bm{a}_{x,...,x+w-1}\mathbb{U}^{\dagger} = \bm{a}_{x+2,...,x+w+1} \; \forall \; x \in \mathbb{E}_{2L},
\end{equation}
and
\begin{equation}
    Q = \sum_{x\in \mathbb{O}_{2L}} \bm{a}_{x,...,x+w-1} = \mathbb{U}Q\mathbb{U}^{\dagger} \iff \mathbb{U}\bm{a}_{x,...,x+w-1}\mathbb{U}^{\dagger} = \bm{a}_{x-2,...,x+w-3} \; \forall \; x \in \mathbb{O}_{2L},
\end{equation}
if $\bm{a}\in\mathcal{A}_w$ and $\mathbb{U}$ is a dual-unitary Floquet operator. We can demand that $\tr_x(\bm{a}_{x,...,x+w-1}) = \tr_{x+w-1}({\bm{a}_{x,...,x+w-1}}) = 0$ - such that $\bm{a} \in \bar{\mathcal{A}}_w$ - in order to fit with our definition of width-$w$ solitons, establishing a one-to-one correspondence between width-$w$ solitons with $\lambda=1$ and width-$w$ conserved densities with $\bm{q}(x) = \bm{a}$.
\par
Our primary contribution in this paper - contained in the following section, Section \ref{sec:results} - is to extend upon this result by showing, for 1+1D dual-unitary circuits, a one-to-one correspondence between the full set of width-$w$ solitons (including those with $\lambda \neq 1$) as per Definition \ref{def:wbodysolitons} and the full set of width-$w$ conserved densities (including those where $\bm{q}(x)$ is not constant with $x$) as per Definition \ref{def:wbodydensities}. This, along with the supporting work contained in Appendices \ref{appendix:fate} and \ref{appendix:solitonbasis} - which highlights the constraints dual-unitarity places on local operators and conserved charges in further detail - sharpens the characterisation of conserved charges in dual-unitary circuits.

\section{Main results}\label{sec:results}

\subsection{Solitons with complex multiplicative phase}\label{sec:complexsolitons}
Earlier, we showed how to construct conserved quantities from solitons with unit phase, $\lambda = 1$. As per Definition \ref{def:wbodysolitons}, however, we can also have solitons with complex, non-unit phases, $\lambda \in \mathbb{C}$ - it remains possible to construct an associated conserved quantity,
\begin{equation}\label{eqn:complexsolitonconservedQs}
    \mathbb{U}\bm{a}_x\mathbb{U}^{\dagger} = \lambda \bm{a}_{x+2}, \; x \in \mathbb{E}_{2L}, \; \lambda^{L} = 1 \implies Q = \sum_{x\in\mathbb{E}_{2L}} \lambda^x \bm{a}_x = \mathbb{U}Q\mathbb{U}^{\dagger},
\end{equation}
as long as $\lambda$ is some $L$-th root of unity, where $2L$ is the length of the chain \cite{borsi2022constructionofduqcs}.
It is these width-$w$ solitons - with complex phases - for which we will be able to establish a correspondence with the set of width-$w$ conserved densities (as per Definition \ref{def:wbodydensities}).
\par
In order to make a convenient constructive definition of the sets of right- and left-moving width-$w$ solitons, we introduce the following maps:
\begin{align}
    \mathcal{M}_{+,w}(\;\cdot\;) &= \frac{1}{d^2}\tr_1\left[V^{\otimes \frac{w+1}{2}}\left(\tr_0\left[U^{\otimes \frac{w+1}{2}}(\;\cdot\;\otimes\mathds{1}){U^{\dagger}}^{\otimes \frac{w+1}{2}}\right] \otimes \mathds{1} \right){V^{\dagger}}^{\otimes \frac{w+1}{2}}\right] \;\; \in \textrm{End}(\mathcal{H}_w), \\
    & = \frac{1}{d^2} \quad \quad \includegraphics[width=0.2\linewidth, valign=c]{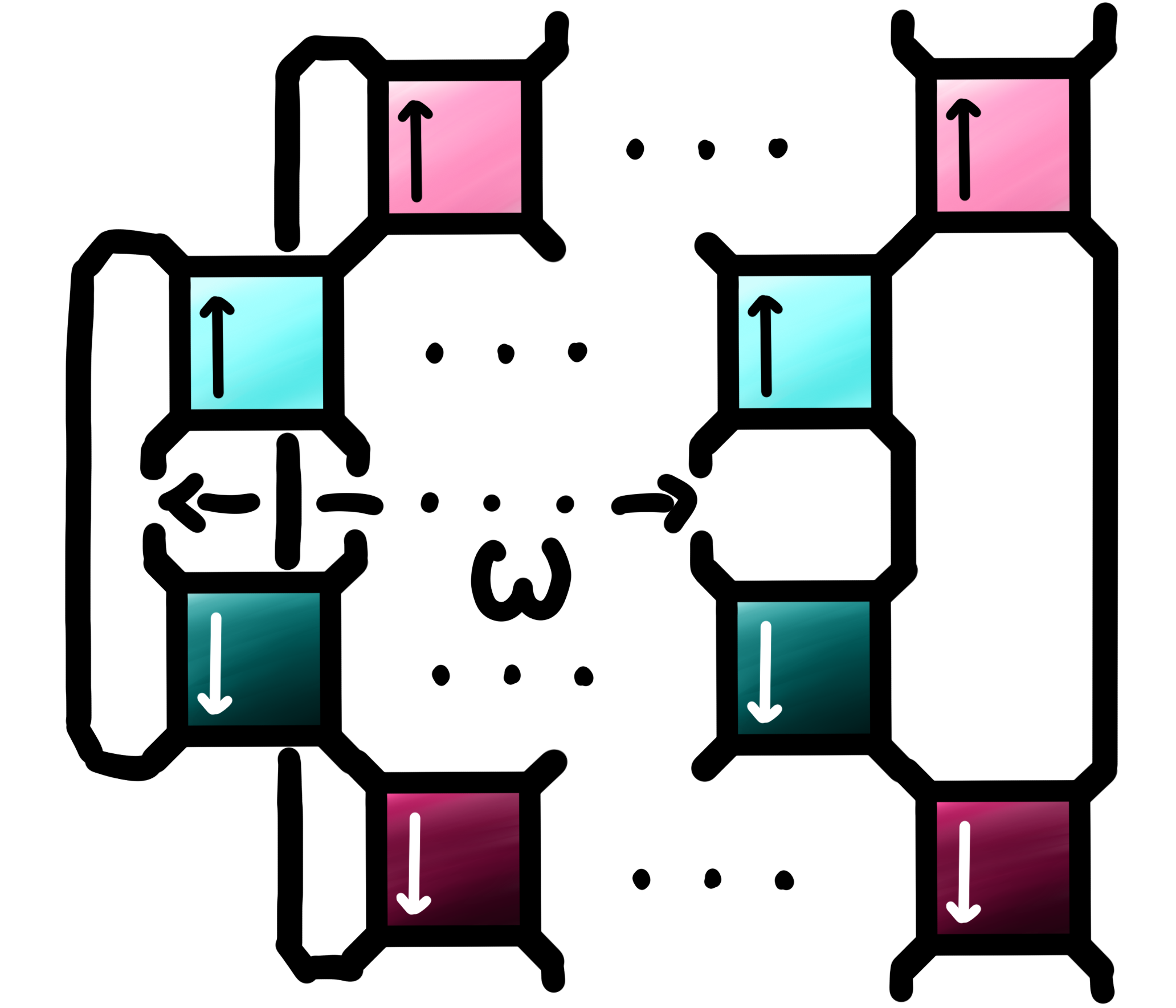} \quad ;
\end{align}
and
\begin{align}
    \mathcal{M}_{-,w}(\;\cdot\;) &= \frac{1}{d^2} \tr_{w-1}\left[V^{\otimes \frac{w+1}{2}}\left(\mathds{1} \otimes \tr_{w}\left[U^{\otimes \frac{w+1}{2}}(\mathds{1}\otimes \;\cdot\;){U^{\dagger}}^{\otimes \frac{w+1}{2}}\right]\right){V^{\dagger}}^{\otimes \frac{w+1}{2}}\right] \;\; \in \textrm{End}(\mathcal{H}_w), \\
    & = \frac{1}{d^2} \quad \quad \includegraphics[width=0.2\linewidth, valign=c]{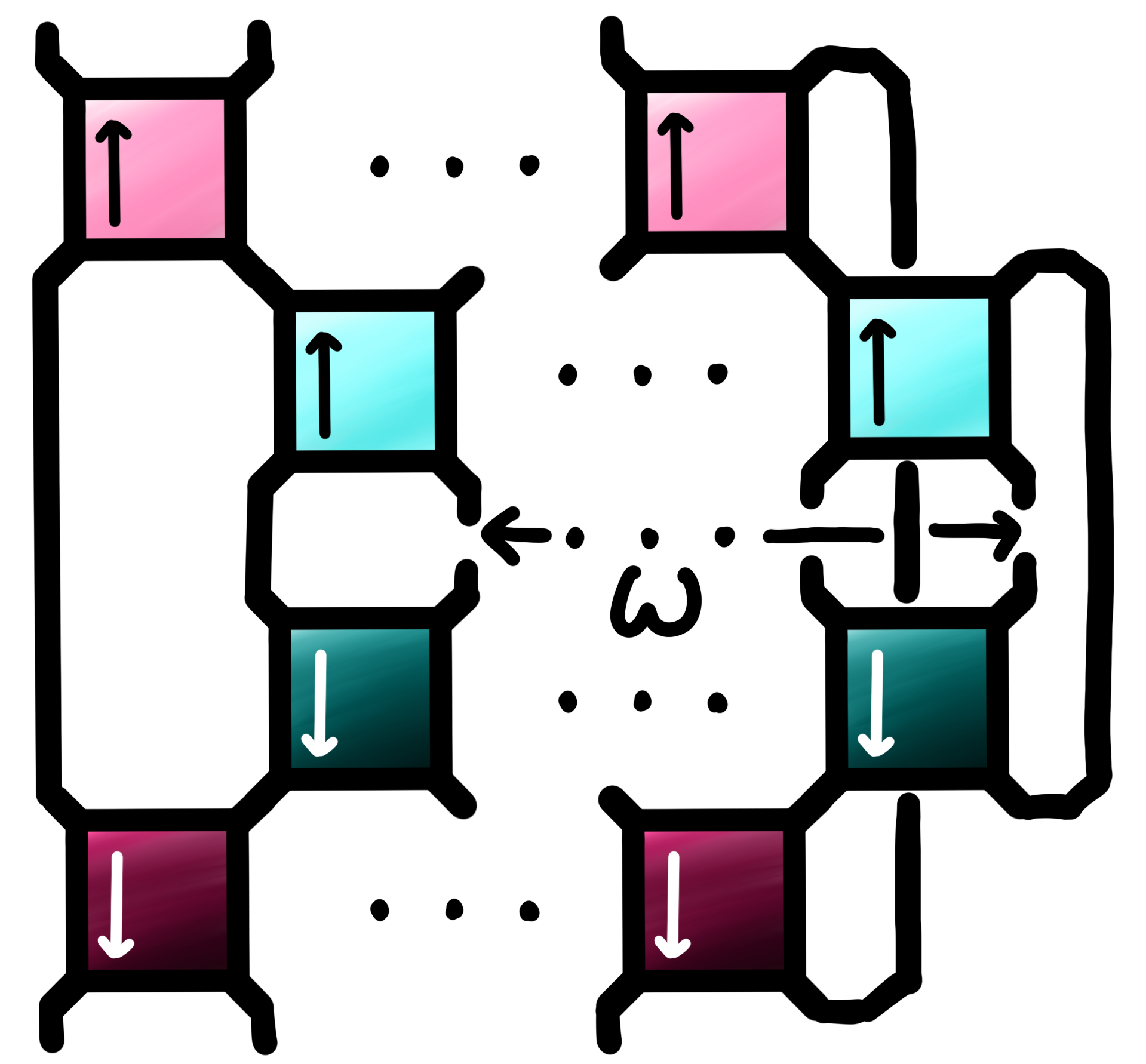} \quad ,
\end{align}
where we have used $\includegraphics[width=0.03\linewidth, valign=c]{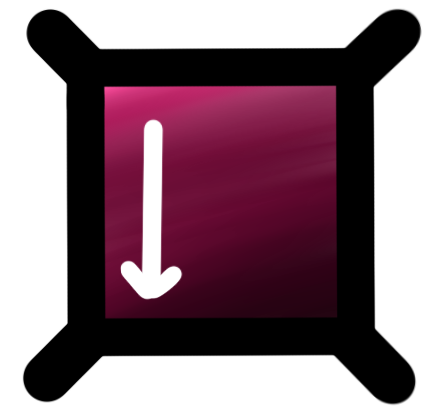} = V^{\dagger}$. These are generalisations of the $M_{\pm}$ maps introduced in \cite{bertini2019exactcorrfuncs}, and the width-$w$ solitons are eigenvectors of them with unimodular eigenvalue. They are defined over two layers of the circuit as in this work we allow, for the purpose of generality, for inhomogeneity between the even and odd layers of the circuit - the gates forming the Floquet operator, $U$ and $V$, are not necessarily the same. 
For the right-moving width-$w$ solitons we define a subset
\begin{equation}\label{eqn:set_S}
    \mathbb{S}_{+,w} = \{\bm{a_k} \;|\; k\neq0, \; \mathcal{M}_{+,w}(\bm{a_k}) = \lambda^{+}_{k,w}\bm{a_k}, \; |\lambda^{+}_{k,w}| = 1\},
\end{equation}
of a basis $\{\bm{a_k}\}_{k=0}^{d^{2^w}-1}$ for $\bar{\mathcal{A}}_w$, with $\bm{a_0} = \mathds{1}^{\otimes w}$. Note that we have extended the bold notation for the operator to its associated subscript here, in contrast to the unbolded subscripts used earlier, to differentiate between indexing that pertains to the composition of the operator - as the subscripts in this equation do, labelling distinct solitons - from indexing that pertains simply to the spatial position of the operator on the lattice, as for example in Eqn.~\ref{eqn:complexsolitonconservedQs} above. In later equations (e.g. Eqns.~\ref{eqn:Q_s} and \ref{eqn:Q_t} below) both kinds of subscripts will be needed, necessitating this notation choice.
\par
Similarly, for the left-moving width-$w$ solitons we define a subset
\begin{equation}\label{eqn:set_T}
    \mathbb{S}_{-,w} = \{\bm{b_l} \;|\; l\neq 0, \;\mathcal{M}_{-,w}(\bm{b_l}) = \lambda^{-}_{l,w}\bm{b_l}, \; |\lambda^{-}_{l,w}| = 1\},
\end{equation}
of another basis $\{\bm{b_l}\}_{l=0}^{d^{2^w}-1}$ for $\bar{\mathcal{A}}_w$, with $\bm{b_0} = \mathds{1}^{\otimes w}$. 
\par
For the $\mathcal{M}_{\pm,w}$ maps to be well-defined we require that $(w+1)/2$ is integer, and hence $w$ must be an odd integer. This means that the above definitions only capture solitonic quantities of odd width - it is fairly straightforward to show, however, that we cannot have even-width operators which are spatially translated by two sites by $\mathbb{U}$ (and hence fit our definiton of a soliton by Definition \ref{def:wbodysolitons}) in a brickwork unitary circuit. This fact is clarified in Appendix \ref{appendix:fate} (See the digraph in Fig.~\ref{fig:digraph}; only operators with $w$ odd can remain the same width and be spatially translated by two sites under the action of a dual-unitary Floquet operator $\mathbb{U}$).
\par
Again, we can associate conserved quantities to each of these width-$w$ solitons,
\begin{equation}\label{eqn:Q_s}
    \bm{a_k} \in \mathbb{S}_{+,w}, \; (\lambda^{+}_{k,w})^{L} = 1 \implies Q^{+}_{k,w} = \sum_{x\in\mathbb{E}_{2L}} \left(\lambda^{+}_{k,w}\right)^{\frac{x}{2}}\left[\bm{a_k}\right]_{x,..,x+w-1} = \mathbb{U}Q^{+}_{k,w}\mathbb{U}^{\dagger},
\end{equation}
for the right-moving solitons, and
\begin{equation}\label{eqn:Q_t}
    \bm{b_l} \in \mathbb{S}_{-,w}, \; \left(\lambda^{-}_{l,w}\right)^{L} = 1 \implies Q^{-}_{l,w} = \sum_{x\in\mathbb{O}_{2L}} \left(\lambda^{-}_{l,w}\right)^{-\frac{x-1}{2}}\left[\bm{b_l}\right]_{x,...,x+w-1} = \mathbb{U}Q^{-}_{l,w}\mathbb{U}^{\dagger},
\end{equation}
for the left-moving solitons, where we have in both cases used square brackets around the local operators to separate the indices labelling the basis with the indices labelling the sites of support.
\par
The requirement that the eigenvalues are unity when raised to the power $L$ means that for some solitons (i.e.~those for which $\lambda^L \neq 1$) it may not be possible to write down an associated width-$w$ conserved density. Note, however, that we can get around this restriction and construct larger conserved densities from any soliton by taking non-local products of it with its hermitian conjugate (see Section \ref{sec:boundsolitons}), and that in the case of an infinite chain, $L \rightarrow \infty$, we can ignore the restriction altogether, and obtain a strict mapping from each width-$w$ soliton to a width-$w$ conserved density (see Remark \ref{remark:onetoonecorrespondence}). 
\par
It is these conserved quantities given in Eqns.~\ref{eqn:Q_s} and \ref{eqn:Q_t} which, in a sense, represent a set of fundamental charges for the system, and out of which we can construct any other conserved quantity formed from sums of finite-range terms.  The formal statement and proof of this result follows in Sec.~\ref{sec:fundamentalcharges} below.

\subsection{Dynamical constraints on the spatial support of local operators in dual-unitary circuits}\label{sec:dynamicalconstraints}

Before moving on to the formal proof of Theorem \ref{thm:fundamentalcharges}, we will first establish some technical results that lay the foundation for the proof. In particular, we turn our attention to the constraints that Property \ref{property1} enforces on the support of local operators as they are time-evolved under dual-unitary dynamics.
\par
To facilitate the analysis, we introduce a compressed diagrammatic notation for representing the support of local operators on the chain, and the effect of local gates acting in adjoint action on this support. We use a large red dot to denote a traceless operator on a single site, and a smaller red dot to denote a single site operator with an unspecified trace. We can then represent larger operators on the chain by chaining these dots together; for instance, we would represent an operator supported over $4$ sites, $\bm{q}_{x,...,x+3}$, as
\begin{equation}
    \bm{q}_{x,...,x+3} = \;\;\;\includegraphics[width=0.2\linewidth, valign=c]{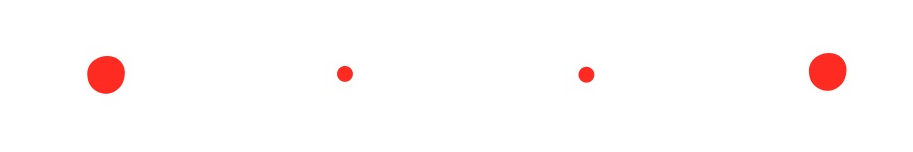}\;\;,
\end{equation}
where $\bm{q} \in \bar{\mathcal{A}}_4$ in this example and our site indexing convention implies that $\tr_x(\bm{q}_{x,...,x+3}) = \tr_{x+3}(\bm{q}_{x,...,x+3}) = 0$, that $\tr_{x+1}(\bm{q}_{x,...,x+3}), \; \tr_{x+2}(\bm{q}_{x,...,x+3}) \in [0,d]$, and that $\tr_y(\bm{q}_{x,x+3}) = d$ for all $y \notin [x,x+3]$. We then use horizontal black lines above neighbouring pairs of dots to represent the adjoint action of each two-qudit gate on them (one could think of this as similar to the `folded' picture used throughout the literature - e.g.~in Ref.~\cite{bertini2020opent1} - in that we are representing unitary conjugation with a single gate). For instance, we could diagrammatically represent the statement of Property \ref{property1} - that a single site operator on a site $x$ is mapped under conjugation by a two-qudit dual-unitary gate into a linear combination of terms that are traceless on the neighbouring site - as
\begin{equation}
    U\bm{a}_xU^{\dagger} = \bm{b}_{x+1} + \bm{c}_{x, x+1} \quad \iff \quad \includegraphics[width=0.15\linewidth, valign=c]{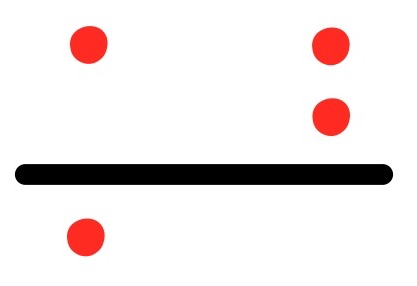}\;\;,
\end{equation}
where $\bm{a}, \bm{b} \in \bar{\mathcal{A}}_1$ are traceless one-qudit operators and $\bm{c} \in \bar{\mathcal{A}}_2$ is a two-qudit operator, with $\tr_x(\bm{c}_{x,x+1}) = \tr_{x+1}(\bm{c}_{x,x+1}) = 0$. Note that chains of red dots stacked on top of each other after (as in the line above), before, or between gates should be taken to represent linear combinations of operators with the support indicated by the dots, with arbitrary (up to normalisation) coefficients. 
\par
As a further example, if we were to evolve a width-1 operator $\bm{q}_x$ by a dual-unitary Floquet operator, $\mathbb{U}$, then we would find by applying Property \ref{property1} that it is mapped to a linear combination of width-1 terms that are supported on the site $x+2$ (assuming $x$ is even), width-$2$ terms supported on sites $x+1$ and $x+2$, and width-$4$ terms supported over the interval $[x-1,x+2]$, i.e.
\begin{equation}\label{eqn:1bodyunderfloquetDU}
    \mathbb{U}\bm{q}_x\mathbb{U}^{\dagger} =  \left[\mathcal{M}_{+,1}(\bm{q})\right]_{x+2} + \bm{q^{\prime}}_{x+1, x+2} + \bm{q^{\prime\prime}}_{x-1,...,x+2},
\end{equation}
where $\mathcal{M}_{+,1}(\bm{q}) \in \bar{\mathcal{A}_1}$, $\bm{q^{\prime}} \in \bar{\mathcal{A}_2}$ and $\bm{q^{\prime\prime}} \in \bar{\mathcal{A}_4}$. Diagrammatically, we would represent this as
\begin{equation}\label{eqn:1bodyunderfloquetDUdiag}
    \mathbb{U}\bm{q}_x\mathbb{U}^{\dagger} \quad \iff \quad \includegraphics[width=0.25\linewidth, valign=c]{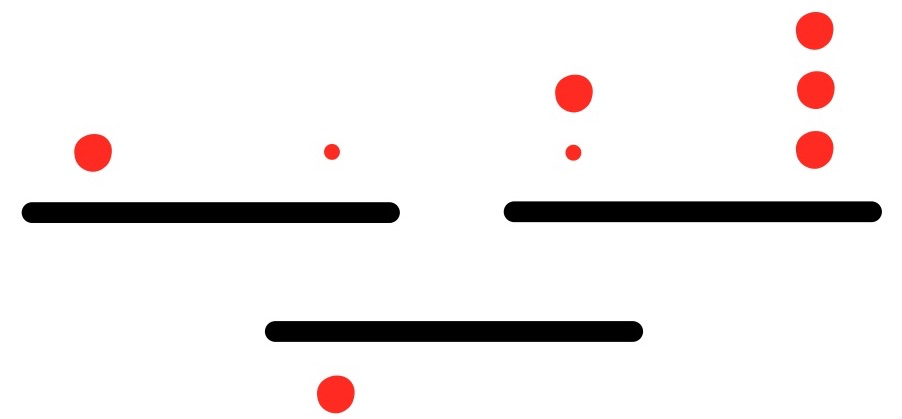}\;\;.
\end{equation}
We have noted in Eqn.~\ref{eqn:1bodyunderfloquetDU} that the width-1 terms supported on site $x+2$ can be calculated by using the 1-qudit light cone transfer matrix $\mathcal{M}_{+,1}$, i.e.
\begin{equation}
    \mathcal{M}_{+,1}(\bm{q}) = \frac{1}{d^2}\tr_1\left[V(\tr_0[U(\bm{q}\otimes\mathds{1})U^{\dagger}]\otimes\mathds{1})V^{\dagger}\right] = \frac{1}{d^2}\includegraphics[width=0.15\linewidth, valign=c]{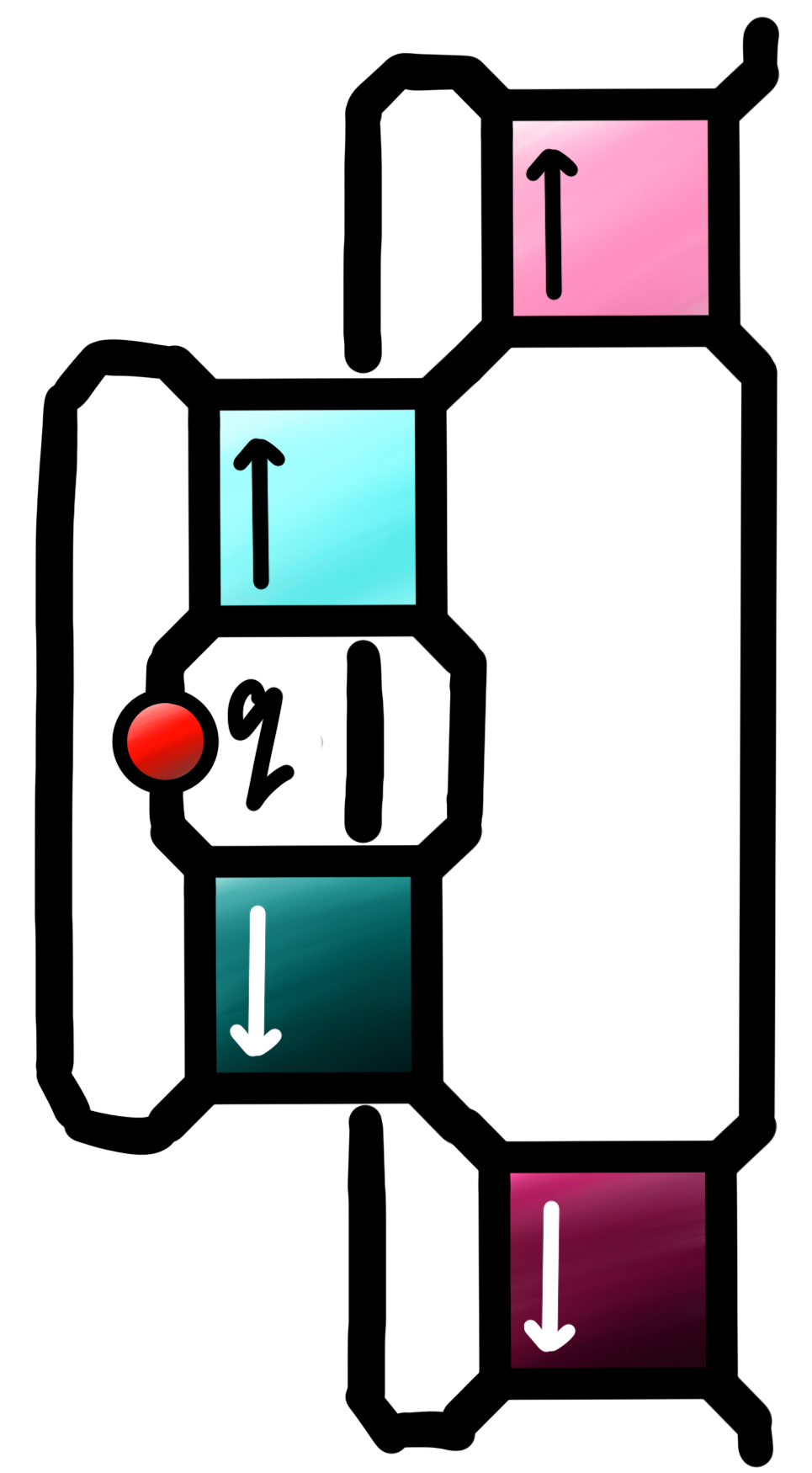},
\end{equation}
Similarly, the width-2 terms $\bm{q^{\prime}}$ could be calculated like so,
\begin{align}
    \bm{q^{\prime}} &= \frac{1}{d^{|\mathcal{R}_+|}}\tr_{\mathcal{R}_+}\left(\mathbb{U}\bm{q}_x\mathbb{U}^{\dagger}\right) - \left(\mathds{1} \otimes \mathcal{M}_{+,1}\left(\bm{q}\right) \right), \\
    &= \frac{1}{d}\includegraphics[width=0.15\linewidth, valign=c]{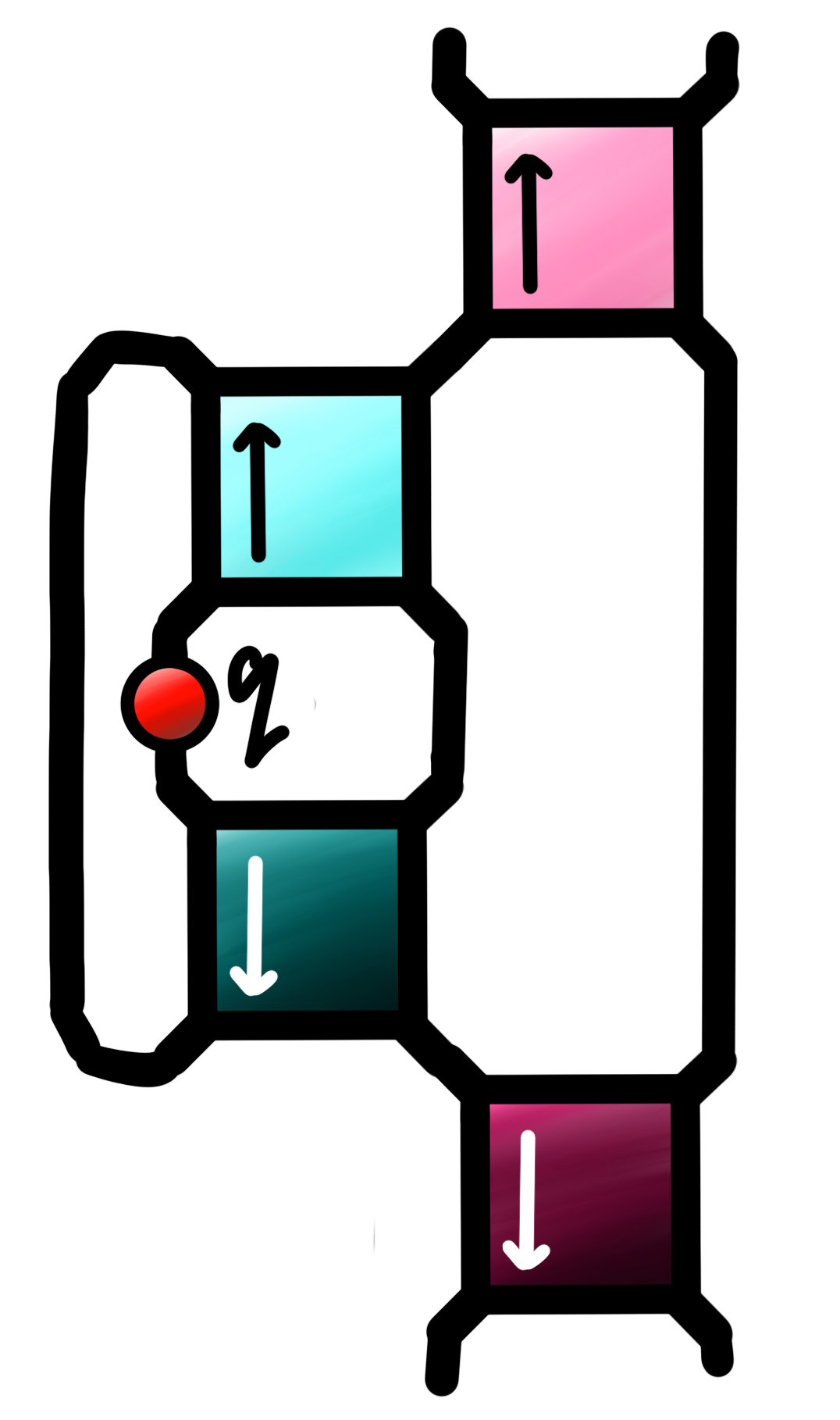} - \Bigg(\mathds{1} \otimes \frac{1}{d^2}\includegraphics[width=0.15\linewidth, valign=c]{M1plusw.png}\Bigg),
\end{align}
where $\mathcal{R}_+ \subseteq \mathbb{Z}_{2L}$ is the complement to the interval over which $\bm{q^{\prime}}$ is supported, i.e. $\mathcal{R}_+ = \{ x \in \mathbb{Z}_{2L} | x \notin [x+1, x+2]\}$. Finally, the width-4 terms $\bm{q^{\prime\prime}}$ can be calculated like so
\begin{align}
    \bm{q^{\prime\prime}} &= \frac{1}{d^{|\mathcal{R}_+|-2}}\tr_{\mathcal{R}_+/\{x-1,x\}}\left(\mathbb{U}\bm{q}_{x}\mathbb{U}^{\dagger}\right) - \left(\mathds{1}^{\otimes 2} \otimes \bm{q^{\prime}} \right) - \left(\mathds{1}^{\otimes 3} \otimes \mathcal{M}_{+,1}\left(\bm{q}\right) \right),\\
    &= \includegraphics[width=0.15\linewidth, valign=c]{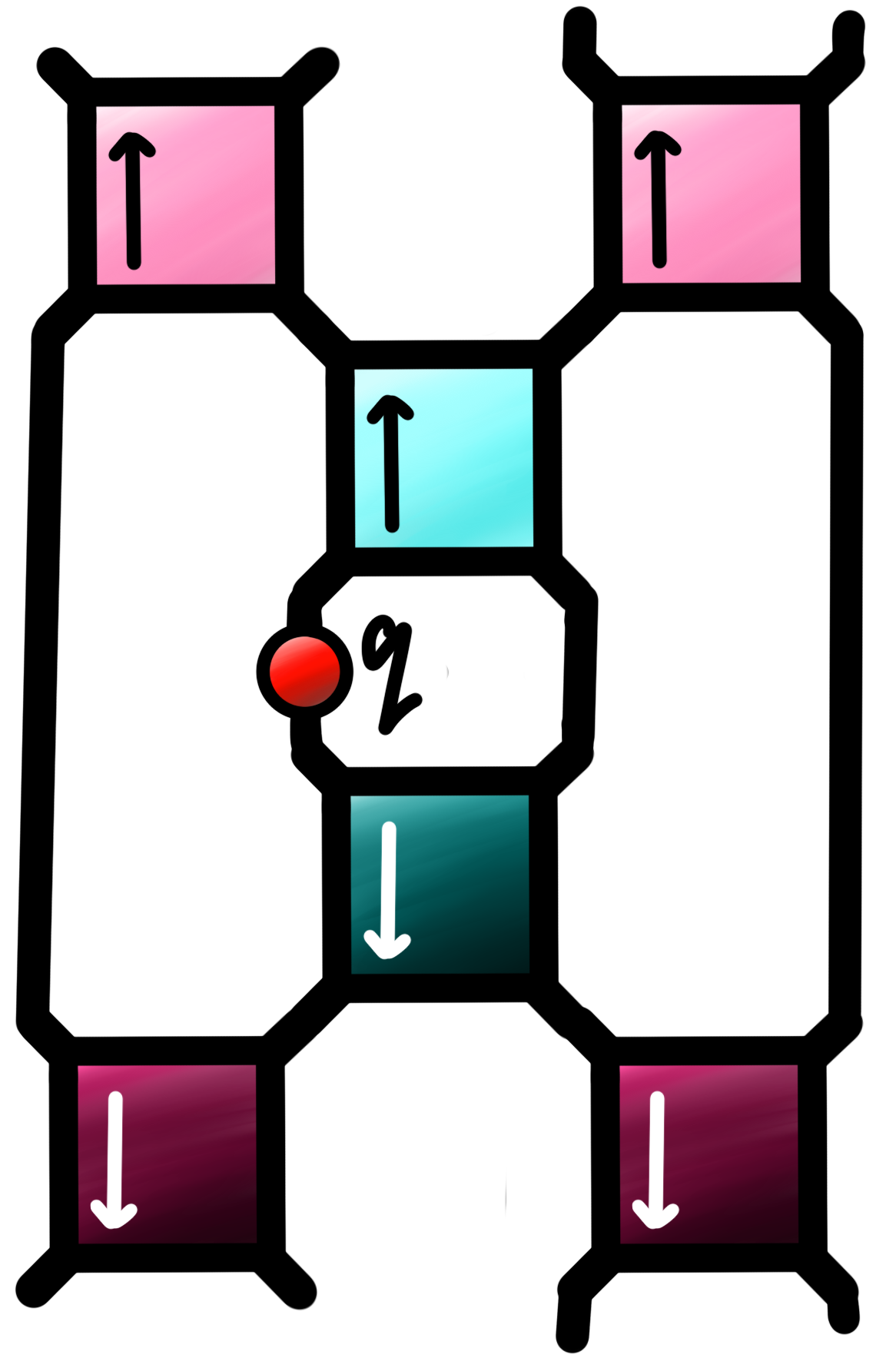} - \Bigg(\mathds{1}^{\otimes2} \otimes \frac{1}{d}\includegraphics[width=0.15\linewidth, valign=c]{1bodyqp.png}\Bigg).
\end{align}
The important lesson to be learned from Eqn.~\ref{eqn:1bodyunderfloquetDU} and the associated diagram is that Property 1 significantly constraints the way in which the support of local operators can change under brickwork Floquet evolution. One way to see this is to consider some product basis $\{\bm{A_{\alpha\beta\gamma\delta}} = \bm{a_{\alpha}} \otimes \bm{a_{\beta}} \otimes \bm{a_{\gamma}} \otimes \bm{a_{\delta}}\}_{\alpha,\beta,\gamma,\delta = 0}^{d^2-1}$ for the 4-site Hilbert space associated to the sites $\{x-1,\;x,\;x+1,\;x+2\}$, with $\{\bm{a_\alpha}\}_{\alpha = 0}^{d^2-1}$ an operator basis for the single site Hilbert space satisfying $\langle \bm{a_{\alpha}}, \bm{a_{\beta}} \rangle = \tr\left(\bm{a_{\alpha}}^{\dagger}\bm{a_{\beta}}\right) = d\delta_{\alpha,\beta}$ and $\bm{a_0} = \bm{\mathds{1}}$. Expanding $\mathbb{U}\bm{q}_x\mathbb{U}^{\dagger}$ in this basis,
\begin{equation}
    \mathbb{U}\bm{q}_x\mathbb{U}^{\dagger} = \sum_{\alpha,\beta,\gamma,\delta = 0}^{d^2-1} c_{\alpha\beta\gamma\delta}\left[\bm{A_{\alpha\beta\gamma\delta}}\right]_{x-1,...,x+2},
\end{equation}
where
\begin{equation}
    c_{\alpha\beta\gamma\delta} = \frac{1}{d^{2L}} \tr\left(\left[\bm{A_{\alpha\beta\gamma\delta}}^{\dagger}\right]_{x-1,...,x+2}\mathbb{U}\bm{q}_x\mathbb{U}^{\dagger}\right),
\end{equation}
all $d^8-1$ of the coefficients $c_{\alpha\beta\gamma\delta}$ (excluding $c_{0000}$, assuming $\bm{q}$ to be traceless) would be non-zero for a generic choice of $\bm{q}$ and a generic choice of $U, V \in U(d^2)$ generating the Floquet operator $\mathbb{U}$. Yet from Eqn.~\ref{eqn:1bodyunderfloquetDU}, we can deduce that enforcing dual-unitarity on $\mathbb{U}$ means that $2d^6-2d^4+d^2-1$ of the coefficients must be zero\footnote{Consult Appendix \ref{sec:dimreductioncalc} for clarification on how this can be determined.}. For small $d$ this can cause quite a significant effective reduction in the Hilbert space; $99$ of the $255$ coefficients are zero for $d=2$, for instance. 
\par
Another (perhaps more insightful) way to see the constraints, however, is directly in the reduced picture illustrated diagrammatically above. The notation used in the diagrams is tailored to focus on the spatial distribution of non-trivial support, with any further details about the specific operator content ignored. This viewpoint is interesting, however, as it highlights the nature of the constraints clearly, even as they vanish asymptotically (in terms of the number of coefficients they constrain to be zero relative to the total Hilbert space dimension) in the limit of a large local Hilbert space dimension, $d\rightarrow\infty$. Considering the $4$ sites $\{x-2,\;x-1,\;x,\;x+1\}$ and the $10$ different ways to place two traceless endpoints over them\footnote{This a stars and bars problem, with $n=2$ stars and $k=4$ bins (or $k-1 = 3$ bars); the standard stars and bars formula $\begin{pmatrix} n + k - 1 \\ n \end{pmatrix}$ gives us $5!/(2!3!) = 10$ combinations.}, all $10$ of these combinations would be produced with some non-zero probability when evolving a generic local operator $\bm{q}_x$ under a generic brickwork Floquet operator, $\mathbb{U}\bm{q}_x\mathbb{U}^{\dagger}$; we see in Eqn.~\ref{eqn:1bodyunderfloquetDUdiag}, however, that enforcing dual-unitarity on $\mathbb{U}$ means we can only generate 3 such combinations.
\par
The procedure outlined in the examples above can be generalised to the case of an initial local operator of arbitrary width $w$, giving us a more comprehensive characterisation of the constraints enforced by Property \ref{property1} on the dynamics of local operators in dual-unitary circuits. The resulting diagrams are shown in Figure \ref{spreadingdiagrams}. The parity of $x$ and $w$ for a given local operator $\bm{q}_{x,...,x+w-1}$ affects how the support can change, so different diagrams are drawn for each of the possible combinations of these parities.

\begin{figure*}
    \begin{subfigure}{.45\textwidth}
\captionsetup{justification=centering}
      \centering
      \includegraphics[width=0.9\linewidth]{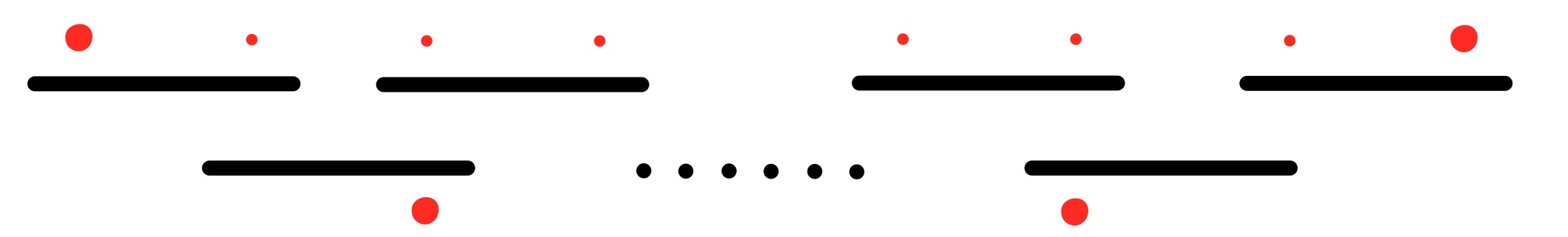}
      \caption{$x$ odd, $w$ even.}
      \label{oddeventerms}
    \end{subfigure}%
    \begin{subfigure}{.45\textwidth}
      \centering
      \includegraphics[width=0.9\linewidth]{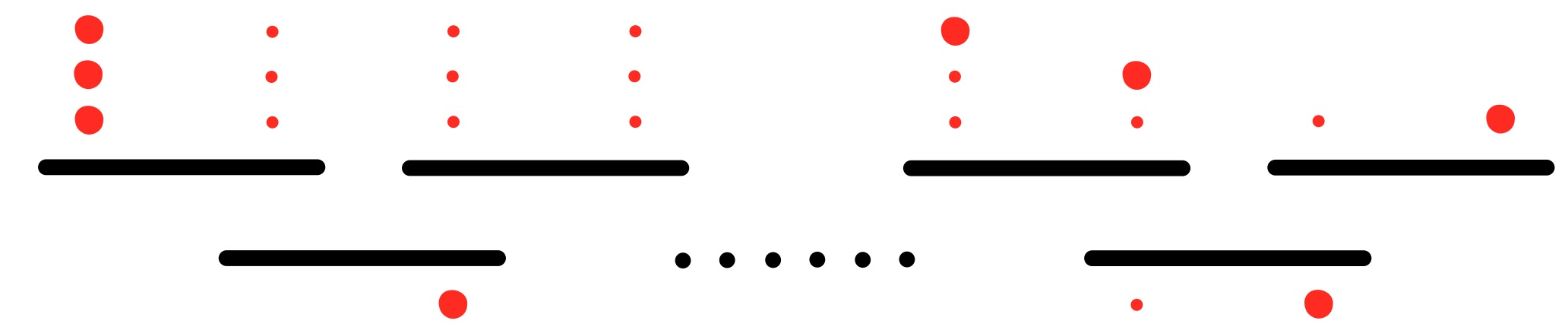}
      \caption{$x$ odd, $w$ odd.}
      \label{oddoddterms}
    \end{subfigure}
    \begin{subfigure}{.45\textwidth}
        \centering
        \includegraphics[width=0.9\linewidth]{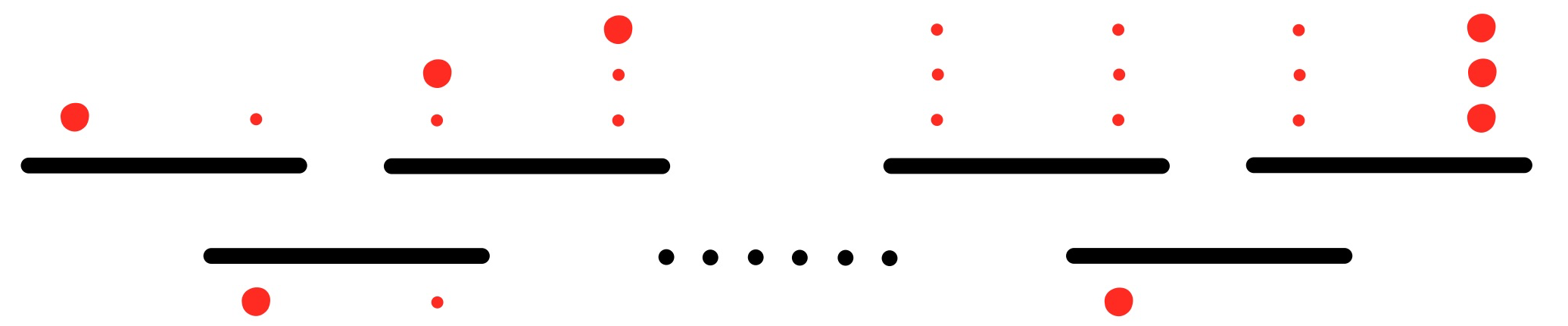}
        \caption{$x$ even, $w$ odd.}
        \label{evenoddterms}
      \end{subfigure}%
      \begin{subfigure}{.45\textwidth}
        \centering
        \includegraphics[width=0.9\linewidth]{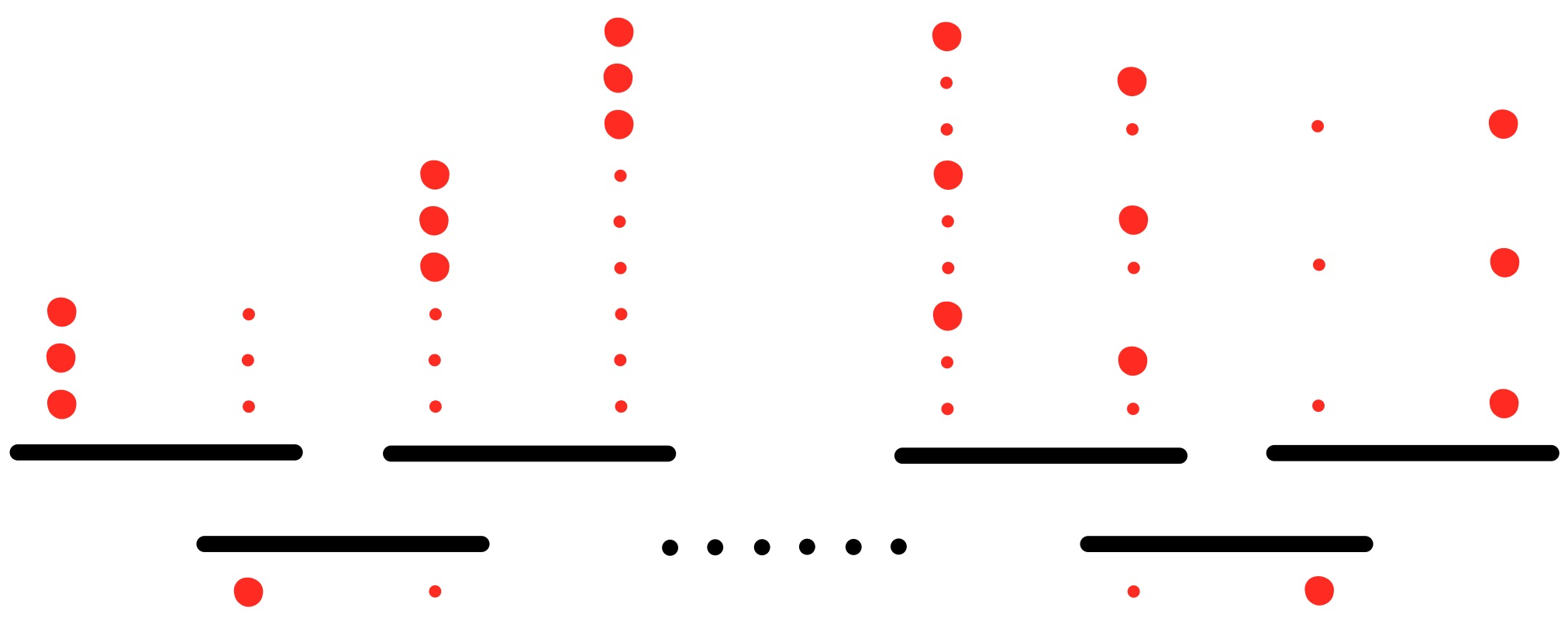}
        \caption{$x$ even, $w$ even.}
        \label{eveneventerms}
      \end{subfigure}
      \begin{subfigure}{.45\textwidth}
        \centering
        \includegraphics[width=0.7\linewidth]{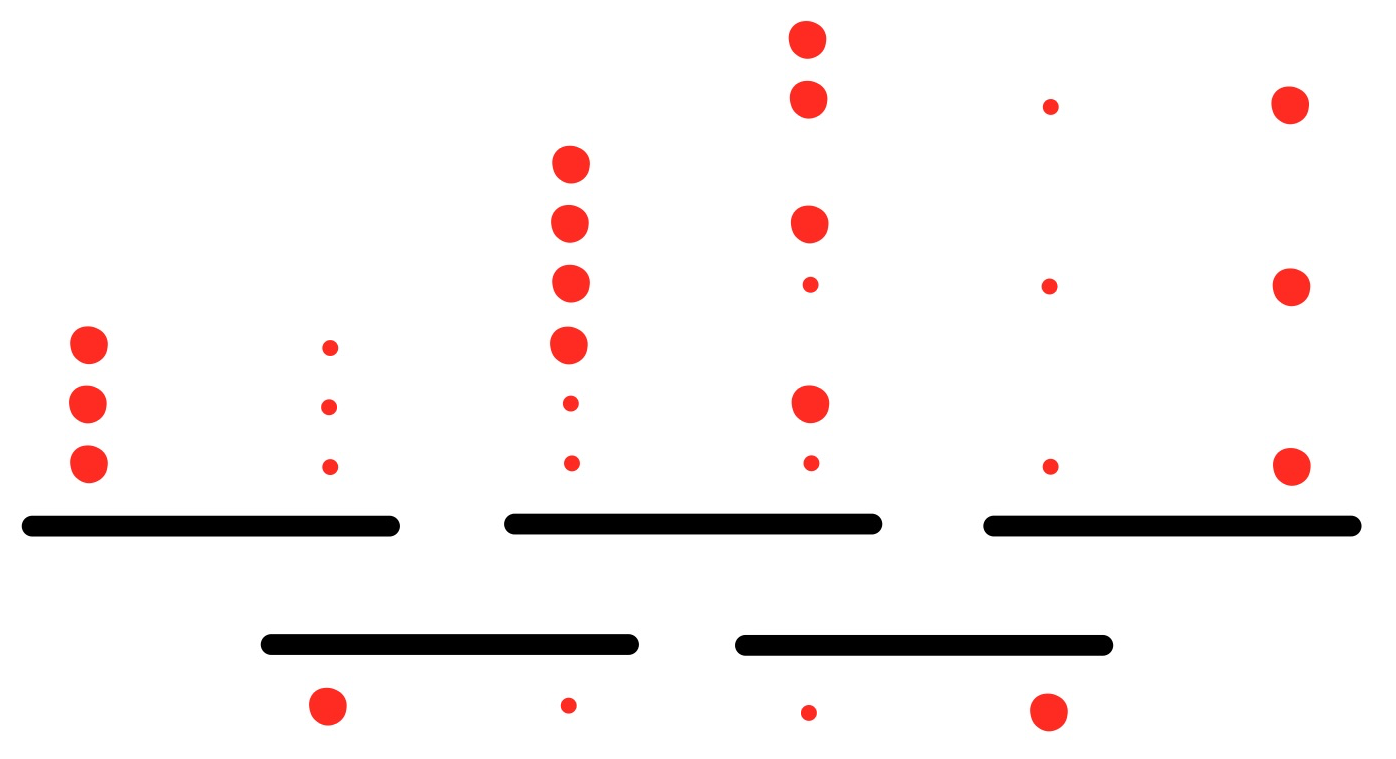}
        \caption{$x$ even, $w=4$.}
        \label{evenevenwidth4}
      \end{subfigure}%
      \begin{subfigure}{.45\textwidth}
        \centering
        \includegraphics[width=0.5\linewidth]{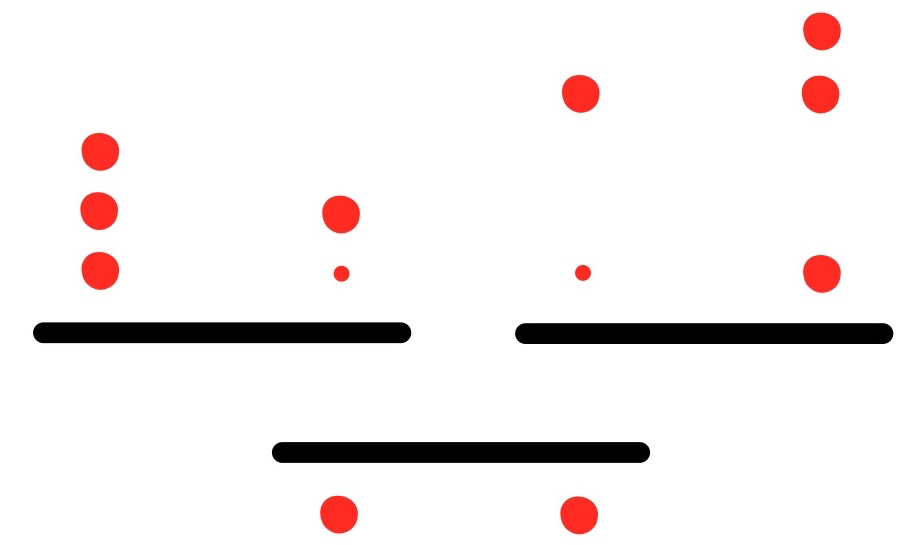}
        \caption{$x$ even, $w=2$.}
        \label{evenevenwidth2}
      \end{subfigure}
\captionsetup{justification=raggedright,singlelinecheck=false}
    \caption{Diagrams showing how the support of finite-range operators acting on an interval of $w$ neighbouring sites can change under conjugation by a dual-unitary Floquet operator $\mathbb{U}$ (as defined in Eqn.~\ref{eqn:floquetoperator}).
    The parity of $x$ (the left-hand most site of non-trivial support) and $w$ turns out to be relevant, due to the restrictions on how the support of an operator can change under dual-unitary dynamics (see Property \ref{property1}), and so different diagrams are shown for the different possible combinations of these two parities. It is necessary to consider two small-$w$ edge cases ($w=4$ and $w=2$) for the $x$-even, $w$-even terms; in all other cases, the value of $w$ does not matter.}
    \label{spreadingdiagrams}
    \end{figure*}

From these diagrams, it is then possible to determine the possible mappings generated by the adjoint action of a dual-unitary Floquet operator between subspaces of the space of local operators $\bm{q}_{x,...,x+w-1}$ defined by the parity of $x$ and $w$. These mappings are summarised in the digraph shown in Fig.~\ref{fig:digraph}. We see in the digraph, as in the example diagrams given above, evidence of constraints on the dynamics of local operators enforced by dual-unitarity; for generic unitary dynamics, the digraph would be complete, with every subspace joined to every other subspace by bi-directional edges. In Fig.~\ref{fig:digraph}, however, we see that the digraph generated by dual-unitary dynamics is oriented (there are no bi-directional edges), and not complete as there is no edge between the even-$x$, odd-$w$ subspace and the odd-$x$, odd-$w$ subspace. We reiterate that this is a direct consequence of Property \ref{property1} - which holds for all dual-unitary gates - and hence applies to all dual-unitary circuits.
\begin{figure}
    \includegraphics[width=0.9\linewidth]{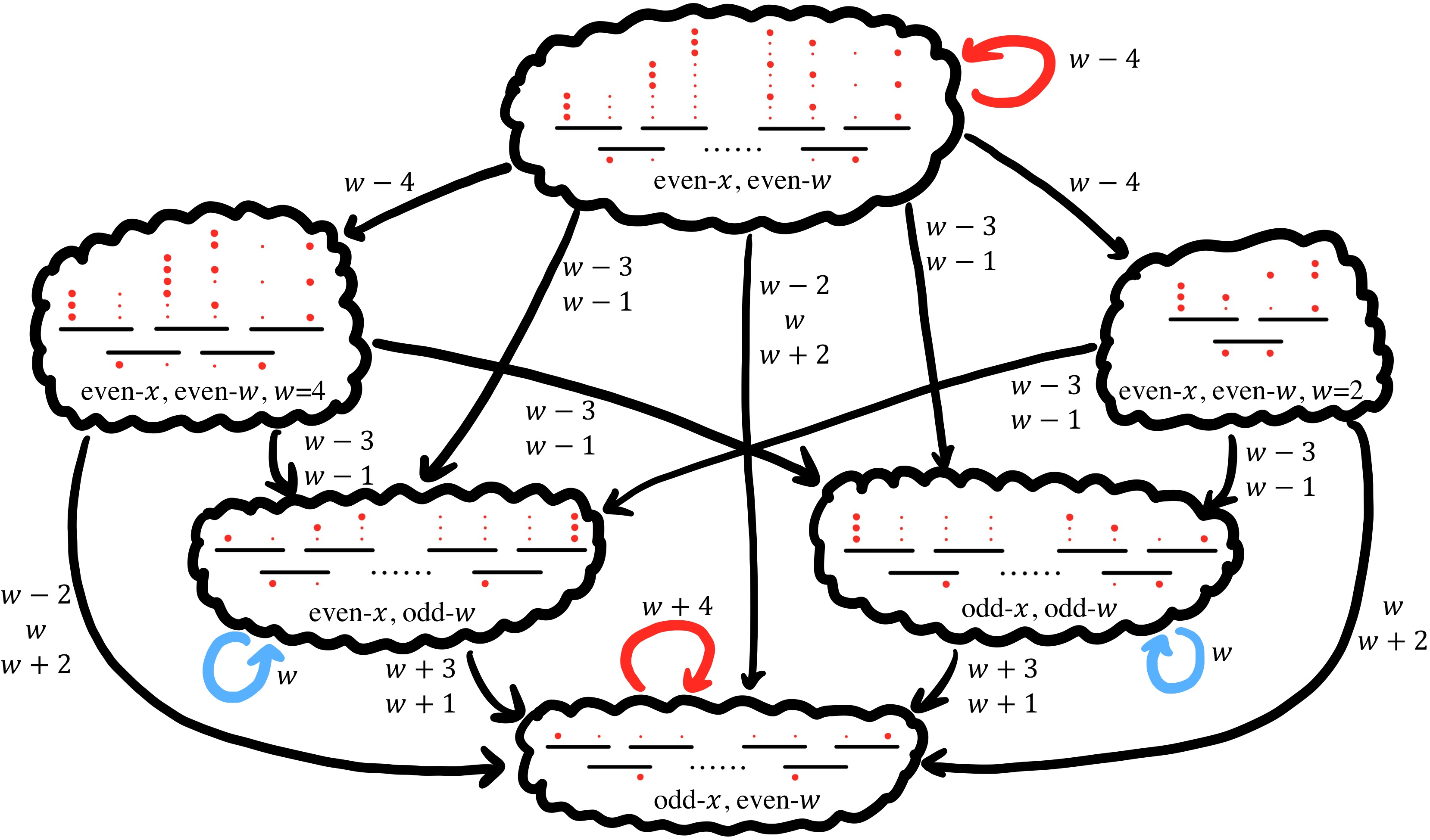}
\captionsetup{justification=raggedright,singlelinecheck=false}
    \caption{When we partition the operator space into subspaces characterised by the parity of $x$ and $w$, for local operators $\bm{q}_{x,...,x+w-1}$, the possible mappings between the subspaces that can be generated by a dual-unitary brickwork Floquet operator $\mathbb{U}$ can be represented as the digraph above. The subspaces are the nodes of the digraph, and the directed edges mappings between them under unitary conjugation by $\mathbb{U}$. The labels on each directed edge denote the width of terms in the subspace at the tip of the edge that can be generated from terms of width $w$ in the subspace at the base of the edge.}
    \label{fig:digraph}
\end{figure}
\par
From the digraph, it is possible to deduce that only operators for which the width $w$ is odd can contribute to conserved quantities in the circuit. This is formalised in the following Lemma.
\begin{lemma}\label{lem:onlyoddw}
    \textit{Only odd-width terms}.
    Let
    \begin{equation}
        \mathcal{Q} = \sum_{x\in\mathbb{Z}_{2L}}\sum_{w\in\mathbb{Z}_{L-3}/0} \bm{q}(x,w)_{x,...,x+w-1}
    \end{equation}
    with $\bm{q}(x,w)\in\bar{\mathcal{A}}_w$, and let $\mathbb{U}$ be a 1+1D brickwork dual-unitary Floquet operator (Eqn.~\ref{eqn:floquetoperator}). Then
    \begin{equation}
        \mathbb{U}\mathcal{Q}\mathbb{U}^{\dagger} = \mathcal{Q} \implies \bm{q}(x,w) = 0 \;  \forall \; w \in \mathbb{E}_{L-3}/0.
    \end{equation}
\end{lemma}
\begin{proof}
The full proof of this Lemma can be found in Appendix \ref{appendix:fate}. Here, we provide a sketch of the proof, as follows: We see that the digraph representing the way in which the support of a local operator $\bm{q}_{x,...,x+w-1}$ can change under the adjoint action of $\mathbb{U}$ (Figure \ref{fig:digraph}) contains no loops, save for the mappings from the odd-$w$ subspaces back into themselves. In order for a given set of local operators to contribute to a conserved quantity we would require such a loop involving that set of operators; we can hence deduce that these odd-$w$ operators and the aforementioned mappings back into their own subspaces - which, it turns out, are precisely the mappings that the maximal velocity solitons obey - constitute the only possible contributions to conserved quantities in the circuit.
\end{proof}

\subsection{Fundamental charges for dual-unitary circuits}\label{sec:fundamentalcharges}

In Section \ref{sec:dynamicalconstraints} above, we showed that Property \ref{property1} significantly constrains the dynamics of local operators in dual-unitary circuits, and that a consequence of this is that only operators that evolve solitonically - at least in the sense in which their support changes under the dual-unitary dynamics - can contribute to conserved quantities in the circuit. We are now ready to build on this and prove the main result of the paper, Theorem \ref{thm:fundamentalcharges}, establishing a deep connection between solitons and conserved quantities in dual-unitary circuits.

\begin{theorem}\label{thm:fundamentalcharges}
    \textit{Fundamental charges for DUCs}. Consider a brickwork dual-unitary Floquet operator, $\mathbb{U}$, as defined in Eqn.~\ref{eqn:floquetoperator}, acting on a chain of $2L$ qudits. Let $\mathcal{Q}$ be a conserved quantity under $\mathbb{U}$ that is formed from a linear combination of finite-range operators that are each supported non-trivially on at most $L-4$ consecutive sites, i.e.
    \begin{equation}\label{eqn:thm1Q}
        \mathbb{U}\mathcal{Q}\mathbb{U}^{\dagger} = \mathcal{Q} = \sum_{w \in \mathbb{Z}_{L-3}/0}\sum_{x\in\mathbb{Z}_{2L}} \bm{q}(x,w)_{x,...,x+w-1},
    \end{equation}
    where $\bm{q}(x,w) \in \bar{\mathcal{A}}_w$ is an operator that potentially depends on $x$ and $w$. Any such $\mathcal{Q}$ can always be written as
    \begin{equation}
        \mathcal{Q} = \sum_{w \in \mathbb{O}_{L-3}} \left( \sum_{k: \; \bm{a_k}\in \mathbb{S}_{+,w}} \alpha_{k,w}Q^{+}_{k,w} + \sum_{l: \; \bm{b_l}\in \mathbb{S}_{-,w}} \beta_{l,w}Q^{-}_{l,w}\right), 
    \end{equation}
    where the sets $\mathbb{S}_{+,w}$, defined in Eqn.~\ref{eqn:set_S}, and $\mathbb{S}_{-,w}$, defined in Eqn.~\ref{eqn:set_T}, contain the left- and right-moving solitons of $\mathbb{U}$ respectively, and where $Q^{+}_{k,w}$ and $Q^{-}_{l,w}$ are the conserved quantities associated to these solitons, as defined in Eqns.~\ref{eqn:Q_s}, and \ref{eqn:Q_t}.
    \end{theorem}
    \begin{proof}
        By Lemma \ref{lem:onlyoddw}, we know that $\mathcal{Q}$ can only contain terms of odd width,
        \begin{equation}
            \mathcal{Q} = \sum_{w \in \mathbb{O}_{L-3}}\sum_{x\in\mathbb{Z}_{2L}} \bm{q}(x,w)_{x,...,x+w-1}.
        \end{equation}
        Let us consider the terms with even $x$ and a fixed width $w$,
        \begin{equation}
            \mathcal{Q}^{+}_{w} = \sum_{x\in\mathbb{E}_{2L}} \bm{q}(x,w)_{x,...,x+w-1}, \quad w \in \mathbb{O}_{L-3}.
        \end{equation}
        As summarised diagrammatically in Fig.~\ref{evenoddterms}, Property \ref{property1} enforces that such terms are restricted under the adjoint action of $\mathbb{U}$ to evolving into  a linear combination of operators of width $w$, $w+1$, and $w+3$, supported over the regions $\left[x+2, x+w+1\right]$, $\left[x+1, x+w+1\right]$, and $\left[x-1, x+w+1\right]$ respectively. Following this, we note that the width-$w$ component of this linear combination, supported on the interval $[x+2, x+w+1]$, is given by $\mathcal{M}_{+,w}\left(\bm{q}(x,w)\right)$. This allows us to then write
        \begin{equation}
            \mathbb{U}\mathcal{Q}^{+}_{w}\mathbb{U}^{\dagger} = \sum_{x \in \mathbb{E}_{2L}} \left[\mathcal{M}_{+,w}\left(\bm{q}(x,w)\right)\right]_{x+2,...,x+w+1} + \bm{q^{\prime}}(x,w)_{x+1,...,x+w+1} + \bm{q^{\prime\prime}}(x,w)_{x-1,...,x+w+1},
        \end{equation}
        where, to be explicit, we have defined
        \begin{equation}
            \bm{q^{\prime}}(x,w) = \frac{1}{d^{|\mathcal{R}_+|}}\tr_{\mathcal{R}_+}\left(\mathbb{U}\bm{q}(x,w)_{x,...,x+w-1}\mathbb{U}^{\dagger}\right) - \left(\mathds{1} \otimes \mathcal{M}_{+,w}\left(\bm{q}(x,w)\right) \right)\in \bar{\mathcal{A}}_{w+1},
        \end{equation}
        and 
        \begin{multline}
            \bm{q^{\prime\prime}}(x,w) = \frac{1}{d^{|\mathcal{R}_+|-2}}\tr_{\mathcal{R}_+/\{x-1,x\}}\left(\mathbb{U}\bm{q}(x,w)_{x,...,x+w-1}\mathbb{U}^{\dagger}\right) \\
            - \left(\mathds{1}^{\otimes 2} \otimes \bm{q^{\prime}}(x,w) \right) - \left(\mathds{1}^{\otimes 3} \otimes \mathcal{M}_{+,w}\left(\bm{q}(x,w)\right) \right)\in \bar{\mathcal{A}}_{w+3},
        \end{multline}
        where $\mathcal{R}_+ \subseteq \mathbb{Z}_{2L}$ is the complement to the interval over which $\bm{q^{\prime}}(x,w)$ is supported. We have defined $\bm{q^{\prime}}(x,w)$ as the width-$(w+1)$ component of $\mathbb{U}\bm{q}(x,w)_{x,...,x+w-1}\mathbb{U}^{\dagger}$ supported on the interval $[x+1, x+w+1]$, and $\bm{q^{\prime\prime}}(x,w)$ as the width-$(w+3)$ component of $\mathbb{U}\bm{q}(x,w)_{x,...,x+w-1}\mathbb{U}^{\dagger}$ supported on the interval $[x-1, x+w+1]$. Diagrammatically, these calculations could be represented exactly as outlined for the width-1 $\bm{q}_x$ example given earlier in Section \ref{sec:dynamicalconstraints}, except now with the relatively straightforward generalisation to arbitrary width $w$; $\mathcal{M}_{+,1}$ would simply be replaced by $\mathcal{M}_{+,w}$ in all cases.
        \par
        Now consider the terms with odd $x$,
        \begin{equation}
            \mathcal{Q}^{-}_{w} = \sum_{x\in\mathbb{O}_{2L}} \bm{q}(x,w)_{x,...,x+w-1}, \quad w \in \mathbb{O}_{L-3}.
        \end{equation}
        From Fig.~\ref{oddoddterms}, it can be deduced that each $\bm{q}_{x,...,x+w-1}$ with $x$ odd and $w$ odd will be mapped by $\mathbb{U}$ into a linear combination of operators of width $w$, $w+1$, and $w+3$, supported over the regions $\left[x-2, x+w-3\right]$, $\left[x-2, x+w-2\right]$, and $\left[x-2, x+w\right]$ respectively. Similarly to before, the width-$w$ component of this linear combination, supported on the interval $[x-2, x+w-3]$, is given by $\mathcal{M}_{-,w}\left(\bm{q}(x,w)\right)$. Hence, we can write
        \begin{equation}\label{eqn:Qpluswevol}
            \mathbb{U}\mathcal{Q}^{-}_{w}\mathbb{U}^{\dagger} = \sum_{x \in \mathbb{O}_{2L}} \left[\mathcal{M}_{-,w}\left(\bm{q}(x,w)\right)\right]_{x-2,...,x+w-3} + \bm{q^{\prime\prime\prime}}(x,w)_{x-2,...,x+w-2} + \bm{q^{\prime\prime\prime\prime}}(x,w)_{x-2,...,x+w}.
        \end{equation}
        where
        \begin{equation}
            \bm{q^{\prime\prime\prime}}(x,w) = \frac{1}{d^{|\mathcal{R}_-|}}\tr_{\mathcal{R}_-}(\mathbb{U}\bm{q}(x,w)_{x,...,x+w-1}\mathbb{U}^{\dagger}) - \left(\mathcal{M}_{-,w}\left(\bm{q}(x,w)\right)\otimes \mathds{1}\right)\in \bar{\mathcal{A}}_{w+1},
        \end{equation}
        and
        \begin{multline}
            \bm{q^{\prime\prime\prime\prime}}(x,w) = \frac{1}{d^{|\mathcal{R}_-|-2}}\tr_{\mathcal{R}_-/\{x+w-1,x+w\}}(\mathbb{U}\bm{q}(x,w)_{x,...,x+w-1}\mathbb{U}^{\dagger}) \\ - \left(\bm{q^{\prime\prime\prime}}(x,w) \otimes \mathds{1}^{\otimes 2} \right) - \left(\mathcal{M}_{-,w}\left(\bm{q}(x,w)\right) \otimes \mathds{1}^{\otimes 3} \right)\in \bar{\mathcal{A}}_{w+3},
        \end{multline}
        with $\mathcal{R}_- \subseteq \mathbb{Z}_{2L}$ the complement to the interval over which $\bm{q^{\prime\prime\prime}}(x,w)$ is supported. We have defined $\bm{q^{\prime\prime\prime}}(x,w)$ as the width-$(w+1)$ component of $\mathbb{U}\bm{q}(x,w)_{x,...,x+w-1}\mathbb{U}^{\dagger}$ supported on the interval $[x-2, x+w-2]$, and $\bm{q^{\prime\prime\prime\prime}}(x,w)$ as the width-$(w+3)$ components of $\mathbb{U}\bm{q}(x,w)_{x,...,x+w-1}\mathbb{U}^{\dagger}$ supported on the interval $[x-2, x+w]$.
        \par
        As $\mathcal{Q} = \sum_{w \in \mathbb{O}_{L-3}} Q_w^+ + Q_w^-$, we have
        \begin{align}
            \begin{split}
            \mathbb{U}\mathcal{Q}\mathbb{U}^{\dagger} = \sum_{w \in \mathbb{O}_{L-3}} \sum_{x \in \mathbb{E}_{2L}} &\left[\mathcal{M}_{+,w}\left(\bm{q}(x,w)\right)\right]_{x+2,...,x+w+1} + \bm{q^{\prime}}(x,w)_{x+1,...,x+w+1} \\ 
            + \; &\bm{q^{\prime\prime}}(x,w)_{x-1,...,x+w+1} + \left[\mathcal{M}_{-,w}\left(\bm{q}(x+1,w)\right)\right]_{x-1,...,x+w-2} \\
            + \; &\bm{q^{\prime\prime\prime}}(x+1,w)_{x-1,...,x+w-1} + \bm{q^{\prime\prime\prime\prime}}(x+1,w)_{x-1,...,x+w+1}.
            \end{split}  
        \end{align}
        Note that the term $\left[\mathcal{M}_{+,w}\left(\bm{q}(x,w)\right)\right]_{x+2,...,x+w+1}$ will be the only term in $\mathbb{U}\mathcal{Q}\mathbb{U}^{\dagger}$ supported exclusively over the interval $\left[x+2,x+w+1\right]$. Similarly, $\left[\mathcal{M}_{-,w}\left(\bm{q}(x+1,w)\right)\right]_{x-1,...,x+w-2}$ will be the only term supported exclusively over the interval $\left[x-1,x+w-2\right]$. As $\mathbb{U}\mathcal{Q}\mathbb{U}^{\dagger} = \mathcal{Q}$, we can hence relate these to the terms supported over the same intervals in $\mathcal{Q}$:
        \begin{equation}
            \mathcal{M}_{+,w}\left(\bm{q}(x,w)\right) = \bm{q}(x+2,w),
        \end{equation}
        for even $x$, and
        \begin{equation}
            \mathcal{M}_{-,w}\left(\bm{q}(x,w)\right) = \bm{q}(x-2,w),
        \end{equation}
        for odd $x$. Summing over even $x$, we have
        \begin{equation}
            \sum_{x \in \mathbb{E}_{2L}} \mathcal{M}_{+,w}\left(\bm{q}(x,w)\right) = \sum_{x \in \mathbb{E}_{2L}} \bm{q}(x+2,w) = \mathcal{Q}_w^+,
        \end{equation}
        and hence we see that Eqn.~\ref{eqn:Qpluswevol} becomes
        \begin{equation}
            \mathbb{U}\mathcal{Q}^{+}_{w}\mathbb{U}^{\dagger} = \mathcal{Q}_w^+ + \sum_{x \in \mathbb{E}_{2L}} \bm{q^{\prime}}(x,w)_{x+1,...,x+w+1} + \bm{q^{\prime\prime}}(x,w)_{x-1,...,x+w+1}.
        \end{equation}
        Taking a scalar product of both sides, $\|A\|^2 = \tr(A^{\dagger}A)$, and noting that by unitarity $\|\mathbb{U}\mathcal{Q}^{+}_{w}\mathbb{U}^{\dagger}\|^2 = \|\mathcal{Q}^{+}_{w}\|^2$, we obtain
        \begin{equation}
            0 = \sum_{x \in \mathbb{E}_{2L}} \|\bm{q^{\prime}}(x,w)_{x+1,...,x+w+1}\|^2 + \|\bm{q^{\prime\prime}}(x,w)_{x-1,...,x+w+1}\|^2.
        \end{equation}
        where as all the terms on the right-hand side are supported over unique intervals of the chain, they will all be orthogonal to each other and hence all cross terms will go to zero (i.e.~the scalar product can be taken linearly). This then implies
        \begin{equation}
            \bm{q^{\prime}}(x,w) = \bm{q^{\prime\prime}}(x,w) = 0, \quad \forall \; (x,w) \in \mathbb{E}_{2L} \times \mathbb{O}_{L-3},
        \end{equation}
        where we have noted that as our choice of $w$ was arbitrary to start with, this must be true for all the possible values of $w$ in $\mathcal{Q}$. Analogously for the odd $x$ terms, we find
        \begin{equation}
            \bm{q^{\prime\prime\prime}}(x,w) = \bm{q^{\prime\prime\prime\prime}}(x,w) = 0, \quad \forall \; (x,w) \in \mathbb{O}_{2L} \times \mathbb{O}_{L-3}.
        \end{equation}
        \par
        We have shown that all the $\bm{q}(x,w)$ terms of $\mathcal{Q}$ are mapped exclusively into terms of the same width under conjugation by $\mathbb{U}$. It turns out that we can always decompose such terms into a basis formed by the width-$w$ solitons of the $\mathcal{M}_{\nu,w}$ map, where $\nu = +$ if $x$ is even or $\nu = -$ if $x$ is odd - we establish this fact rigorously in Appendix \ref{appendix:solitonbasis} (see Lemma \ref{lem:solitonbasis}). We can hence write
        \begin{equation}
            \bm{q}(x=0, w) = \sum_{\bm{a_k} \in \mathbb{S}_{+,w}} \alpha_{k,w}\bm{a_k},
        \end{equation}
        where $\bm{a_k}$ are the width-$w$ solitons of the $\mathcal{M}_{+,w}$ map, defined via the set $\mathbb{S}_{+,w}$ in Eqn.~\ref{eqn:set_S}, and $\alpha_{k,w}$ are a set of normalised complex coefficients.
        \par
        As $\bm{q}^{\prime} = \bm{q}^{\prime\prime} = 0$ and $\mathcal{M}_{+,w}\left(\bm{q}(x,w)\right) = \bm{q}(x+2,w)$ for all even $x$, we have
        \begin{equation}
            \mathbb{U}\bm{q}(x,w)_{x,...,x+w-1}\mathbb{U}^{\dagger} = \bm{q}(x+2, w)_{x+2,...,x+w+1} \quad \forall \; x \in \mathbb{E}_{2L},
        \end{equation}
        which gives us a kind of recurrence relation,
        \begin{equation}
            \mathbb{U}^n\bm{q}(x=0, w)_{0,...,w-1}\left(\mathbb{U}^{\dagger}\right)^n = \bm{q}(x=2n, w)_{2n,...,2n+w-1}.
        \end{equation}
        As $\bm{q}(x=0, w)$ is formed solely of solitons, we can easily calculate the left-hand side of the above,
        \begin{equation}
            \sum_{\bm{a_k} \in \mathbb{S}_{+,w}} \alpha_{k,w} \left(\lambda^{+}_{k,w}\right)^{n}\left[\bm{a_k}\right]_{2n,...,2n+w-1} = \bm{q}(x=2n, w)_{2n,...,2n+w-1},
        \end{equation}
        where again we have used square brackets to help separate the basis labelling and spatial position indices. This then gives us a general expression for $\bm{q}(x, w)$ for all even $x$,
        \begin{equation}
            \bm{q}(x,w) = \sum_{\bm{a_k} \in \mathbb{S}_{+,w}} \alpha_{k,w} \left(\lambda^{+}_{k,w}\right)^{\frac{x}{2}}\bm{a_k}, \quad \forall \; x\in\mathbb{E}_{2L}, \quad \textrm{and} \quad \forall \; w \in \mathbb{O}_{L-3}.
        \end{equation}
        Upon substitituing this into our expression for $\mathcal{Q}^{+}_{w}$ and recalling our definition for the conserved quantities $Q^{+}_{k,w}$ associated to each soliton $\bm{a_k}$ (as given in Eqn.~\ref{eqn:Q_s}), we obtain
        \begin{equation}
            \mathcal{Q}^{+}_{w} = \sum_{x \in \mathbb{E}_{2L}} \bm{q}(x,w)_{x,...,x+w-1} = \sum_{k: \; \bm{a_k}\in \mathbb{S}_{+,w}} \alpha_{k,w} Q^{+}_{k,w},
        \end{equation}
        which means that by summing over $w$ we get
        \begin{equation}
            \mathcal{Q}^{+} = \sum_{w \in \mathbb{O}_{L-3}} \; \sum_{k: \; \bm{a_k}\in \mathbb{S}_{+,w}} \alpha_{k,w}Q^{+}_{k,w}.
        \end{equation}
        Recall that we strictly require $\left(\lambda^{+}_{k,w}\right)^L = 1$ in order for $Q^{+}_{k,w}$ to be a conserved quantity - this means that $\mathbb{U}\mathcal{Q}\mathbb{U}^{\dagger} = \mathcal{Q}$ implies that $ \alpha_{k,w} = 0 $ for all $k,w$ such that $\left(\lambda^{+}_{k,w}\right)^L \neq 1$.
        \par
        Repeating the analysis for the terms with odd-$x$, the only difference we will find is that these terms will be shifted by $\mathbb{U}$ in the negative-$x$ direction, and must be decomposed into the width-$w$ solitons $\bm{b_l}$ of the $\mathcal{M}_{-,w}$ map (defined via the set $\mathbb{S}_{-,w}$ in Eqn.~\ref{eqn:set_T}) like so,
        \begin{equation}
            \bm{q}(x=1, w) = \sum_{\bm{b_l} \in \mathbb{S}_{-,w}} \beta_{l,w} \bm{b_l},
        \end{equation}
        where we have arbitrarily chosen $x=1$ to define the decomposition (similarly to how we arbitrarily chose $x=0$ when studying the even-$x$ terms), and $\beta_{l,w}$ are a set of normalised complex coefficients, and where again we will have that $\beta_{l,w} = 0$ for all $l$ such that $\left(\lambda^{-}_{l,w}\right)^L \neq 1$. Summing over (odd-)$x$ and $w$ again, and recalling now our definition for the conserved quantities $Q^{-}_{l,w}$ associated to the $\bm{b_l}$ solitons (and as given in Eqn.~\ref{eqn:Q_t}), we will get
        \begin{equation}
            \mathcal{Q}^- = \sum_{w \in \mathbb{O}_{L-3}} \; \sum_{l: \; \bm{b_l}\in \mathbb{S}_{-,w}} \beta_{l,w}Q^{-}_{l,w},
        \end{equation}
        which means that we can rewrite $\mathcal{Q} = \mathcal{Q}^{+} + \mathcal{Q}^{-}$ as
        \begin{equation}
            \mathcal{Q} = \sum_{w \in \mathbb{O}_{L-3}} \left( \sum_{k: \; \bm{a_k}\in \mathbb{S}_{+,w}} \alpha_{k,w}Q^{+}_{k,w} + \sum_{l: \; \bm{b_l}\in \mathbb{S}_{-,w}} \beta_{l,w}Q^{-}_{l,w}\right),
        \end{equation}
        which is what we wished to prove.
    \end{proof}
    We conclude this section by noting that in the thermodynamic limit of an infinite chain, a converse to the above holds.
    \begin{remark}\label{remark:onetoonecorrespondence}
        \textit{One-to-one correspondence in an infinite chain}. Recall that for a finite chain with periodic boundary conditions we required, for each width-$w$ soliton with phase $\lambda$, that $\lambda^{L} = 1$ in order to be able to construct an associated width-$w$ conserved density as per Eqns.~\ref{eqn:Q_s} and \ref{eqn:Q_t}. In the limit of an infinite chain ($L\rightarrow \infty$), however, this is no longer necessary (in an infinite chain we no longer have periodic boundary conditions, and so the solitons no longer return to their previous positions on the chain - they simply propagate off to infinity; consequently, we do not require their phases to ever return to unity either), and so it will be possible to construct a width-$w$ conserved density for every width-$w$ soliton. This implies, along with Theorem \ref{thm:fundamentalcharges}, the existence of a one-to-one correspondence between solitons and conserved densities in the $1+1$D brickwork dual-unitary circuits to which Theorem \ref{thm:fundamentalcharges} applies.
    \end{remark}

\section{Constructing many-body conserved quantities in DUCs}\label{sec:constructingmbcqs}

The results of the above section state that any conserved quantity (composed of terms of a finite range, $w \leq L - 4$) in a dual-unitary circuit can be broken down into linear combinations of width-$w$ solitons. In what follows, we demonstrate two ways in which to construct these width-$w$ solitons: firstly, via products of smaller, constituent solitons; and secondly, via products of fermionic operators.

\subsection{Composite configurations of solitons}\label{sec:boundsolitons}

Consider a quantity formed by producting a right-moving $1$-body soliton $\bm{a} \in \bar{\mathcal{A}}_1$ across multiple different even sites
\begin{equation}
    \bm{A}_{j,j^{\prime}} = \bm{a}_j\bm{a}_{j^{\prime}}, \quad j,j^{\prime}\in\mathbb{E}_{2L}, \quad \mathcal{M}_{+,1}(\bm{a}) = \lambda \bm{a}, \quad |\lambda|=1,
\end{equation}
 which we will loosely refer to as a \textit{composite} configuration of solitons. The action of the Floquet operator on $\bm{A}_{j,j^{\prime}}$ is also to simply shift it two sites to the right (up to a phase $\lambda^2$), meaning it is also a soliton. This can be seen by inserting a resolution of the identity,
\begin{equation}
    \mathbb{U}\bm{A}_{j,j^{\prime}}\mathbb{U^{\dagger}} = \mathbb{U}\bm{a}_j\bm{a}_{j^{\prime}}\mathbb{U}^{\dagger} = \mathbb{U}\bm{a}_j\mathbb{U^{\dagger}}\mathbb{U}\bm{a}_{j^{\prime}}\mathbb{U}^{\dagger} = \lambda\bm{a}_{j+2}\lambda\bm{a}_{j^{\prime}+2} = \lambda^2\bm{A}_{j+2,j^{\prime}+2}. 
\end{equation}
By the same logic as with the 1-body solitons, it is clear that we can again construct a conserved quantity by summing over all even sites,
\begin{equation}
    Q_{\textrm{composite}} = \sum_{x \in \mathbb{E}_{2L}} \lambda^{x}\bm{A}_{x,x+r} = \sum_{x \in \mathbb{E}_{2L}} \lambda^{x}\bm{a}_{x}\bm{a}_{x+r} = \mathbb{U}Q_{\textrm{composite}}\mathbb{U}^{\dagger}, \quad r = j^{\prime}-j, 
\end{equation}  
The above would have held true no matter what set of even sites we had chosen to place the soliton on. So, more generally, and as noted in \cite{bertini2020opent2}, there will be an exponentionally large number of local conserved quantities that can be constructed from $\bm{a}$:
\begin{equation}
    Q = \sum_{x \in \mathbb{E}_{2L}} \lambda^{|\mathbb{J}|\frac{x}{2}}\prod_{j\in \mathbb{J}} \bm{a}_{j+x} = \mathbb{U}Q\mathbb{U}^{\dagger},
\end{equation}
where the set of sites $\mathbb{J}$ can be any collection of even sites i.e.~$\mathbb{J} \subseteq \mathbb{E}_{2L}$\footnote{For the $Q^{+}_{k,w}$ and $Q^{-}_{l,w}$ conserved quantities talked about in Section \ref{sec:results}, we required that the phases of the associated solitons were $L$-th roots of unity; here, we similarly require $\lambda^{|\mathbb{J}|L} = 1$ for $Q$ to be a conserved quantity.}. There will be $\mathcal{O}(\textrm{exp}(L))$ such sets\footnote{Strictly we should only really count the number of sets of sites which are not equivalent under spatial translations by an even number of sites (those which are equivalent will lead to the same $Q$, up to a global phase), however, there will still be $\mathcal{O}(\textrm{exp}(L))$ such sets.}, and hence an exponential in $L$ number of corresponding conserved quantities $Q$. Note that the product $\prod_{j\in \mathbb{J}} \bm{a}_{j+x}$ is itself a many-body soliton, that has been explictly constructed as a composite configuration of the $1$-body soliton $\bm{a}$.
\par
It is clear that we could also create exponentially many conserved quantities out of a left-moving $1$-body soliton $\bm{b} \in \bar{\mathcal{A}}_1$ by producting over odd sites,
\begin{equation}
    Q = \sum_{x \in \mathbb{O}_{2L}}  \lambda^{-|\mathbb{K}|\frac{x-1}{2}} \prod_{k\in \mathbb{K}} \bm{b}_{k+x} = \mathbb{U}Q\mathbb{U}^{\dagger}, \quad \mathcal{M}_{-,1}(\bm{b}) = \lambda \bm{b},
\end{equation}
where the set of sites $\mathbb{K}$ is any collection of odd sites, i.e.~$\mathbb{K} \subseteq \mathbb{O}_{2L}$. It should also be clear that this prescription would also work with larger solitons i.e.~if we were to product together width-$w$ solitons on sites of the same parity, with $w \geq 1$.
\par
Additionally, we also do not need to product together the same soliton. If $\bm{a}$ and $\bm{b}$ are both right-moving solitons (of any width), then $\bm{A}_{j,j^{\prime}} = \bm{a}_j \bm{b}_{j^{\prime}}, \; j,j^{\prime}\in\mathbb{E}_{2L}$ will also be a soliton. An interesting corollary of this is that it allows us to create conserved quantities from \textit{any} soliton (even when $L$ is finite), without requiring that the phase $\lambda$ is a $2L$-th root of unity: Take a soliton, $\bm{a} \in \bar{\mathcal{A}}_w$, $\mathcal{M}_{\pm}(\bm{a}) = \lambda \bm{a}$. The Hermitian conjugate of $\bm{a}$ will also be a soliton, with phase $\lambda^{*}$, i.e.~$\mathcal{M}_{\pm}(\bm{a}^{\dagger}) = \lambda^{*} \bm{a}^{\dagger}$. If we product together $\bm{a}$ and $\bm{a}^{\dagger}$ on sites of the same parity then we will get a soliton $\bm{A}_{j,j^{\prime}} = \bm{a}_j \bm{a}^{\dagger}_{j^{\prime}}$ with unit phase $\lambda\lambda^{*} = 1$, as $|\lambda| = \sqrt{\lambda\lambda^{*}} = 1$ by our definition of solitons. As we can always trivially create a conserved quantity from any soliton with unit phase (see Eqn.~\ref{eqn:unitphasesolitonconservedQs}), then we will have an associated conserved density for $\bm{A}$, even if it were not possible to construct one for $\bm{a}$ (i.e.~if $\lambda$ were not an $L$-th root of unity).
\par
As an example for $d=2$, we consider the gate
\begin{equation}
    U = \left(u_- \otimes \sigma_xe^{i\phi_1\sigma_z}\right)V[J]\left(e^{i\phi_2\sigma_z} \otimes v_-\right),
\end{equation}
with $u_-, v_- \in \textrm{SU}(2)$, $\phi_1, \phi_2 \in [0,\pi]$, and $J \in [0,\pi/2]$, and where
\begin{equation}
    V[J] = e^{i\left(\frac{\pi}{4}\sigma_x \otimes \sigma_x + \frac{\pi}{4}\sigma_y \otimes \sigma_y + J\sigma_z\otimes\sigma_z\right)}. 
\end{equation}
For generic choices of all the free parameters\footnote{For $J=\pi/4$, for instance, $V[J] \equiv \textrm{SWAP}$, and the circuit has solitons in both directions. Also, for certain fine-tuned choices of $u_-$ and $v_-$, the gate can have solitons in both directions.} this is a \textit{chiral} dual-unitary gate (falling into the parameterisation of all chiral dual-unitary gates provided in Ref.~\cite{bertini2020opent2}), with $\bm{\sigma_{z}}$ a right-moving soliton with phase $\lambda = -1$, but no ultralocal ($w=1$) solitons moving to the left. The associated set of charges for a 1D brickwork circuit generated from gates of this form would then be
\begin{equation}
    Q_{\mathbb{J}} = \sum_{x \in \mathbb{E}_{2L}} (-1)^{|\mathbb{J}|\frac{x}{2}} \prod_{j\in \mathbb{J}} \left[\bm{\sigma_z}\right]_{j+x} = \mathbb{U}Q_{\mathbb{J}}\mathbb{U}^{\dagger},
\end{equation}
for sets of even sites $\mathbb{J} \subseteq \mathbb{E}_{2L}$.

\subsection{Composite configurations of fermions}\label{sec:fermionic}

By Theorem \ref{thm:fundamentalcharges}, any conserved quantity (composed of terms of a finite range, $w \leq L - 4$) in a dual-unitary circuit can be broken down into linear combinations of width-$w$ solitons. We now also know one way to construct some of these width-$w$ solitons, via products of $1$-body solitons (as outlined in Section \ref{sec:boundsolitons} above). One question to ask at this stage might be: is this way (i.e.~via composite configurations of $1$-body solitons) the only way to construct all solitons (i.e.~of any width $w$)? Given the form of the dynamical constraints used in Theorem \ref{thm:fundamentalcharges}, it seems plausible that the width-$w$ solitons must also behave solitonically on a more microscopic level.
\par
A simple counterexample - a conserved quantity that cannot be formed out of composite configurations of $1$-body solitons - can be identified by studying a dual-unitary circuit consisting of fermionic SWAP operators. That is, we consider a chain of qubits, $d=2$, and set every gate in our dual-unitary circuit acting on this chain to be equal to
\begin{equation}
    \textrm{FSWAP} = \begin{pmatrix}
        1 & 0 & 0 & 0 \\
        0 & 0 & 1 & 0 \\
        0 & 1 & 0 & 0 \\
        0 & 0 & 0 & -1 \\
    \end{pmatrix},
\end{equation}
a gate which has been studied recently in the context of simulating physically relevant Hamiltonians on quantum computers, such as the Ising model \cite{cervera2018exactisingsimqc} and the Fermi-Hubbard model \cite{cade2020strategiesfermihubbardnearterm}.
\par
It is straightforward to verify that FSWAP is dual-unitary, and that $\bm{\sigma_z}$, the Pauli-Z operator, is a soliton of the resulting circuit, i.e.
\begin{equation}
    U = \textrm{FSWAP} \implies \mathbb{U}\left[\bm{\sigma_z}\right]_x\mathbb{U}^{\dagger} = \left[\bm{\sigma_z}\right]_{x+2} \; \forall \; x \in \mathbb{E}_{2L}, \quad \mathbb{U}\left[\bm{\sigma_z}\right]_x\mathbb{U}^{\dagger} = \left[\bm{\sigma_z}\right]_{x-2} \; \forall \; x \in \mathbb{O}_{2L}.
\end{equation}
The action of FSWAP on the remaining Pauli operators, $\bm{\sigma_x}$ and $\bm{\sigma_y}$, is as follows:
\begin{equation}
    \textrm{FSWAP}(\bm{\sigma_x} \otimes \mathds{1})\textrm{FSWAP}^{\dagger} = \bm{\sigma_z} \otimes \bm{\sigma_x},
\end{equation}
and
\begin{equation}
    \textrm{FSWAP}(\bm{\sigma_y} \otimes \mathds{1})\textrm{FSWAP}^{\dagger} = \bm{\sigma_z} \otimes \bm{\sigma_y}.
\end{equation}
It follows from the above that any operator $\bm{p}\bm{\sigma_z}^{\otimes l}\bm{q}$, with $l$ odd and the single-qubit operators $\bm{p}$ and $\bm{q}$ being arbitrary linear combinations of Pauli-$X$ and -$Y$ operators, i.e.~$\bm{p},\bm{q} \in \textrm{Span}(\{ \bm{\sigma_x}, \bm{\sigma_y}\})$ (where $\textrm{Span}(S)$ specifically refers to the linear span with complex coefficients of the elements of some set $S$), will be a width-$(l+2)$ soliton of an FSWAP circuit,
\begin{equation}
   U = \textrm{FSWAP} \implies \mathcal{M}_{\pm, l+2}\left(\bm{p}\bm{\sigma_z}^{\otimes l}\bm{q}\right) = \bm{p}\bm{\sigma_z}^{\otimes l}\bm{q},
\end{equation}
and there will exist an associated conserved quantity,
\begin{equation}\label{eqn:fermionicconsq}
    \mathcal{Q} = \sum_{x \in \mathbb{Z}_{2L}} \left[\bm{p}\bm{\sigma_z}^{\otimes l}\bm{q}\right]_{x,...,x+l+1} = \mathbb{U}\mathcal{Q}\mathbb{U}^{\dagger}.
\end{equation}
Yet, there is no way to express these operators as a composite configuration of $1$-body solitons; $\bm{\sigma_z}$ is the only $1$-body soliton of the FSWAP circuit, and so cannot be used to decompose the $\bm{p}$ and $\bm{q}$ terms, which will be orthogonal to it. Additionally, each $\bm{p}\bm{\sigma_z}^{\otimes l}\bm{q}$ will have non-trivial support on both even-$x$ and odd-$x$ sites, in stark contrast to the composite solitons (which will only have non-trivial support on exclusively odd or even sites). We must conclude that conserved quantities in DUCs cannot, in general, be broken down into linear combinations of products of single-body solitons.
\par
However, there is still an intuitive way to understand these kinds of conserved quantities: defining an operator $\bm{f}^{\dagger}_x$ which creates a fermion at site $x$ via a Jordan-Wigner transformation,
\begin{equation}
    \bm{f}^{\dagger}_x = \left[\bm{\sigma_z}^{\otimes x} \bm{\sigma_+}\right]_{0,...,x}, \quad \bm{\sigma_+} = \begin{pmatrix} 0 & 1 \\ 0 & 0 \\ \end{pmatrix} = \frac{\bm{\sigma_x} - i\bm{\sigma_y}}{2},  \quad \{\bm{f}_x,\bm{f}_{x^{\prime}}\} = 0,
\end{equation}
we see that a fermionic bilinear term such as
\begin{equation}
    \bm{F}_{x,w} = \bm{f}_x\bm{f}^{\dagger}_{x+w-1} = \left[\bm{\sigma_-}\bm{\sigma_z}^{\otimes w-2}\bm{\sigma_+}\right]_{x,...,x+w-1},
\end{equation}
will be a $w$-body soliton of FSWAP with $\bm{\sigma_-} = \bm{p}, \bm{\sigma_+} = \bm{q} \in \textrm{Span}(\{\bm{\sigma_x}, \bm{\sigma_y}\})$, and will obey bosonic statistics
\begin{equation}
    \left[ \bm{F}_{x,w}, \bm{F}_{x^{\prime},w} \right] = 0.
\end{equation}
So, although we cannot break the quantities $\bm{p}\bm{\sigma_z}^{\otimes l}\bm{q}$ into products of simple, single-body solitons, we can (at least under certain conditions) break them down into products of single site fermions, which themselves behave somewhat similarly to solitons under the dynamics of the FSWAP circuit, in that
\begin{equation}
    \mathbb{U}\bm{f}_x\mathbb{U}^{\dagger} = \bm{f}_{x+2\nu}, \quad \nu \in \{+1, -1\}, \quad x \in \mathbb{Z},
\end{equation}
if we drop our periodic boundary conditions and instead consider an infinite chain, and where $\nu=1$ if $x$ is even, and $\nu=-1$ if $x$ is odd. Of course, this is perhaps not surprising - by its own name, the action of the FSWAP circuit on fermions is to swap them.
\par
Note that although we could also in the case of an infinite chain construct conserved quantities out of these individual fermions, such as
\begin{equation}\label{eqn:fermionicsolitonicconslaw}
    \mathcal{Q} = \sum_{x} \bm{f}_x = \mathbb{U}\mathcal{Q}\mathbb{U^{\dagger}},
\end{equation}
for example, this would not violate Theorem \ref{thm:fundamentalcharges} in any way, as $\mathcal{Q}$ (as given in Eqn.~\ref{eqn:fermionicsolitonicconslaw}) contains terms which (before doing the Jordan-Wigner transformation) have extensive (infinite) spatial support - which are explicitly not considered in the theorem.

\subsubsection{Dual-unitary matchgates}
Quantum circuits that host fermionic conserved quantities have also been studied in the context of matchgate circuits. Matchgates are two-qubit gates $G(A,B) \in U(4)$ parameterised as
\begin{equation}
    G(A,B) = \begin{pmatrix}
        a & 0 & 0 & b \\
        0 & e & f & 0 \\
        0 & g & h & 0 \\
        c & 0 & 0 & d \\
    \end{pmatrix}, \;\;\; A = \begin{pmatrix} a & b \\ c & d \\ \end{pmatrix}, \;\;\; B = \begin{pmatrix} e & f \\ g & h \\ \end{pmatrix},
\end{equation}
with  $A, B \in U(2)$ and $\det(A) = \det(B) = \pm 1$. Circuits comprised of matchgates are classically simulable, as the dynamics can be mapped on to that of free fermions \cite{jozsa2008matchgates}.\par
Given the free-fermionic properties of FSWAP illustrated in the above section, it is natural to assume that it may also be a matchgate. Indeed, it is - explicitly, it is the matchgate $G(\sigma_z,\sigma_x)$. This raises the question as to what other dual-unitary gates are also matchgates; this can be answered by comparing the KAK decomposition of an arbitrary matchgate, given in Ref.~\cite{mocherla2024extendingmatchgatesimulationmethods} as
\begin{equation}
    G(A,B) \equiv \left(e^{i\phi_1\sigma_z} \otimes e^{i\phi_2\sigma_z}\right)e^{i\left(a \sigma_x \otimes \sigma_x + b \sigma_y\otimes \sigma_y\right)}\left(e^{i\phi_3\sigma_z} \otimes e^{i\phi_4\sigma_z}\right), \quad a,b,\phi_1,\phi_2,\phi_3,\phi_4 \in \mathbb{R},
\end{equation}
with the parameterisation of 2-qubit dual-unitary gates given in Ref.~\cite{bertini2019exactcorrfuncs} as
\begin{equation}
    U = e^{i\phi}\left(u_- \otimes u_+\right)e^{i(\frac{\pi}{4}(\sigma_x \otimes \sigma_x + \sigma_y \otimes \sigma_y) +J\sigma_z\otimes \sigma_z)}\left(v_+ \otimes v_-\right), \quad u_{\pm}, v_{\pm} \in \textrm{SU}(2), \quad \phi, J\in\mathbb{R}.
\end{equation}
It is evident that dual-unitary gates with $J = 0$ and $u_{\pm}, v_{\pm}$ restricted to being $\bm{\sigma_z}$-rotations are also matchgates (with $a=b=\frac{\pi}{4}$). Any pair of the single qubit $\bm{\sigma_z}$-rotations can be pulled through the entangling gate and absorbed into the other pair\footnote{One way to see this is to note that $V[J]$ is equivalent, for all $J$, to a $\textrm{SWAP}$ multiplied by a 2-qubit unitary that is diagonal in the computational basis (a $ZZ$-rotation), and so the single qubit $\bm{\sigma_z}$-rotations commute with it (up to the $\textrm{SWAP}$). Equivalently, $\bm{\sigma_z}$ is a soliton of $V[J]$ with $\lambda = 1$ for all $J$, and so any $Z$-rotation - which will just be a linear combination of $\bm{\mathds{1}}$, which is a trivial $\lambda = 1$ soliton, and $\bm{\sigma}_z$ - will also be a soliton (and so will commute with $V[J]$ up to a spatial translation by one site).}, leaving us with the family of gates 
\begin{equation}\label{eqn:DUmatchgates}
    U(\phi_+, \phi_-) = e^{i\frac{\pi}{4}\left( \sigma_x \otimes \sigma_x +  \sigma_y\otimes \sigma_y\right)}\left(e^{i\phi_+\sigma_z} \otimes e^{i\phi_-\sigma_z}\right),
\end{equation}
parameterised by $\phi_+, \phi_- \in \mathbb{R}$, where $\phi_+ = \phi_2 + \phi_3$ and $\phi_- = \phi_1 + \phi_4$. In the standard matchgate language, we have 
\begin{equation}
    U(\phi_+, \phi_-) = G(e^{-i\theta_+\sigma_z}, ie^{-i\theta_-\sigma_z}\sigma_x),
\end{equation}
where $e^{-i\theta_+\sigma_z}$ acts on the even-parity subspace spanned by $\{\ket{00}, \ket{11}\}$, and $ie^{-i\theta_-\sigma_z}\sigma_x$ acts on the odd-parity subspace spanned by $\{\ket{01}, \ket{10}\}$, and where $\theta_{\pm} = \phi_+ \pm \phi_-$.
\par
We note that these gates are a restricted subclass of the dual-unitary gates with solitons in both directions, as parameterised in Ref.~\cite{bertini2020opent2}; they have $\bm{\sigma_z}$ as a unit-phase soliton in both directions. In Ref.~\cite{bertini2020opent2} it was also noted that the $J=0$ parameter line of the entangling gate is equivalent to the self-dual kicked Ising model (the paradigmatic Hamiltonian model that dual-unitary circuits are related to \cite{bertini2018exactsff,bertini2019entanglementspreading,gopalakrishnan2019finitedepthinfinitewidth,claeys2020mvqcs}); all the dual-unitary matchgates are hence equivalent to this model too, up to the single qubit $Z$-rotations.
\par
The adjoint action of $U(\phi_+, \phi_-)$ on $\bm{\sigma_{\pm}}$ terms can be neatly summarised as
\begin{equation}\label{eqn:sigmapluszerophase}
    U(\phi_+, \phi_-)\left[\bm{\sigma_+}\right]_0 U(\phi+, \phi_-)^{\dagger} = ie^{-i\phi_+}\left[\bm{\sigma_z}\otimes\bm{\sigma_+}\right]_{0,1},
\end{equation}
\begin{equation}
    U(\phi_+, \phi_-)\left[\bm{\sigma_+}\right]_1 U(\phi_+, \phi_-)^{\dagger} = ie^{-i\phi_-} \left[\bm{\sigma_+}\otimes\bm{\sigma_z}\right]_{0,1},
\end{equation}
\begin{equation}\label{eqn:sigmaminuszerophase}
    U(\phi_+, \phi_-)\left[\bm{\sigma_-}\right]_0 U(\phi_+, \phi_-)^{\dagger} = -ie^{i\phi_+} \left[\bm{\sigma_z}\otimes\bm{\sigma_-}\right]_{0,1},
\end{equation}
\begin{equation}
    U(\phi_+, \phi_-)\left[\bm{\sigma_-}\right]_1 U(\phi_+, \phi_-)^{\dagger} = -ie^{i\phi_-} \left[\bm{\sigma_-}\otimes\bm{\sigma_z}\right]_{0,1}.
\end{equation}
Using these mappings, one can then verify that all these dual-unitary matchgates will have an exponential number of solitons that can be constructed from fermionic operators. In particular, the bilinears considered above, $\bm{F}_{x, w} = \bm{f}_x\bm{f}^{\dagger}_{x+w-1} = \left[\bm{\sigma_-}\otimes\bm{\sigma_z}^{\otimes w-2}\otimes{\bm{\sigma_+}}\right]_{x,...,x+w-1}$ with $w$ odd, will be $w$-body solitons with unit phase\footnote{The phases from Eqns.~\ref{eqn:sigmapluszerophase} and \ref{eqn:sigmaminuszerophase} cancel each other out.}.

\section{Conclusion}\label{sec:conc}

Dual-unitary circuits represent, without doubt, a pertinent development in the study of many-body quantum systems from a quantum information perspective; they have afforded us precious analytical insights on the dynamics of ergodic, chaotic many-body quantum systems. Here, however, we have clarified a story that has been developing in the dual-unitaries literature since their advent \cite{bertini2019exactcorrfuncs, bertini2020opent1, bertini2020opent2,borsi2022constructionofduqcs}: the set of conserved quantities realisable in dual-unitary circuits is limited to those that can be constructed from solitons, leading to trivial, non-diffusive dynamics of any localised density of these conserved quantities, and hence restricting the range of \textit{non-ergodic} dynamics we can expect to see in these systems.
\par
We reiterate that this implies that a dual-unitary circuit either has a soliton (of any finite size) - leading to a strong form of non-ergodicity, in that an exponential (in the system size) number of conserved quantities can be constructed, and correlations associated to the soliton will never decay - or it does not - in which case, as we have shown, we can essentially rule out the existence of any conserved quantities, and the circuit will exhibit (provably maximal, for certain classes of circuits \cite{bertini2021randomsffDUCs}) quantum chaos.
\par
Despite these limitations on the non-ergodicity they can realise, exemplified by the relatively trivial dynamics of the solitons, there may well still be plenty of rich non-trivial out-of-equilibrium phenomena to explore in dual-unitary circuits - the recent discovery that they can host quantum many-body scars being a prime example \cite{logaric2023scarsinDUCs}. Attempts to understand the link, if it exists, between such forms of non-ergodicity more generally in DUCs and the presence of solitons would be worthwhile. Moreover, despite the aforementioned relatively trivial dynamics of solitons, this does not preclude the possibility of non-trivial structure that could be associated with their presence; the fact that they exist in their own, separate \textit{fragment} of Hilbert space - reminiscent of a decoherence-free subspace \cite{lidar1998decoherencefreesubspsforqc, lidar2014reviewofdecohfreesubsps} - amongst an otherwise highly chaotic system\footnote{Any non-conserved operator weight in DUCs will continue to spread out with an exponentially decaying (away from the trailing lightcone edge), maximally chaotic OTOC, even in the presence of solitons \cite{claeys2020mvqcs}.} is intruiging, and invites investigation of the potential utility of this kind of dynamics for protecting or processing quantum information in nascent quantum technologies.
\par 
We use the term `fragment' here to allude to the phenomenon of Hilbert space fragmentation, that has recently been studied in the context of certain many-body Hamiltonians \cite{moudgalya2022scarsandfragmentationreview, moudgalya2022fragmentationandcommutantalgbrs, li2023fragmentationinopenqsys,li2024highlyentangledstationarystatesstrong}. It seems highly likely that the formalism used therein can be modified and applied to dual-unitary circuits, to identify fragmentation in DUCs with solitons. The presence of symmetries and fragmentation can leave fingerprints in the late-time dynamics of many-body systems, such as in entanglement in steady states \cite{li2024highlyentangledstationarystatesstrong}, and modifications to thermal states (e.g. generalised Gibbs ensembles \cite{rigol2007GGE}). Understanding the nature of conserved quantities and means of constructing them in dual-unitary circuits, as studied here, could then be pertinent with a view towards using integrable dual-unitary circuits as toy models for these phenomena. We note that since dissemination of an initial version of this manuscript there has been pleasing progress in this direction, with fragmentation and thermalisation to a generalised Gibbs ensemble identified in integrable dual-unitary circuits in Ref.~\cite{foligno2024chargedDUCs}.
\par
More generally in the context of integrability, understanding how to `Floquetify' integrable, non-Floquet, Hamiltonians (perhaps via standard Trotterisation techniques) without destroying the integrability is a problem of interest, studied for instance in Ref.~\cite{pozsgay2024qcircsfreefermionsindisguise}. The integrable circuits studied here in this work could be viewed as a way in which to study this problem from the other direction. For DUCs on qubits, $d=2$, the connection to integrable models is fairly well understood, arising due to the dual-unitary entangling gate $V[J]$ generating a special point of the Trotterised XXZ model \cite{claeys2022corrsinintegrablecircs}. Dual-unitarity naturally generalises to arbitrary $d$, however; understanding how integrable dual-unitary circuits on qudits connect to integrable Hamiltonians on higher spin systems is perhaps less well understood, but worthy of investigation. Some progress has already been made in this direction, with connections to the integrable Potts model (for $d=3$) studied in Ref.~\cite{claeys2024opdynamicsinSTdualHlattices}. The constructions and results presented herein could be relevant for finding further and associated generalisations of integrable dual-unitary circuits. Understanding how to deform integrable dual-unitary circuits into new non-integrable dual-unitary circuits circuits is also of interest (see e.g. Ref.~\cite{gombor2022superintegrableCA}), as well as the converse - exploring breaking dual-unitarity while maintaining integrability, leading to a potential breakdown of the solitonic quasi-particle picture in the process. The latter has been studied in Ref.~\cite{rampp2023haydenpreskill}, where the study of Hayden-Preskill decoding in integrable and non-integrable (dual- and non-dual-) unitary circuits highlighted that integrability, particularly in combination with dual-unitary, can have stark implications for quantum information processing tasks.
\par
It would also be interesting to establish whether the two methods of constructing many-body solitons in $1+1$D dual-unitary circuits on qubits given in Section \ref{sec:constructingmbcqs} are exhaustive. Both methods require the existence of a 1-body soliton in order to construct any many-body ($w > 1$) solitons (stricly speaking, the construction from products of constituent solitons can use solitons of any width, but one can ask in turn how these constituent solitons are constructed - if we only have access to the two constructions outlined in Section \ref{sec:constructingmbcqs}, then eventually we will need a $1$-body soliton). Could it be the case that for these circuits (i.e.~$1+1$D DUCs on qubits) local thermalisation (i.e.~ergodic, decaying correlation functions of all $1$-site observables) implies a more extensive, global sense of thermalisation (i.e.~no many-body solitons and hence no non-trivial many-body conserved quantities, leading to a decay of all many-body correlation functions)? It is already known that this is not the case for circuits with $d>2$ (examples of $1+1$D dual-unitary circuits on qutrits with ergodic local correlations but non-ergodic many-body correlations were found in Ref.~\cite{borsi2022constructionofduqcs}). In the restrictive case of $d=2$, however, it seems plausible that the constructions given in Section \ref{sec:constructingmbcqs} could be the only way to generate solitons, and hence (according to Theorem \ref{thm:fundamentalcharges}) also be the only way to generate conserved quantities and non-ergodic correlations. Dual-unitary circuits have already proved to be very useful for studying the ETH and strong notions of thermalisation (see Refs.~\cite{fritzsch2021eth, ho2022exactemergent,ippoliti2022dynamical, Claeys2022emergentquantum}); this may represent another potential application of them in this direction.
\par
We wish to mention the links already made between dual-unitary circuits and conformal field theories (CFTs): for integrable and non-integrable dual-unitary circuits, operator entanglement has been shown to exhibit CFT behaviour \cite{reid2021entanglementbarriersinducs}; it also has been shown that dual-unitary circuits with a discrete version of conformal symmetry can be constructed, which are being explored as a potential toy model of the holographic AdS/CFT correspondence \cite{masanes2023discrete}. Our results here could be understood as solidifying another link between dual-unitary circuits and CFTs, in that the restriction of charges in dual-unitary circuits to being right- or left-moving solitons proven herein is highly reminiscent of the restriction of quanta of charge to being independent, right- or left-moving modes propagating at the effective speed of light (known as chiral separation \cite{bernard2016outofeqCFTreview}) in one-dimensional CFTs. This further motivates the exploration of using these circuits as toy models for studying CFTs in important physical contexts: in holographic models of quantum gravity (as has been initiated in Ref.~\cite{masanes2023discrete}); or perhaps in phase transitions between phases of matter involving symmetry protected topological phases (SPTs) \cite{cho2017SPTsandBCFTs}, in particular those with chiral transport properties \cite{po2016chiralfloqphasesofMBLbosons,po2017radicalchiralfloq}. We note that it has already been established, in Ref.~\cite{stephen2022universal}, that $1+1$D dual-unitary circuits can be constructed which generate SPT phases with `infinite order' (they generate states with 1D SPT order from initial product states with, as $L \rightarrow \infty$, complexity that grows unboundedly with the number of Floquet periods). Their ability to realise models with symmetries related to Jordan-Wigner strings, as demonstrated in Section \ref{sec:fermionic}, perhaps also gives some further creedence to the idea that they could be useful for studying SPTs \cite{verresen2017oneDSPTs} and their detection via string order parameters \cite{perezgarcia2008stringorder}.
\par
As a final comment on connections to topological phases, we note that one might envision that similar results as obtained herein could be formulated for local quantum circuits with a brickwork structure and multidirectional unitarity in higher spatial dimensions (such as the $2+1$D circuits considered in Refs.~\cite{milbradt2022ternary} and \cite{jonay2021triunitary}); we hence also, in this vein, echo a call first made in Ref.~\cite{sommers2023crystallineqcircs} to understand how these families of circuits might fit into the wider framework of quantum cellular automata in higher dimensions and the associated characterisation of topological phases \cite{haah2023nontrivialQCAinhighdim,haah2021cliffordQCA,freedman2020classificationQCA,freedman2022groupstrucQCA,shirley2022QCAchiralsemionsurface} (where often the presence of gliders - which the solitons considered here would be counted as - is of particular relevance \cite{stephen2019subsystemsymandQCAs}).
\par
And, to conclude, while we have ascertained here that dual-unitary circuits can only realise trivial dynamics of any local density of conserved charge, we note the success of recent studies on \textit{perturbed} dual-unitary circuits \cite{kos2021ducpathint, rampp2023perturbedDUCsopspreading}. In perturbed DUCs one can expect - and often analytically verify - the emergence of features of more `generic' unitary dynamics; for instance, the behaviour of the out-of-time-order correlator in perturbed chaotic DUCs can be shown to resemble that of random unitary circuits \cite{rampp2023perturbedDUCsopspreading}. Analytical results suggest that weakly breaking dual-unitarity in DUCs with solitons leads to the emergence of non-trivial (diffusive) behaviour of local densities of charges in a late time limit \cite{kos2021ducpathint,kos2021extendedsyswnoisydriving}. It would be interesting to establish a handle on how this picture emerges, and what physical transport phenomena can be realised in this setting.

\section{Acknowledgements}
THD acknowledges support from the EPSRC Centre for Doctoral Training in Delivering Quantum Technologies [Grant Number EP/S021582/1], and would like to thank Christopher J. Turner, Pieter W. Claeys, Pavel Kos, Georgios Styliaris, Harriet Apel, Anastasia Moroz, and Sougato Bose for fruitful discussions. AP is funded by the European Research Council (ERC) under the EU's Horizon 2020 research and innovation program (Grant Agreement No. 853368).

\bibliographystyle{quantum}
\bibliography{main}

\appendix

\section{Further analysis of dynamical constraints on the spatial support of local operators in dual-unitary circuits}\label{appendix:fate}
In Fig.~\ref{spreadingdiagrams} we show how the spatial support of an operator $\bm{q}_{x,...,x+w-1}$ can change (according to Property \ref{property1}) under the evolution of a brickwork dual-unitary Floquet operator in 4 different cases, corresponding to the different combinations of $x$ and $w$ being odd or even, using the notation introduced in Sec.~\ref{sec:dynamicalconstraints}. Explicitly, we consider dual-unitary conjugation by $\mathbb{U}$ on elements of the following subspaces:
\begin{align}
    \mathcal{B}_{\textrm{ee}} &= \left\{\bm{q}_{x,...,x+w-1} \in \mathcal{A}_{2L} |\; (x, w) \in \mathbb{E}_{2L} \times \mathbb{E}_{L+1}/0, \;\bm{q} \in \bar{\mathcal{A}}_w\right\}, \quad \textrm{(The `even-even' subspace)}, \\
    \mathcal{B}_{\textrm{eo}} &= \left\{\bm{q}_{x,...,x+w-1} \in \mathcal{A}_{2L} |\; (x, w) \in \mathbb{E}_{2L} \times \mathbb{O}_{L+1}, \;\bm{q} \in \bar{\mathcal{A}}_w\right\}, \quad \textrm{(The `even-odd' subspace)}, \\
    \mathcal{B}_{\textrm{oe}} &= \left\{\bm{q}_{x,...,x+w-1} \in \mathcal{A}_{2L} |\; (x, w) \in \mathbb{O}_{2L} \times \mathbb{E}_{L+1}/0, \;\bm{q} \in \bar{\mathcal{A}}_w\right\}, \quad \textrm{(The `odd-even' subspace)}, \\
    \mathcal{B}_{\textrm{oo}} &= \left\{\bm{q}_{x,...,x+w-1} \in \mathcal{A}_{2L} |\; (x, w) \in \mathbb{O}_{2L} \times \mathbb{O}_{L+1}, \;\bm{q} \in \bar{\mathcal{A}}_w\right\}, \quad \textrm{(The `odd-odd' subspace)}.
\end{align}
\par
We have restricted ourselves to considering operators with width $w \leq L$ (i.e.~half the chain). Only by making this restriction can we ensure that the subspaces above have no overlap - this is necessary as we will later want to be able to say that elements from different $\mathcal{B}$ subspaces are orthogonal to each other. To illustrate, consider an operator $\bm{q}_{0,...,2L-1} = \bm{a}_0\bm{b}_{2L-1}$, with $\bm{a}, \bm{b} \in \bar{\mathcal{A}}_1$. This could be viewed as an operator supported over the region $\left[0,2L-1\right]$ with $x=0$ even and width $w=2L$ even, and would hence be in the  $\mathcal{B}_{\textrm{ee}}$ subspace (if we had not restricted $w\leq L$). However, it could also be viewed as an operator supported over the region $\left[2L-1, 0\right]$ now with $x=2L-1$ \textit{odd} and width $w=2$ even (i.e.~viewed as $\bm{q}_{2L-1,0} = \bm{b}_{2L-1}\bm{a}_0$), and would hence be in the $\mathcal{B}_{\textrm{oe}}$ subspace. If we restrict $w \leq L$, however, then it is clear that $\bm{q}_{2L-1,0}$ is only in $\mathcal{B}_{\textrm{oe}}$ (it does not fit in the constructive definition of $\mathcal{B}_{\textrm{ee}}$), and we avoid this problem. As we will wish to utilise orthogonality of elements from different subspaces in the full proof of Lemma \ref{lem:onlyoddw} (and the associated Sublemmas) later in this Appendix, we hence enforce $w \leq L$, such that we have
\begin{equation}
    \bm{q} \in \mathcal{B}_{i}, \; \bm{p} \in \mathcal{B}_{j}, \; i,j \in \{\textrm{ee},\textrm{eo},\textrm{oe},\textrm{oo}\}, \; i \neq j \implies \tr(\bm{q}^{\dagger}\bm{p}) = 0 .
\end{equation}
\par

 The mappings between the different $\mathcal{B}$ subspaces under dual-unitary conjugation by $\mathbb{U}$ can be derived by considering the diagrams in Fig.~\ref{spreadingdiagrams}, and these mappings are collated into a directed graph in Fig.~\ref{fig:digraph}. One of the first things that we can note by studying Fig.~\ref{fig:digraph} is that we cannot have any terms in the subspace $\mathcal{B}_{\textrm{ee}}$ contribute to a quantity $\mathcal{Q}$ if we want $\mathcal{Q}$ to be conserved under the action of a brickwork dual-unitary Floquet operator. We provide a proof of this below.  

 \begin{sublemma}\label{lem:noeveneven}
    \textit{No even-even terms}. Let
    \begin{equation}
        \mathcal{Q} = \sum_{x\in\mathbb{Z}_{2L}}\sum_{w\in\mathbb{Z}_{L-3}/0} \bm{q}(x,w)_{x,...,x+w-1}
    \end{equation}
    with $\bm{q}(x,w)\in\bar{\mathcal{A}}_w$ be a conserved quantity under the dynamics of a 1+1D brickwork Floquet operator (Eqn.~\ref{eqn:floquetoperator}), $\mathbb{U} \mathcal{Q} \mathbb{U}^{\dagger} = \mathcal{Q}$. If $\mathbb{U}$ is dual-unitary, then $\mathcal{Q}$ cannot contain any terms $\bm{a}\in \mathcal{B}_{\textrm{ee}}$ (i.e.~operators with even $x$ and even width $w \leq L-4$),
    \begin{equation}
        \tr\left(\mathcal{Q}^{\dagger}\bm{a}\right) = 0 \;\; \forall \; \bm{a} \in \mathcal{B}_{\textrm{ee}}.
    \end{equation}
\end{sublemma}
\begin{proof}
    For proof by contradiction, let's say that $\mathcal{Q}$ does contain some terms from the even-even subspace (i.e.~there exists some  $\bm{a} \in \mathcal{B}_{\textrm{ee}}$ such that $\tr\left(\mathcal{Q}^{\dagger}\bm{a}\right) \neq 0$) - one of these terms will be the widest such term. Let's say this term has a width $w_{\textrm{max}}$.
    \par
    From Fig.~\ref{fig:digraph}, we can note that this term will either be mapped by $\mathbb{U}$ into a new term in $\mathcal{B}_{\textrm{ee}}$ with width $w_{\textrm{max}}-4$, or into terms in the other subspaces (if $w_{\textrm{max}} \leq 4$). We can also note that there is no way to map from terms in the other subspace to terms in the even-even subspace.
    \par
    Putting these two facts together, we can conclude that we won't have any terms of width $w_{\textrm{max}}$ in the even-even subspace present in $\mathbb{U} \mathcal{Q} \mathbb{U}^{\dagger}$ (recall the elements in the other subspaces will be orthogonal to those in $\mathcal{B}_{\textrm{ee}}$, so we will not be able to reconstruct any even-even terms from terms in the other subspaces) - but this would mean that $\mathcal{Q} \neq  \mathbb{U} \mathcal{Q} \mathbb{U}^{\dagger}$, in contradiction with our requirement that $\mathcal{Q}$ is a conserved quantity. We can hence conclude that we cannot have this widest even-even term present in $\mathcal{Q}$ (if we want it to be conserved). But, as we have already stated, if we have any even-even terms present in $\mathcal{Q}$ then we will have one which is the widest. Hence, we must conclude that $\tr\left(\mathcal{Q}^{\dagger}\bm{a}\right) = 0 \;\; \forall \; \bm{a} \in \mathcal{B}_{\textrm{ee}}$ as required.

\end{proof}

Next, we can note that terms in the `odd-even' subspace cannot form a conserved quantity by themselves - $\mathcal{Q}$ cannot consist solely of terms with odd $x$ and even $w$.

\begin{sublemma}\label{lem:notjustoddeven}
    \textit{Odd-even terms not sufficient}. Let
    \begin{equation}
        \mathcal{Q} = \sum_{x\in\mathbb{Z}_{2L}}\sum_{w\in\mathbb{Z}_{L-3}/0} \bm{q}(x,w)_{x,...,x+w-1}
    \end{equation}
    with $\bm{q}(x,w)\in\bar{\mathcal{A}}_w$. If $\mathcal{Q}$ consists solely of terms with $x$ odd and $w$ even, then it cannot be a conserved quantity under a brickwork dual-unitary Floquet operator $\mathbb{U}$. Concretely, defining the set
    \begin{equation}
        A = \{\bm{q}_{x,...,x+w-1}| \tr\left(\mathcal{Q}^{\dagger}\bm{q}_{x,...,x+w-1}\right) \neq 0, \; x\in\mathbb{Z}_{2L}, \; w\in\mathbb{Z}_{L-3}/0, \; \bm{q}\in\bar{\mathcal{A}}_w \},
    \end{equation}
    then
    \begin{equation}
        A \subseteq \mathcal{B}_{\textrm{oe}} \implies \mathbb{U}\mathcal{Q}\mathbb{U}^{\dagger} \neq \mathcal{Q}.
    \end{equation}
\end{sublemma}
\begin{proof}
    Suppose $A \subseteq \mathcal{B}_{\textrm{oe}}$. There will be some $\bm{q}(x,w)_{x,...,x+w-1} \in \mathcal{B}_{\textrm{oe}}$ which contributes to $\mathcal{Q}$ for which $w$ is minimum. Let's say this term has a width $w_{\textrm{min}}$.
    \par
    From Fig.~\ref{oddeventerms}, we can see that all elements of $\mathcal{B}_{\textrm{oe}}$ remain in $\mathcal{B}_{\textrm{oe}}$ (as long as $w\leq L-4$, as assumed herein) and grow in width by 4 sites under the action of a dual-unitary Floquet operator $\mathbb{U}$. Consequently, there will be no terms of width $w_{\textrm{min}}$ present in $\mathbb{U}\mathcal{Q}\mathbb{U}^{\dagger}$. Hence, $\mathcal{Q} \neq \mathbb{U}\mathcal{Q}\mathbb{U}^{\dagger}$.
\end{proof}

Finally, we can prove that a quantity that is conserved under the action of a dual-unitary Floquet operator $\mathbb{U}$ can only consist of terms from the subspaces with odd-$w$.

\begin{extralemma}
    \textit{Only odd-width terms (with formal proof)}.
    Let
    \begin{equation}
        \mathcal{Q} = \sum_{x\in\mathbb{Z}_{2L}}\sum_{w\in\mathbb{Z}_{L-3}/0} \bm{q}(x,w)_{x,...,x+w-1}
    \end{equation}
    with $\bm{q}(x,w)\in\bar{\mathcal{A}}_w$, and let $\mathbb{U}$ be a 1+1D brickwork dual-unitary Floquet operator (Eqn.~\ref{eqn:floquetoperator}). Then
    \begin{equation}
        \mathbb{U}\mathcal{Q}\mathbb{U}^{\dagger} = \mathcal{Q} \implies \bm{q}(x,w) = 0 \;  \forall \; w \in \mathbb{E}_{L-3}/0.
    \end{equation}
\end{extralemma}
\begin{proof}
    %By Sublemma \ref{lem:noeveneven}, we already know that $\bm{q}(x,w) = 0 \; \forall \; (x,w) \in \mathbb{E}_{2L} \times \mathcolorbox{yellow}{\mathbb{E}_{L-3}/0}$. Hence, let us write $\mathcal{Q}$ as
    By Sublemma \ref{lem:noeveneven}, we already know that $\bm{q}(x,w) = 0 \; \forall \; (x,w) \in \mathbb{E}_{2L} \times \mathbb{E}_{L-3}/0$. Hence, let us write $\mathcal{Q}$ as
    \begin{equation}
        \mathcal{Q} = \mathcal{Q}_{\textrm{odd-}w} + \mathcal{Q}_{\textrm{odd-even}},
    \end{equation}
    where
    \begin{equation}
        \mathcal{Q}_{\textrm{odd-}w} = \sum_{x\in\mathbb{Z}_{2L}}\sum_{w\in\mathbb{O}_{L-3}} \bm{q}(x,w)_{x,...,x+w-1} \in \textrm{Span}\left(\mathcal{B}_{eo} \cup \mathcal{B}_{oo}\right),
    \end{equation}
    and
    \begin{equation}
        \mathcal{Q}_{\textrm{odd-even}} = \sum_{x\in\mathbb{O}_{2L}}\sum_{w\in\mathbb{E}_{L-3}/0} \bm{q}(x,w)_{x,...,x+w-1} \in \textrm{Span}\left(\mathcal{B}_{oe}\right),
    \end{equation}
    We know, by the diagrams given in Fig.~\ref{spreadingdiagrams}, that under the action of a dual-unitary Floquet operator, $\mathbb{U}$, we can have 3 types of terms produced from these 2 components: we can have terms in the odd-$w$ subspaces which remain in these subspaces - we will denote the new terms with $\mathcal{Q}^{\prime}_{\textrm{odd-}w}$; we can have some terms in $\mathcal{B}_{oe}$ which terms from the odd-$w$ subspaces are mapped into - we will denote these terms with $\mathcal{Q}^{\prime}_{\textrm{odd-}w \mapsto \textrm{odd-even}}$; finally, we will have terms in $\mathcal{B}_{oe}$ which the original $\mathcal{B}_{oe}$ terms (in $\mathcal{Q}_{\textrm{odd-even}}$) are mapped into - we will denote these terms with $\mathcal{Q}^{\prime}_{\textrm{odd-even}}$. Together, this gives
    \begin{equation}
        \mathbb{U} \mathcal{Q} \mathbb{U}^{\dagger} = \mathcal{Q}^{\prime}_{\textrm{odd-}w} + \mathcal{Q}^{\prime}_{\textrm{odd-}w \mapsto \textrm{odd-even}} + \mathcal{Q}^{\prime}_{\textrm{odd-even}}.
    \end{equation} 
    If $\mathcal{Q}$ is to be a conserved quantity, we require that 
    \begin{equation}\label{eqn:lem3firstQcondition}
        \mathcal{Q}^{\prime}_{\textrm{odd-}w} = \mathcal{Q}_{\textrm{odd-}w},
    \end{equation}
     and
    \begin{equation}\label{eqn:lem3secondQcondition}
        \mathcal{Q}^{\prime}_{\textrm{odd-even}} + \mathcal{Q}^{\prime}_{\textrm{odd-}w \mapsto \textrm{odd-even}} = \mathcal{Q}_{\textrm{odd-even}},
    \end{equation}
    in order to ensure that $\mathcal{Q} = \mathbb{U}\mathcal{Q}\mathbb{U}^{\dagger}$. If $\mathcal{Q}_{\textrm{odd-even}} \neq 0$, then by Sublemma \ref{lem:notjustoddeven} we know that $\mathcal{Q}^{\prime}_{\textrm{odd-even}} \neq \mathcal{Q}_{\textrm{odd-even}}$, and so we require that $\mathcal{Q}^{\prime}_{\textrm{odd-}w \mapsto \textrm{odd-even}} = \mathcal{Q}^{\prime}_{\textrm{odd-even}} - \mathcal{Q}_{\textrm{odd-even}}$ is non-zero such that Eqn.~\ref{eqn:lem3secondQcondition} holds. However, if $\mathcal{Q}^{\prime}_{\textrm{odd-}w \mapsto \textrm{odd-even}} \neq 0$ is non-zero, and $\mathcal{Q}^{\prime}_{\textrm{odd-}w} = \mathcal{Q}_{\textrm{odd-}w}$, then we will have that the norm of the terms produced by acting on $\mathcal{Q}_{\textrm{odd-}w}$ is greater than the norm of $\mathcal{Q}_{\textrm{odd-}w}$ itself; this contradicts the assumption that $\mathbb{U}$ is unitary - unitary transformations are, by definition, norm-preserving. The only way to resolve this is if $\mathcal{Q}_{\textrm{odd-even}} = 0$, in which case
    \begin{equation}
        \mathcal{Q} = \mathcal{Q}_{\textrm{odd-}w} = \sum_{x\in\mathbb{Z}_{2L}}\sum_{w\in\mathbb{O}_{L-3}} \bm{q}(x,w)_{x,...,x+w-1},
    \end{equation}
    and hence
    \begin{equation}
        \bm{q}(x,w) = 0 \;  \forall \; w \in \mathbb{E}_{L-3}/0.
    \end{equation}

\end{proof}

    \section{The unitary subspaces of the \texorpdfstring{$\mathcal{M}_{\pm,w}$}{TEXT} maps}\label{appendix:solitonbasis}

    In the proof of Theorem \ref{thm:fundamentalcharges}, we use the fact that any operator $\bm{a} \in \bar{\mathcal{A}}_w$ supported strictly non-trivially over an interval $[x, x+w-1]$ which is mapped to a new operator supported strictly non-trivially over an interval $[x\pm 2t, x+w-1\pm2t]$ by $t$ repeated applications of the Floquet operator can be decomposed into a basis of solitons; here, we rigorously establish this fact. Specifically, we show that the subspace spanned by the solitons of the $\mathcal{M}_{\pm, w}$ maps is the same as the subspace spanned by operators which remain the same width and lead to a preservation of norm under repeated applications of the $\mathcal{M}_{\pm, w}$ maps (which is an equivalent statement to that of the previous sentence).
    \par
    \begin{lemma}\label{lem:solitonbasis}
        Take a dual-unitary Floquet operator, $\mathbb{U}$, as defined in Eqn.~\ref{eqn:floquetoperator}. Consider the following subspaces of $\bar{\mathcal{A}}_w$, defined as
        \begin{equation}
            \mathcal{S}_{+,w} = \textrm{Span}\left(
                \bm{a} \in \bar{\mathcal{A}}_{w} 
                \; : \; \mathcal{M}_{+,w}^t(\bm{a}) \in \bar{\mathcal{A}}_w, \; \|\mathcal{M}_{+,w}^t(\bm{a})\| = \|\bm{a}\|, \; \forall \; t \in \mathbb{N}
                \right),
        \end{equation}
        and
        \begin{equation}
            \mathcal{S}_{-,w} = \textrm{Span}\Bigl(
                \bigl\{\bm{b} \in \bar{\mathcal{A}}_{w} 
                \; : \; \mathcal{M}_{-,w}^t(\bm{b}) \in \bar{\mathcal{A}}_w, \; \|\mathcal{M}_{-,w}^t(\bm{b})\| = \|\bm{b}\|, \; \forall \; t \in \mathbb{N}\bigr\}
                \Bigr),
        \end{equation}
        where $\|A\| = \sqrt{\tr(A^{\dagger}A)}$ is the Hilbert-Schmidt norm. The right-moving width-$w$ solitons, defined as the elements of the set $\mathbb{S}_{+,w}$ in Eqn.~\ref{eqn:set_S}, form a complete basis for $\mathcal{S}_{+,w}$, and the left-moving width-$w$ solitons, defined as the elements of the set $\mathbb{S}_{+,w}$ in Eqn.~\ref{eqn:set_T}, form a complete basis for $\mathcal{S}_{-,w}$. That is,
        \begin{equation}
            \textrm{Span}\left(\mathbb{S}_{+,w} \right) = \mathcal{S}_{+,w}, \quad \textrm{Span}\left(\mathbb{S}_{-,w} \right) = \mathcal{S}_{-,w}.
        \end{equation}
    
    \end{lemma}
    \begin{proof}
        We start by considering the $\mathcal{S}_{+,w}$ subspace, and define a map
        \begin{equation}
            f: \mathcal{S}_{+,w} \mapsto \mathcal{S}_{+,w},
        \end{equation}
        which is the restriction of $\mathcal{M}_{+,w}$ to $\mathcal{S}_{+,w}$, i.e.~
        \begin{equation}
            f(\bm{a}) = \mathcal{M}_{+,w}(\bm{a}).
        \end{equation}
        Trivially, $f$ is linear, and also preserves the Hilbert-Schmidt inner product
        \begin{align}
            \langle \bm{a}, \bm{c}\rangle
            &= \tr \left(
                \bm{a}^{\dagger}\bm{c}
            \right) \\
            &= \frac{1}{d^{2L-w}} \tr \left(\mathbb{U}^{\dagger}\mathbb{U}
                \bm{a}^{\dagger}_{\mathcal{R}}\bm{c}_{\mathcal{R}}
            \right), \quad |\mathcal{R}| = w \\
            &= \frac{1}{d^{2L-w}} \tr \left(\mathbb{U}
            \bm{a}^{\dagger}_{\mathcal{R}}\mathbb{U}^{\dagger}\mathbb{U}\bm{c}_{\mathcal{R}}\mathbb{U}^{\dagger}
        \right) \\
            &= \frac{1}{d^{2L-w}} \tr \left(
            \left[\mathcal{M}_{+,w}(\bm{a}^{\dagger})\right]_{\mathcal{R}+2}\left[\mathcal{M}_{+,w}(\bm{c})\right]_{\mathcal{R}+2}
        \right) \\
            &= \tr \left(\mathcal{M}_{+,w}(\bm{a})^{\dagger}\mathcal{M}_{+,w}(\bm{c})
            \right) \\
            &= \langle f(\bm{a}), f(\bm{c})\rangle,
        \end{align}
        where we have defined an interval $\mathcal{R} = [x, x+w-1]$ such that $x$ is even, and noted that $\mathcal{M}_{\pm,w}(\bm{a})^{\dagger} = \mathcal{M}_{\pm,w}(\bm{a}^{\dagger})$ and so $f(\bm{a}^{\dagger}) = f(\bm{a})^{\dagger}$. We also note that, by definition, $f$ is an endomorphism, $f \in \textrm{End}(\mathcal{S}_{+,w})$. This allows us to deduce that $f$ is a unitary transformation, and hence its eigenvectors must form a complete orthonormal basis for $\mathcal{S}_{+,w}$. For these eigenvectors, which we will conveniently denote as $\{\bm{a_k}\}_{k=1}^n$, with $1 \leq n = (\textrm{dim}(\mathcal{S}_{+,w}) - 1) \leq d^{2^w}-1$, 
        it will be true that 
        \begin{equation}\label{eqn:unimod_eigs_of_fmap}
            \mathbb{U}([\bm{a_k}]_{x,...,x+w-1})\mathbb{U}^{\dagger} = \lambda_{k,w}[\bm{a_k}]_{x+2,...,x+w+1}, \quad \lambda_{k,w} \in \mathbb{C}, \; |\lambda_{k,w}| = 1,
        \end{equation}
        which makes them all right-moving, width-$w$ solitons - the elements of the set $\mathbb{S}_{+,w}$ defined in Eqn.~\ref{eqn:set_S}; these eigenvectors \textit{are} the solitons, and we have shown that they span a \textit{unitary subspace} $\mathcal{S}_{+,w}$ of the $\mathcal{M}_{+,w}$ map. The proof that the left-moving width-$w$ solitons - the elements of the set $\mathbb{S}_{-,w}$ - span a unitary subspace of the $\mathcal{M}_{-,w}$ map follows completely analogously. We conclude by noting the resemblance this illuminates between the $\mathcal{S}_{\pm,w}$ subspaces and the well-known notion of decoherence-free subspaces \cite{lidar1998decoherencefreesubspsforqc,lidar2014reviewofdecohfreesubsps} (i.e.~they are subspaces of $\mathcal{H}_w$ on which the generically non-unitary $\mathcal{M}_{\pm, w}$ maps act unitarily).
    \end{proof}

    \section{Non-maximal velocity solitons}\label{sec:v1solitons}
    In Section \ref{sec:solitonsandlocalconservedQs}, we made a brief comment about our definition of solitons being somewhat restrictive, in that we only consider maximal velocity solitons. Specifically, we are referring to the fact that operators $\bm{a} \in \mathcal{A}_w$ which satisfy
    \begin{equation}\label{eqn:arbitraryvsolitons}
        \mathbb{U}\bm{a}_{x,...,x+w-1}\mathbb{U}^{\dagger} = \lambda \bm{a}_{x+v,...,x+w-1+v}, \; \lambda \in \mathbb{C}, \; x \in \mathbb{Z}_{2L},
    \end{equation}
    with $v \in \{0, \pm1, \pm2\}$ (because of the brickwork geometry of the circuit these are the only possible values of $v$) have been referred to as solitons in the literature (see Ref. \cite{bertini2020opent2}), whereas we only allow (in Definition \ref{def:wbodysolitons}) for $v= \pm 2$.
    \par
    In Ref. \cite{bertini2020opent2}, the authors proved that when $w=1$ and $U = V$ in the definition of the Floquet operator (i.e. the even and odd layers are generated by the same 2-qudit unitary, see Eqn. \ref{eqn:floquetoperator}), then $v = 0$ and $v = \pm 2$ are the only possible values which lead to consistent solutions. Moreover, when $\mathbb{U}$ is dual-unitary, they showed that $v = 0$ solitons cannot exist. It is not immediately clear, however, that these results generalise to the case of $U \neq V$ (i.e. temporal inhomogeneity in the Floquet operator) and many-body operators (i.e. $w > 1$) as considered here.
    \par
    In this Appendix we explain why we incorporate only the $v=\pm 2$ solutions into our definition of solitons used in the main body of the paper. We show that allowing for $U \neq V$ and $w > 1$ does not allow for $v=0$ solutions to Eqn. \ref{eqn:arbitraryvsolitons}, and that while it becomes possible to find $v = \pm 1$ solutions (of which we provide an explicit example) these are not stable under repeated applications of the Floquet operator, and so should not be considered to be solitons (solitons are long-lived excitations) despite the fact they technically satisfy Eqn. \ref{eqn:arbitraryvsolitons}. We note these results are already implied by Theorem \ref{thm:fundamentalcharges} (if we have an operator that satisifies Eqn. \ref{eqn:arbitraryvsolitons} repeatedly, under repeated conjugation by $\mathbb{U}$, then we can construct a conserved quantity out it; any such operator with $v \neq \pm 2$ would violate Theorem \ref{thm:fundamentalcharges}), but spell them out more explicitly here for illustrative purposes.
    \par
    Consider a dual-unitary Floquet operator $\mathbb{U}$ on a chain of $2L$ qubits, defined as in Eqn.~\ref{eqn:floquetoperator} with $U = \textrm{FSWAP}(H\otimes \mathds{1})$ and $V = \textrm{FSWAP}(\mathds{1}\otimes H)$, where $H$ is a single-qubit Hadamard gate and FSWAP is the fermionic SWAP operator studied in Section $\ref{sec:fermionic}$. It is straightforward to verify that both $U$ and $V$ are dual-unitary.
    \par
    Let $\mathbb{U}$ act on a pair of neighbouring $\bm{\sigma_z}$ operators, $\left[\bm{\sigma_z} \otimes \bm{\sigma_z}\right]_{x, x+1}$, where $x$ is even. We only need to consider the gates within a causal light cone; this amounts to calculating
    \begin{equation}
        \mathbb{U}\left[\bm{\sigma_z} \otimes \bm{\sigma_z}\right]_{x, x+1}\mathbb{U}^{\dagger} = \left[(V \otimes V)\left[U(\bm{\sigma_z} \otimes \bm{\sigma_z})U^{\dagger}\right]_{x, x+1} (V \otimes V)^{\dagger}\right]_{x-1,...,x+2}.
    \end{equation}
    The action of the single Hadamard in $U$ maps $\bm{\sigma_z} \otimes \bm{\sigma_z}$ to $\bm{\sigma_x} \otimes \bm{\sigma_z}$. From the relations for FSWAP given in Section \ref{sec:fermionic}, we know that this is then mapped to $\mathds{1} \otimes \bm{\sigma_x}$, i.e. $\left[U(\bm{\sigma_z} \otimes \bm{\sigma_z})U^{\dagger}\right]_{x, x+1} = \left[\bm{\sigma_x}\right]_{x+1}$. This means that we then only need to calculate the action of the right-hand $V$ gate, i.e.
    \begin{equation}
        \mathbb{U}\left[\bm{\sigma_z} \otimes \bm{\sigma_z}\right]_{x, x+1}\mathbb{U}^{\dagger} = \left[V(\bm{\sigma_x} \otimes \mathds{1})V^{\dagger}\right]_{x+1, x+2}.
    \end{equation}
    The FSWAP in $V$ maps $\bm{\sigma_x} \otimes \mathds{1}$ to $\bm{\sigma_z} \otimes \bm{\sigma_x}$, and then the Hadamard on the $x+2$ qubit maps this to $\bm{\sigma_z} \otimes \bm{\sigma_z}$, so overall we have
    \begin{equation}\label{eqn:ZZv1soliton}
        \mathbb{U}\left[\bm{\sigma_z} \otimes \bm{\sigma_z}\right]_{x, x+1}\mathbb{U}^{\dagger} = \left[\bm{\sigma_z} \otimes \bm{\sigma_z}\right]_{x+1, x+2}.
    \end{equation}
    This shows that we can have many-body operators which are spatially translated by just one site - which we might refer to as $v = \pm 1$ solitons - by a brickwork dual-unitary Floquet operator. However, note that we stated that $x$ must be even, and so the mapping defined in Eqn.~\ref{eqn:ZZv1soliton} is from the $\mathcal{B}_{ee}$ subspace to the $\mathcal{B}_{oe}$ subspace. In Appendix \ref{appendix:fate}, we showed that any operators in the $\mathcal{B}_{oe}$ subspace cannot evolve solitonically - their region of spatial support grows ballistically at the maximal velocity in both spatial directions. So, the $\bm{\sigma_z} \otimes \bm{\sigma_z}$, $v=1$ solution to Eqn. \ref{eqn:arbitraryvsolitons} presented here could only evolve solitonically (i.e. simply be spatially translated) for one period of dual-unitary Floquet evolution - after this one time step, it ends up in the $\mathcal{B}_{oe}$ subspace and is no longer able to evolve solitonically (another way to look at this is to realise that operators in $\mathcal{B}_{oe}$ cannot be solutions to Eqn. \ref{eqn:arbitraryvsolitons} for a dual-unitary $\mathbb{U}$ - this follows from the digraph in Fig. \ref{fig:digraph}).
    \par
    Furthermore, we know by the mappings given in Figs.~\ref{fig:digraph} and \ref{spreadingdiagrams} that this must be true of all $v = \pm 1$ solitons in a dual-unitary circuit - the mapping between the $\mathcal{B}_{ee}$ subspace and the $\mathcal{B}_{oe}$ subspace after conjugation by a dual-unitary $\mathbb{U}$ is the only such mapping under which the site of left-most support changes by one (giving $v = \pm 1$) and the width $w$ stays constant. Hence, all $v = \pm 1$ solitons must suffer the same fate as the $\bm{\sigma_z} \otimes \bm{\sigma_z}$ example considered above. As these $v = \pm 1$ solitons must belong to an even-width subspace ($\mathcal{B}_{ee}$) we also know by Lemma \ref{lem:noeveneven} that they cannot contribute to any conserved quantities.
    \par
    We conclude by noting, from the digraph in Fig. \ref{fig:digraph}, that the self-loops of the odd-width subspaces ($\mathcal{B}_{oo}$ and $\mathcal{B}_{eo}$) are the only other possible mappings (generated by the adjoint action of a dual-unitary Floquet operator) under which a local operator can remain the same width, and hence the only other mappings that can provide potential solutions to Eqn. \ref{eqn:arbitraryvsolitons}. These are the $v = \pm 2$ soliton solutions as discussed throughout the main body of the paper. This leaves no mappings (for a dual-unitary $\mathbb{U}$) which can provide $v = 0$ solutions to Eqn. \ref{eqn:arbitraryvsolitons}. Hence, the result of Ref. \cite{bertini2020opent2} that dual-unitary circuits can only support maximal velocity 1-body solitons generalises to the many-body case considered here (with the caveat that when $U \neq V$ in the definition of the Floquet operator, it is possible to find $v = \pm 1$ solutions to Eqn. \ref{eqn:arbitraryvsolitons} which are stable for only one period of Floquet evolution).

    \section{Quantifying the effective reduction on Hilbert space imposed by the dual-unitary constraints}\label{sec:dimreductioncalc}
    In Section \ref{sec:dynamicalconstraints}, we discussed the dynamical constraints imposed by dual-unitarity on the spatial support of local operators as they are time-evolved under a brickwork dual-unitary circuit. As an example, we noted that of the $d^8-1$ coefficients generically needed to describe (using some orthonormal operator basis) a one-site local operator $\bm{q}$ evolved to have non-trivial support on 4-sites under one layer of a brickwork unitary Floquet operator, $\mathbb{U}\bm{q}_x\mathbb{U}^{\dagger}$, precisely $4d^6 - 2d^4 + d^2 - 1$ of those coefficients are strictly zero if $\mathbb{U}$ is dual-unitary. Here, we provide further clarification on how this can be deduced from Eqn.~\ref{eqn:1bodyunderfloquetDU}.
    \par
    We start by noting that Eqn.~\ref{eqn:1bodyunderfloquetDU} tells us that $\mathbb{U}\bm{q}_x\mathbb{U}^{\dagger}$ can only have overlap with single-site operators on the sites $x+2$, corresponding to non-zero coefficents of the form $c_{000\alpha}$, $\alpha \neq 0$. This also immediately follows from the main result of Ref.~\cite{bertini2019exactcorrfuncs}. Consequently, any of the coefficients $c_{\alpha000}$, $c_{0\alpha00}$, $c_{00\alpha0}$ with $\alpha \neq 0$, that correspond to single-site operators on the other 3 sites must be zero. There are precisely $d^2-1$ traceless elements (i.e. $\alpha \neq 0$) of the operator basis for the single-site Hilbert space, so this gives us $3(d^2-1)$ coefficients that must be zero.
    \par
    We also know that the only overlap $\mathbb{U}\bm{q}_x\mathbb{U}^{\dagger}$ can have with operators that are traceless on two contiguous sites is with those supported on the sites $x+1$ and $x+2$, corresponding to non-zero coefficents of the form $c_{00\alpha\beta}$, $\alpha, \beta \neq 0$. This means that all coefficients of the form $c_{\alpha\beta00}$ and $c_{0\alpha\beta0}$ with $\alpha, \beta \neq 0$ must be zero. This gives us $2(d^2-1)^2$ coefficients that must be zero.
    \par
    Finally, we know that $\mathbb{U}\bm{q}_x\mathbb{U}^{\dagger}$ cannot have any overlap with operators that are supported strictly on three contiguous sites, and so all coefficients of the form $c_{\alpha\beta\gamma0}$ or $c_{0\alpha\beta\gamma}$ with $\alpha, \gamma \neq 0$ must be zero. Here, as the site in the middle can be traceless or identity, we get a $d^2$ contribution from the index $\beta$, which together with the $d^2-1$ contributions from the indices $\alpha$ and $\gamma$ gives us $2d^2(d^2-1)^2$ coefficients that must be zero in total.
    \par
    Putting this all together, this gives us $3(d^2-1) + 2(d^2-1)^2 + 2d^2(d^2-1)^2 = 2d^6 - 2d^4 + d^2 - 1$ coefficients that must be zero in the description of $\mathbb{U}\bm{q}_x\mathbb{U}^{\dagger}$, as claimed in Section \ref{sec:dynamicalconstraints}.
\end{document}